\documentclass[a4paper,UKenglish]{lipics}

\usepackage{amsfonts}
\usepackage{stmaryrd}
\usepackage{verbatim}

\usepackage{todonotes}
\usepackage{etex}
\usepackage{booktabs}
\usepackage{tabularx}
\usepackage{multirow}
\usepackage{pgf}
\usepackage{pgfplots}
\usepackage{tikz}
\usetikzlibrary{automata,arrows,intersections,positioning,calc,shapes}
\tikzset{astate/.style={draw,circle,inner sep=0pt,minimum size=8mm}}
\tikzset{astateq/.style={draw,rectangle,inner sep=3pt,rounded corners}}
\tikzset{adistr/.style={draw,circle,fill,minimum size=1mm,inner sep=0mm}}
\usepackage{flowchart}

\usepackage{xspace}
\usepackage{wrapfig}
\usepackage{cite}
\usepackage{paralist}
\usepackage[noend]{algorithmic}
\usepackage{algorithm}

\usepackage{dcolumn}
\newcolumntype{d}{D{.}{.}{4.0}}

\renewcommand{\comment}[1]{ }

\newcommand{\proc}[1]{\textnormal{\scshape#1}}
\newcommand{\progHeader}[1]{#1}
\newcommand{\ComputeAcceptingScc}{\proc{ComputeAccScc}}
\newcommand{\BuildBreakpoint}{\proc{BuildBP}}
\newcommand{\BuildSubset}{\proc{BuildSubset}}
\newcommand{\BuildProduct}{\proc{BuildProd}}
\newcommand{\ComputeScc}{\proc{ComputeScc}}
\newcommand{\ComputeMec}{\proc{ComputeMec}}
\newcommand{\IsAccepting}{\proc{IsAcc}}

\newcommand{\functionGenericArgument}{\,\cdot\,}

\newcommand{\calE}{\mathcal{E}}

\newcommand{\calT}{\mathcal{T}}

\newcommand{\sdToDet}{\mathcal{D}}
\newcommand{\bToSub}{\mathcal{S}}
\newcommand{\bToBp}{\mathcal{BP}}
\newcommand{\bToSd}{\mathcal{SD}}

\newcommand{\labelFunc}{L}

\newcommand{\pred}{\mathsf{pred}}

\newcommand{\transtyle}[1]{\mathrm{#1}}

\newcommand{\aut}{\mathcal{A}}
\newcommand{\paut}{\mathcal{P}}
\newcommand{\baut}{\mathcal{B}}
\newcommand{\astates}{Q}
\newcommand{\ainit}{I}

\newcommand{\amat}{\transtyle{T}}

\newcommand{\afinal}{\transtyle{F}}
\newcommand{\aacc}{\transtyle{A}}
\newcommand{\arej}{\transtyle{R}}

\newcommand{\abuechiset}[1][k]{\boldsymbol{\afinal}_{#1}}
\newcommand{\arabinsetacc}[1][k]{\boldsymbol{\aacc}_{#1}}
\newcommand{\arabinsetrej}[1][k]{\boldsymbol{\arej}_{#1}}
\newcommand{\arabinset}[1][k]{(\arabinsetacc[#1], \arabinsetrej[#1])}
\newcommand{\ACC}{\boldsymbol{\transtyle{ACC}}}
\newcommand{\aparityfun}{\pri \colon \amat \to \natIntK}

\newcommand{\autElements}{(\Sigma, \astates, \ainit, \amat, \arabinset)}
\newcommand{\bautElements}{(\Sigma, \astates, \ainit, \amat, \abuechiset)}
\newcommand{\pautElements}{(\Sigma, \astates, \ainit, \amat, \pri)}

\newcommand{\arun}{\rho}

\newcommand{\tran}[1]{\transtyle{tr}(#1)}

\newcommand{\bpsubscript}{\mathit{bp}}
\newcommand{\rbpaut}{\mathcal{BP}}
\newcommand{\bbpaut}{\mathcal{BP}}
\newcommand{\bpstates}{Q_{\bpsubscript}}
\newcommand{\bpinit}{q_{\bpsubscript}}
\newcommand{\bpmat}{\transtyle{T}_{\bpsubscript}}
\newcommand{\bpfinal}{\transtyle{F}_{\bpsubscript}}

\newcommand{\sssubscript}{\mathit{ss}}
\newcommand{\ssaut}{\mathcal{S}}
\newcommand{\ssstates}{Q_{\sssubscript}}
\newcommand{\ssinit}{q_{\sssubscript}}
\newcommand{\ssmat}{\transtyle{T}_{\sssubscript}}
\newcommand{\ssfinal}{\transtyle{F}_{\sssubscript}}

\newcommand{\daut}{\mathcal{D}}
\newcommand{\dstates}{Q_{d}}
\newcommand{\dinit}{q_{0}}
\newcommand{\dmat}{\transtyle{T}_{d}}

\newcommand{\sdped}{\mathit{sd}}
\newcommand{\sdaut}{\mathcal{SD}}
\newcommand{\sdstates}[1][\sdped]{\astates_{#1}}
\newcommand{\sdstatesini}{\sdstates[i]}
\newcommand{\sdstatesfin}{\sdstates[f]}
\newcommand{\sdinit}{q_{\sdped}}
\newcommand{\sdmat}[1][\sdped]{\transtyle{T}_{#1}}
\newcommand{\sdmatini}{\sdmat[i]}
\newcommand{\sdmattra}{\sdmat[t]}
\newcommand{\sdmatfin}{\sdmat[f]}
\newcommand{\sdfinal}{\transtyle{F}_{\sdped}}

\newcommand{\mc}{\mathcal{M}}
\newcommand{\mdp}{\mathcal{M}}
\newcommand{\pstates}{M}
\newcommand{\pinit}{\mu_{0}}
\newcommand{\pmat}{\transtyle{P}}
\newcommand{\ppath}{\xi}
\newcommand{\ppaths}{\mathit{Paths}}

\newcommand{\word}{\alpha}

\newcommand{\lang}{\mathcal{L}}

\newcommand{\insym}{\mathop{\mathrm{Inf}}}

\newcommand{\cyl}{\mathit{Cyl}}

\renewcommand{\phi}{\varphi}
\renewcommand{\epsilon}{\varepsilon}

\newcommand{\scc}{\mathtt{S}}
\newcommand{\scctran}{\pmat_{\scc}}
\newcommand{\mec}{\mathtt{C}}
\newcommand{\mectran}{\pmat_{\mec}}
\newcommand{\mecprimetran}{\pmat_{\mec'}}

\newcommand{\proj}[1]{\pi_{#1}}
\newcommand{\projaut}{\proj{\aut}}

\newcommand{\prodMC}{\times}
\newcommand{\prodMDP}{\times}
\newcommand{\prodAut}{\otimes}

\newcommand{\set}{\mathsf{set}}
\newcommand{\ind}{\mathsf{index}}
\newcommand{\children}{\mathsf{children}}

\newcommand{\Aset}{\mathit{Act}}
\newcommand{\en}{\mathit{en}}
\newcommand{\dist}{\mathit{Dist}}

\newcommand{\blank}{\mbox{\textvisiblespace}}
\newcommand{\bp}{\mathsf{bp}}
\newcommand{\buchi}{B\"{u}chi\xspace}

\newcommand{\host}{\mathsf{host}}

\newcommand{\pri}{\mathsf{pri}}
\newcommand{\os}{\mathsf{os}}
\newcommand{\q}{\overline{q}}
\newcommand{\reached}{\mathsf{rchd}}
\newcommand{\rename}{\mathsf{rename}}
\newcommand{\width}{\mathsf{width}}

\newcommand{\iscasmc}{\textsc{IscasMC}\xspace}
\newcommand{\prism}{\textsc{PRISM}\xspace}
\newcommand{\liquor}{\textsc{LiQuor}\xspace}

\newcommand{\ltlthreeba}{\textsc{LTL3BA}\xspace}
\newcommand{\spot}{\textsc{SPOT}\xspace}
\newcommand{\rabinizerthree}{\textsc{Rabinizer~3}\xspace}

\newcommand{\naturals}{\mathbb{N}}
\newcommand{\booleans}{\mathbb{B}}

\newcommand{\setcond}[2]{\{\, #1 \mid #2 \,\}}
\newcommand{\setnocond}[1]{\{#1\}}

\newcommand{\bigO}{\mathcal{O}}

\newcommand{\natIntK}[1][k]{[1..#1]}

\newcommand{\sched}{\upsilon}

\newcommand{\prob}{\mathfrak{P}}

\newcommand{\defeq}{\overset{\mathrm{def}}{=}}

\newcommand{\runs}{\mathop{\mathrm{Run}}}

\newcommand{\successor}[1]{\widehat{#1}}

\newcommand{\ltlP}{\mathbb{P}_{\!=?}}
\newcommand{\ltlPmin}{\mathbb{P}_{\!\mathrm{min}=?}}
\newcommand{\ltlPmax}{\mathbb{P}_{\!\mathrm{max}=?}}
\newcommand{\ltlU}{\mathbin{\mathbf{U}}}

\newcommand{\ltlF}{\mathbin{\mathbf{F}}}
\newcommand{\ltlG}{\mathbin{\mathbf{G}}}
\newcommand{\ltlGF}{\mathbin{\mathbf{GF}}}
\newcommand{\ltlFG}{\mathbin{\mathbf{FG}}}

\theoremstyle{plain}
\newtheorem{proposition}[theorem]{Proposition}

\newcommand{\neueZeile}{\newline}
\renewcommand{\neueZeile}{\\}

\begin{document}

\title{Lazy Probabilistic Model Checking without Determinisation}

\date{}
\author[1]{Ernst Moritz Hahn}
\author[1]{Guangyuan Li}
\author[2]{Sven Schewe}
\author[1]{Andrea Turrini}
\author[1]{Lijun Zhang}
\affil[1]{State Key Laboratory of Computer Science, Institute of Software, CAS}
\affil[2]{University of Liverpool}

\authorrunning{E.\,M. Hahn, G. Li, S. Schewe, A. Turrini, and L. Zhang}
\Copyright{Ernst Moritz Hahn, Guangyuan Li, Sven Schewe, Andrea Turrini, and Lijun Zhang}


\maketitle
\begin{abstract}
The bottleneck in the quantitative analysis of Markov chains and Markov decision processes against specifications given in LTL or as some form of nondeterministic \buchi automata is the inclusion of a determinisation step of the automaton under consideration.
In this paper, we show that full determinisation can be avoided:
subset and breakpoint constructions suffice.
We have implemented our approach---both explicit and symbolic versions---in a prototype tool.
Our experiments show that our prototype can compete with mature tools like \prism.
\end{abstract}

\section{Introduction}
Markov chains (MCs) and Markov decision processes (MDPs) are widely used to study systems that exhibit 
both, probabilistic and nondeterministic choices.
Properties of these systems are often specified by temporal logic formulas, such as the branching time logic PCTL~\cite{HaJo94}, the linear time logic PLTL~\cite{BiancoA95}, or their combination PCTL*~\cite{BiancoA95}.
While model checking is tractable for PCTL~\cite{BiancoA95}, it is more expensive for PLTL:
PSPACE-complete for Markov chains and 2EXPTIME-complete for MDPs~\cite{Courcoubetis+Yannakakis/95/Markov}.

In classical model checking, one checks whether a model $\mc$ satisfies an LTL formula $\phi$ by first constructing a nondeterministic \buchi automaton $\baut_{\neg\phi}$~\cite{DBLP:conf/lics/VardiW86}, which recognises the models of its negation $\neg\phi$.
The model checking problem then reduces to an emptiness test for the product $\mc \prodAut \baut_{\neg\phi}$.
The translation to \buchi automata may result in an exponential blow-up compared to the length of $\phi$.
However, this translation is mostly very efficient in practice, and highly optimised off-the-shelf tools like \ltlthreeba~\cite{BabiakKRS12} or \spot~\cite{DuretLutz11} are available.

The quantitative analysis of a probabilistic model $\mc$ against an LTL specification $\phi$ is more involved.
To compute the maximal probability $\prob^{\mc}(\phi)$ that $\phi$ is satisfied in $\mc$, the classic automata-based approach includes the determinisation of an intermediate \buchi automaton $\baut_{\phi}$.
If such a deterministic automaton $\aut$ is constructed for $\baut_{\phi}$, then determining the probability $\prob^{\mc}(\phi)$ reduces to solving an equation system for Markov chains, and a linear programming problem for MDPs~\cite{BiancoA95}, both in the product $\mc \prodAut \aut$.
Such a determinisation step usually exploits a variant of Safra's~\cite{Safra/88/Safra} determinisation construction, such as the techniques presented in~\cite{Piterman/07/Parity,Schewe/09/determinise}.

Kupferman, Piterman, and Vardi point out in \cite{KPV/06/Safraless} that ``Safra’s determinization construction has been notoriously resistant to efficient implementations.''
Even though analysing long LTL formulas would surely be useful as they allow for the description of more complex requirements on a system's behaviour, model checkers that employ determinisation to support LTL, such as \liquor~\cite{DBLP:conf/qest/CiesinskiB06} or \prism~\cite{KwiatkowskaNP11}, might fail to verify such properties.

In this paper we argue that applying the Safra determinisation step in full generality is only required in some cases, while simpler subset and breakpoint constructions often suffice.
Moreover, where full determinisation is required, it can be replaced by a combination of the simpler constructions, and it suffices to apply it locally on a small share of the places.

A subset construction is known to be sufficient to determinise finite automata, but it fails for \buchi automata.
Our first idea is to construct an under- and an over-approximation starting from the subset construction.
That is, we construct two (deterministic) subset automata $\ssaut^{u}$ and $\ssaut^{o}$ such that $\lang(\ssaut^{u}) \subseteq \lang(\baut_{\phi}) \subseteq \lang(\ssaut^{o})$ where $\lang(\baut_{\phi})$ denotes the language defined by the automaton $\baut_{\phi}$ for $\phi$.
The subset automata $\ssaut^{u}$ and $\ssaut^{o}$ are the same automaton $\ssaut$ except for their accepting conditions.
We build a product Markov chain with the subset automata.
We establish the useful property that the probability $\prob^{\mc}(\phi)$ equals the probability of reaching some \emph{accepting} bottom strongly connected components (SCCs) in this product:
for each bottom SCC $\scc$ in the product, we can first use the accepting conditions in $\ssaut^{u}$ or $\ssaut^{o}$ to determine whether $\scc$ is accepting or rejecting, respectively.
The challenge remains when the test is inconclusive.
In this case, we first refine $\scc$ using a breakpoint construction.
Finally, if the breakpoint construction fails as well, we have two options:
we can either perform a Rabin-based determinisation for the part of the model where it is required, thus avoiding to construct the larger complete Rabin product.
Alternatively, a refined multi-breakpoint construction is used.
An important consequence is that we no longer need to implement a Safra-style determinisation procedure:
subset and breakpoint constructions are enough.
From a theoretical point of view, this reduces the cost of the automata transformations involved from $n^{\bigO(k \cdot n)}$ to $\bigO(k \cdot 3^n)$ for generalised \buchi automata with $n$ states and $k$ accepting sets.
From a practical point of view, the easy symbolic encoding admitted by subset and breakpoint constructions is of equal value.
We discuss that (and how) the framework can be adapted to MDPs---with the same complexity---by analysing the end components \cite{Courcoubetis+Yannakakis/95/Markov,BiancoA95}.

We have implemented our approach---both explicit and symbolic versions---in our \iscasmc tool~\cite{iscasmc}, which we applied
on various Markov chain and MDP case studies.
Our experimental results confirm that our new algorithm outperforms the Rabin-based approach in most of the properties considered.
However, there are some cases in which the Rabin determinisation approach performs better when compared to the multi-breakpoint construction:
the construction of a single Rabin automaton suffices to decide a given connected component, while the breakpoint construction may require several iterations.
Our experiments also show that our prototype can compete with mature tools like \prism.

To keep the presentation clear, the detailed proofs are provided in the appendix.

\section{Preliminaries}
\label{sec:preliminaries}

\subsection{$\omega$-Automata}
Nondeterministic \buchi automata are used to represent $\omega$-regular languages $\lang \subseteq \Sigma^{\omega} = \omega \to \Sigma$ over a finite alphabet $\Sigma$.
In this paper, we use automata with trace-based acceptance mechanisms.
We denote by $\natIntK$ the set $\setnocond{1, 2, \ldots, k}$ and by $j \oplus_{k} 1$ the successor of $j$ in $\natIntK$.
I.e., $j \oplus_{k} 1 = j+1$ if $j < k$ and $j \oplus_{k} 1 = 1$ if $j = k$.

\begin{definition}
\label{def:nonDetBuechiAut}
	A \emph{nondeterministic generalised \buchi automaton} (NGBA) is a quintuple $\baut = \bautElements$, consisting of
	\begin{itemize}
	\item
		a finite alphabet $\Sigma$ of input letters,
	\item
		a finite set $\astates$ of states with a non-empty subset $\ainit \subseteq \astates$ of initial states,
	\item
		a set $\amat \subseteq \astates \times \Sigma \times \astates$ of transitions from states through input letters to successor states, and
	\item
		a family $\abuechiset = \setcond{\afinal_{j} \subseteq \amat}{j \in \natIntK}$ of accepting (final) sets.
	\end{itemize}
\end{definition}

Nondeterministic \buchi automata are interpreted over infinite sequences $\word \colon \omega \to \Sigma$ of input letters.
An infinite sequence $\arun \colon \omega \to \astates$ of states of $\baut$ is called a \emph{run} of $\baut$ on an input word $\word$ if $\arun(0) \in \ainit$ and, for each $i \in \omega$, $\big(\arun(i), \word(i), \arun(i+1)\big) \in \amat$.
We denote by $\runs(\word)$ the set of all runs $\arun$ on $\word$.
For a run $\arun \in \runs(\word)$, we denote with $\tran{\arun} \colon i \mapsto \big(\arun(i), \word(i), \arun(i+1)\big)$ the \emph{transitions of $\arun$}.
We sometimes denote a run $\arun$ by the associated states, that is, $\arun = q_{0} \cdot q_{1} \cdot q_{2} \cdot \ldots$ where $\arun(i) = q_{i}$ for each $i \in \omega$ and we call a finite prefix $q_{0} \cdot q_{1} \cdot q_{2} \cdot \ldots \cdot q_{n}$ of $\arun$ a \emph{pre-run}.
A run $\arun$ of a NGBA is \emph{accepting} if its transitions $\tran{\arun}$ contain infinitely many transitions from all final sets, i.e., for each $j \in \natIntK$, $\insym(\tran{\arun}) \cap \afinal_{j} \neq \emptyset$, where $\insym(\tran{\arun}) = \setcond{t \in \amat}{\forall i \in \omega \; \exists j > i \text{ such that } \tran{\arun}(j) = t}$.
A word $\word \colon \omega \to \Sigma$ is \emph{accepted} by $\baut$ if $\baut$ has an accepting run on $\word$, and
the set $\lang(\baut) = \setcond{\word \in \Sigma^\omega}{\text{$\word$ is accepted by $\baut$}}$ of words accepted by $\baut$ is called its \emph{language}.

\begin{wrapfigure}[6]{right}{40mm}
	\centering
	\vskip-4mm
	\begin{tikzpicture}[->,>=stealth,auto]
	\path[use as bounding box] (-1.5,1.3) rectangle (1.5,0.2);
	
	\node (BA) at (0,0) {$\baut_{\calE}$};
	\node (bax) at ($(BA) + (0,1)$) {$x$};
	\node (bay) at ($(BA) + (-1.5,0)$) {$y$};
	\node (baz) at ($(BA) + (1.5,0)$) {$z$};
	
	\draw ($(bax.north) + (0,0.3)$) to node {} (bax.north);
	\draw (bax) to node [right,near end] {$a,1$} (bay);
	\draw (bay) to[bend left=30] node [left] {$c$} (bax);
	\draw (bax) to node [left,near end] {$a$} (baz);
	\draw (baz) to[bend right=30] node [right] {$b,2$} (bax);
	
	\end{tikzpicture}
	\caption{A \buchi automaton}
	\label{fig:exampleBuechi}
\end{wrapfigure}
Figure~\ref{fig:exampleBuechi} shows an example of \buchi automaton.
The number $j$ after the label as in the transition $(x,a,y)$, when present, indicates that the transition belongs to the accepting set $\afinal_{j}$, i.e., $(x,a,y)$ belongs to $\afinal_{1}$.
The language generated by $\baut_{\calE}$ is a subset of $(ab|ac)^{\omega}$ and a word $\word$ is accepted if each $b$ (and $c$) is eventually followed by a $c$ (by a $b$, respectively).

We call the automaton $\baut$ a \emph{nondeterministic \buchi automaton} (NBA) whenever $|\abuechiset| = 1$ and we denote it by $\baut = (\Sigma, \astates, \ainit, \amat, \afinal)$.
For technical convenience we also allow for finite runs $q_{0} \cdot q_{1} \cdot q_{2} \cdot \ldots \cdot q_{n}$ with $\amat \cap \setnocond{q_{n}} \times \setnocond{\word(n)} \times \astates = \emptyset$.
In other words, a run may end with $q_{n}$ if action $\word(n)$ is not enabled from $q_{n}$.
Naturally, no finite run satisfies the accepting condition, thus it is not accepting and has no influence on the language of an automaton.

To simplify the notation, the transition set $\amat$ can also be seen as a function $\amat \colon \astates \times \Sigma \to 2^{\astates}$ assigning to each pair $(q, \sigma) \in \astates \times \Sigma$ the set of successors according to $\amat$, i.e., $\amat(q, \sigma) = \setcond{q' \in \astates}{(q, \sigma, q') \in \amat}$.
We extend $\amat$ to sets of states in the usual way, i.e., by defining $\amat(S, \sigma) = \bigcup_{q \in S} \amat(q, \sigma)$.

\begin{definition}
\label{def:nonDetParityAut}
	A \emph{(transition-labelled) nondeterministic parity automaton} (NPA) with $k$ priorities is a quintuple $\paut = \pautElements$ where $\Sigma$, $\astates$, $\ainit$, and $\amat$ are as in Definition~\ref{def:nonDetBuechiAut} and a function $\aparityfun$ from transitions to a finite set $\natIntK$ of priorities.
\end{definition}

A run $\arun \colon \omega \to \astates$ of a NPA is \emph{accepting} if the lowest priority that occurs infinitely often is even, that is if $\liminf_{n \to \infty} \pri(\tran{\arun}(n))$ is even.

\begin{definition}
\label{def:nonDetRabinAut}
	A \emph{(transition-labelled) nondeterministic Rabin automaton} (NRA) with $k$ accepting pairs is a quintuple $\aut = \autElements$ where $\Sigma$, $\astates$, $\ainit$, and $\amat$ are as in Definition~\ref{def:nonDetBuechiAut} and $\arabinset = \setcond{(\aacc_{i}, \arej_{i})}{i \in \natIntK,\ \aacc_{i}, \arej_{i} \subseteq \amat}$ is a finite family of Rabin pairs.
	(For convenience, we sometimes use other finite sets of indices rather than $\natIntK$.)
\end{definition}

A run $\arun$ of a NRA is accepting if there exists $i \in \natIntK$ such that $\insym(\tran{\arun}) \cap \aacc_{i} \neq \emptyset$ and $\insym(\tran{\arun}) \cap \arej_{i} = \emptyset$.

An automaton $\aut = (\Sigma, \astates, \ainit, \amat, \ACC)$, where $\ACC$ is the acceptance condition (parity, Rabin, \buchi, or generalised \buchi), is called \emph{deterministic} if, for each $(q, \sigma) \in \astates \times \Sigma$, $|\amat(q, \sigma)| \leq 1$, and $\ainit = \setnocond{q_{0}}$ for some $q_{0} \in \astates$.
For notational convenience, we denote a deterministic automaton $\aut$ by the tuple $(\Sigma, \astates, q_{0}, \amat, \ACC)$ and $\amat \colon \astates \times \Sigma \to \astates$ is the partial function, which is defined at $(q, \sigma)$ if, and only if, $\sigma$ is enabled at $q$.
For a given deterministic automaton $\daut$, we denote by $\daut_{d}$ the otherwise similar automaton with initial state $d$.
Similarly, for a NGBA $\baut$, we denote by $\baut_{R}$ the NGBA with $R$ as set of initial states.

\begin{definition}
\label{def:semiDetAut}
	We call a NGBA $\baut = (\Sigma,$ $\astates, \ainit, \amat, \abuechiset)$ a \emph{semi-deterministic \buchi automaton} (SDBA) if the set of states $\astates$ can be partitioned into two sets $\sdstatesini$ and $\sdstatesfin$, where $\sdstatesini$ is the set of states $\baut$ is initially in and $\sdstatesfin$ is the set of finally reached states $\baut$ is eventually always in, such that
	\begin{inparaenum}[1.)]
	\item
		$\ainit \subseteq \sdstatesini$ is singleton, and
	\item
		the set of transitions can be partitioned into three sets:
		transitions $\sdmatini \subseteq \sdstatesini \times \Sigma \times \sdstatesini$ that are taken initially, transit transitions $\sdmattra \subseteq \sdstatesini \times \Sigma \times \sdstatesfin$, and transitions $\sdmatfin \subseteq \sdstatesfin \times \Sigma \times \sdstatesfin$ the automaton takes after a transit transition has been used, such that $\sdmatini$ and $\sdmatfin$ are partial functions, i.e., for each $\sigma \in \Sigma$, $|\sdmatini(q_{i}, \sigma)| \leq 1$ for each $q_{i} \in \sdstatesini$ and $|\sdmatfin(q_{f}, \sigma)| \leq 1$ for each $q_{f} \in \sdstatesfin$.
	\end{inparaenum}
\end{definition}

\subsection{Determinisation of Generalised \buchi Automata}
\label{ssec:determinisation}

NGBAs can be translated to deterministic Rabin automata (DRAs) using the following construction from~\cite{Schewe+Varghese/12/generalisedBuchi}.
We first define the structure that captures the acceptance mechanism of our deterministic Rabin automaton.
\begin{definition}[Ordered trees]
	We call a tree $\calT \subseteq \naturals^{*}$ an \emph{ordered tree}  if it satisfies the following constraints:
	\begin{itemize}
	\item
		Each element $v \in \calT$ is called a \emph{node}.
		The empty sequence $\epsilon\in\calT$ is called the \emph{root}.
	\item
		If a node $v = n_{1} \ldots n_{j} n_{j+1}$ is in $\calT$, then $v' = n_{1} \ldots n_{j}$ is also in $\calT$; $v'$ is called the \emph{predecessor} of $v$, denoted by $\pred(v)$; $\pred(\epsilon)$ is undefined.
	\item
		If a node $v = n_{1} \ldots n_{j-1}n_{j}$ is in $\calT$, then, for each $0 \leq i < n_{j}$, $v' = n_{1} \ldots n_{j-1} i$ is also in $\calT$;
		$v'$ is called an \emph{older sibling} of $v$ and $v$ a \emph{younger sibling} of $v'$; the set of older siblings of $v$ is denoted by $\os(v)$.
	\end{itemize}
\end{definition}

Ordered trees are therefore non-empty sets of nodes that are closed under predecessors and older siblings.

\begin{definition}[Generalised history tree]
	Let $\baut = \bautElements$ be a NGBA.
	A \emph{generalised history tree} (GHT) $d$ over $\astates$ for $k$ accepting sets is a triple $d = (\calT, l, h)$, where
	\begin{itemize}
	\item
		$\calT$ is an ordered tree,
	\item
		$l \colon \calT \rightarrow 2^{\astates} \setminus \setnocond{\emptyset}$ is a labelling function such that
		\begin{itemize}
		\item
			$l(v) \subsetneq l(\pred(v))$ holds for all $v \in \calT$, $v \neq \epsilon$,
		\item
			the intersection of the labels of two siblings is disjoint, i.e., for each $v \in \calT$ and each $v' \in \os(v)$, $l(v) \cap l(v') = \emptyset$, and
		\item
			the union of the labels of all siblings is \emph{strictly} contained in the label of their predecessor, i.e., for each $v \in \calT$ there exists $q \in l(v)$ such that for each $v' \in \calT$, $v = \pred(v')$ implies $q \notin l(v')$,
		\end{itemize}
	\item
		$h \colon \calT \to \natIntK$ is a function that labels every node with a number from $\natIntK$.
	\end{itemize}
\end{definition}

For a GHT $d = (\calT, l, h)$, $(\calT, l)$ is the history tree introduced in~\cite{Schewe/09/determinise}.
GHTs are enriched by the second labelling function, $h$, which is used to relate nodes in the tree with a particular accepting set.
Intuitively, $h(v)$ denotes the \emph{active} index of the accepting transitions, i.e., $\afinal_{h(v)}$.
The construction separates the transition mechanism from the acceptance condition.
A GHT contains the set of currently reached states in its root.
We therefore denote $l(\epsilon)$ for a given GHT $d = (\calT, l, h)$ by $\reached(d)$.
As it helps to understand the correctness of our breakpoint construction, we recall how to determinise a NGBA.

\begin{definition}[Determinisation construction~\cite{Schewe+Varghese/12/generalisedBuchi}]
\label{def:determinisation}
	Given a NGBA $\baut = \bautElements$ with $|\astates|=n$ states and $k$ accepting sets, we construct an equivalent DRA $\daut = (\Sigma, \dstates, \dinit, \dmat, \setcond{(\aacc_{i}, \arej_{i})}{i \in J})$, denoted by $\det(\baut)$, as follows.
	\begin{itemize}
	\item
		$\dstates$ is the set of generalised history trees over $\astates$.
	\item
		$\dinit$ is the generalised history tree $(\setnocond{\epsilon}, l \colon \epsilon \mapsto \ainit, h \colon \epsilon \mapsto 1)$.
	\item
		For each tree $d \in \dstates$ and $\sigma \in \Sigma$ with $\amat(l(\epsilon), \sigma) \neq \emptyset$, the transition $d' = \dmat(d,\sigma)$ is the result of the following sequence of steps:
		for each node $v \in d$, let $\astates_{v}$ be $l(v)$;
		\begin{enumerate}
		\item \emph{Raw update of $l$.}
		\label{item:determinisation:rawupdate}
			For each node $v \in d$, set $l(v) = \amat(\astates_{v}, \sigma)$.

		\item \emph{Sprouting new children.}
		\label{item:determinisation:newchildren}
			For each node $v \in d$ with $c$ children and $h(v) = i$, create a new child $vc$ with $l(vc) = \afinal_{i}(\astates_{v},\sigma)$ and $h(vc) = 1$.

		\item \emph{Stealing of labels.}
		\label{item:determinisation:stealing}
			For each node $v$ and each $q \in \astates_{v}$, $q$ is removed from the labels of all younger siblings of $v$ and all of their descendants.

		\item \emph{Accepting and removing.}
		\label{item:determinisation:accepting}
			For each node $v$ whose label is equal to the union of the labels of its children, remove all descendants of $v$ from the tree, and restrict the domain of $l$ and $h$ accordingly.
			Update $h(v)$ with $h(v) \oplus_{k} 1$.

			The transition is in $\aacc_{v}$ for all nodes $v$ for which this applies.

		\item \emph{Removing nodes.}
		\label{item:determinisation:removing}
			Remove all nodes with empty labels. (The resulting tree $\calT'$ may no longer be ordered.) $l$ and $h$ are updated by restricting their domain.

		\item \emph{Renaming and rejecting.}
		\label{item:determinisation:rejecting}
			To repair the orderedness, we denote by $\|v\| = |\os(v) \cap \calT'|$ the number of (still existing) older siblings of $v$, and map $v = n_{1} \ldots n_{j}$ to $v' = \|n_{1}\|\ \|n_{1} n_{2}\|\ \|n_{1} n_{2} n_{3}\| \ldots \|v\|$, denoted $\rename(v)$.
			We update a triple $(\calT, l, h)$ from the previous step to $\big(\rename(\calT'), l' \colon \rename(v) \mapsto l(v), h' \colon \rename(v) \mapsto h(v)\big)$.

			For each $v \notin \calT \cap \calT'$ and each $v \in \calT'$ with $\rename(v) \neq v$, the transition is in $\arej_{v}$.
		\end{enumerate}
		The sets $\aacc_{i}$ and $\arej_{i}$ are constructed according to above Steps~\ref{item:determinisation:accepting} and~\ref{item:determinisation:rejecting}.
	\item
		$J$ is the set of nodes that occur in some ordered tree of size $n$ (the number of nodes).
	\end{itemize}
\end{definition}

\begin{wrapfigure}[8]{right}{36mm}
	\centering
	\scriptsize
	\vskip-5mm
\resizebox{36mm}{!}{
	\begin{tikzpicture}[->,>=stealth,auto]
		\path[use as bounding box] (-2,3.3) rectangle (2,0);
		
		\node (A) at (0,0) {\normalsize$\aut_{\calE}$};
		
		\node[rectangle, draw, rounded corners, minimum width=15mm, minimum height=6mm] (T0) at ($(A) + (-1.25,2.5)$) {};
		\node (T0R) at (T0) {\normalsize$\epsilon, \setnocond{x}, 1$};
		
		\draw ($(T0) + (0,0.6)$) to node {} ($(T0) + (0,0.3)$);
		
		\node[rectangle, draw, rounded corners, minimum width=15mm, minimum height=16mm] (T1) at ($(A) + (1.25,2.5)$) {};
		\node (T1R) at ($(T1) + (0,0.5)$) {\normalsize$\epsilon, \setnocond{yz}, 1$};
		\node (T1C0) at ($(T1) + (0,-0.5)$) {\normalsize$0, \setnocond{y}, 1$};
		\draw[-] (T1R) to node {} (T1C0);

		\node[rectangle, draw, rounded corners, minimum width=15mm, minimum height=6mm] (T2) at ($(A) + (1.25,0.75)$) {};
		\node (T2R) at (T2) {\normalsize$\epsilon, \setnocond{x}, 2$};

		\node[rectangle, draw, rounded corners, minimum width=15mm, minimum height=6mm] (T3) at ($(A) + (-1.25,0.75)$) {};
		\node (T3R) at (T3) {\normalsize$\epsilon, \setnocond{yz}, 2$};

		\draw ($(T0) + (0.75,0.1)$) to[bend left=15] node {$a$} ($(T1) + (-0.75,0.1)$);
		\draw ($(T1) + (-0.75,-0.1)$) to[bend left=15] node {$b$} ($(T0) + (0.75,-0.1)$);
		
		\draw[double distance=0.7pt] ($(T1) + (0,-0.8)$) to node {$c$} ($(T2) + (0,0.3)$);

		\draw ($(T2) + (-0.75,0.1)$) to[bend right=15] node[above] {$a$} ($(T3) + (0.75,0.1)$);
		\draw ($(T3) + (0.75,-0.1)$) to[bend right=15] node[below] {$c$} ($(T2) + (-0.75,-0.1)$);
		
		\draw[double distance=0.7pt] ($(T3) + (0,0.3)$) to node {$b$} ($(T0) + (0,-0.3)$);
	\end{tikzpicture}
}
\vskip 3mm
	\caption{The determinisation for $\baut_{\calE}$}
	\label{fig:exampleDeterm}
\end{wrapfigure}
Figure~\ref{fig:exampleDeterm} shows the DRA $\aut_{\calE}$ corresponding to the NGBA $\baut_{\calE}$ in Figure~\ref{fig:exampleBuechi}.
Each rounded box is a tree and each node contains the node identifier and the associated $l$- and $h$-labels, respectively.
$\aacc_{\epsilon}$ contains the double arrow transitions while $\arej_{0}$ contains the remaining transitions;
$\arej_{\epsilon}$ and $\aacc_{0}$ are empty.

\subsection{Markov Chains and Product}

A \emph{distribution} $\mu$ over a set $X$ is a function $\mu \colon X \to [0,1]$ such that $\sum_{x \in X} \mu(x) = 1$.
A \emph{Markov chain (MC)} is a tuple $\mc = (\pstates, \labelFunc, \pinit, \pmat)$, where $\pstates$ is a finite set of states, $\labelFunc \colon \pstates \to \Sigma$ is a \emph{labelling function}, $\pinit$ is the \emph{initial distribution}, and $\pmat \colon \pstates \times \pstates \to [0,1]$ is a \emph{probabilistic transition matrix} satisfying $\sum_{m' \in \pstates} \pmat(m, m') \in \setnocond{0,1}$ for all $m \in \pstates$.
A state $m$ is called \emph{absorbing} if $\sum_{m' \in \pstates} \pmat(m, m') = 0$.
We write $(m,m') \in \pmat$ for $\pmat(m,m') > 0$.

\begin{wrapfigure}[7]{right}{32mm}
	\centering
	\scriptsize
	\vskip-6mm
	\begin{tikzpicture}[->,>=stealth,auto]
		\path[use as bounding box] (-1.5,1.75) rectangle (1.5,0.15);
		
		\node (MC) at (0,0) {\normalsize$\mc_{\calE}$};
		\coordinate (mcatr) at ($(MC) + (0,0.75)$);
		\node (mca) at ($(mcatr) + (0,0.5)$) {\normalsize$a$};
		\node (mcc) at ($(mcatr) + (-1.4,-0.75)$) {\normalsize$c$};
		\node (mcb) at ($(mcatr) + (1.4,-0.75)$) {\normalsize$b$};
		
		\draw ($(mca.north) + (0,0.3)$) to node {} (mca.north);
		\draw (mca) -- (mcatr) -- node [below,near end] {~$2/3$} (mcc);
		\draw (mcatr) to node [below, near end] {$1/3$~~} (mcb);
		\draw (mcc) to[bend left=30] node [above, near end] {$1$} (mca);
		\draw (mcb) to[bend right=30] node [above, near end] {$1$} (mca);
	\end{tikzpicture}
	\caption{A MC with $\labelFunc(m) = m$ for each $m$.}
	\label{fig:exampleMC}
\end{wrapfigure}
A \emph{maximal path} of $\mc$ is an infinite sequence $\ppath = m_{0} m_{1} \ldots$ satisfying $\pmat(m_{i}, m_{i+1}) > 0$ for all $i \in \omega$, or a finite one if the last state is absorbing.
We denote by $\ppaths^{\mc}$ the set of all maximal paths of $\mc$.
An infinite path $\ppath = m_{0} m_{1} \ldots$ defines the word $\word(\ppath) = w_{0} w_{1} \ldots \in \Sigma^{\omega}$ with $w_{i} = \labelFunc(m_{i})$, $i \in \omega$.

Given a finite sequence $\ppath = m_{0} m_{1} \ldots m_{k}$, the \emph{cylinder} of $\ppath$, denoted by $\cyl(\ppath)$, is the set of maximal paths starting with prefix $\ppath$.
We define the probability of the cylinder set by $\prob^{\mc}\big(\cyl(m_{0} m_{1} \ldots m_{k})\big) \defeq \pinit(m_{0}) \cdot \prod_{i = 0}^{k-1} \pmat(m_{i}, m_{i+1})$.
For a given MC $\mc$, $\prob^{\mc}$ can be uniquely extended to a probability measure over the $\sigma$-algebra generated by all cylinder sets.

In this paper we are interested in $\omega$-regular properties $\lang \subseteq \Sigma^{\omega}$ and the probability $\prob^{\mc}(\lang)$ for some measurable set $\lang$.
Further, we define $\prob^{\mc}(\baut) \defeq \prob^{\mc}(\setcond{\ppath \in \ppaths^{\mc}}{\word(\ppath) \in \lang(\baut)})$ for an automaton $\baut$.
We write $\prob_{m}^{\mc}$ to denote the probability function when assuming that $m$ is the initial state.
Moreover, we omit the superscript $\mc$ whenever it is clear from the context.
We follow the standard way of computing this probability in the product of $\mc$ and a deterministic automaton for $\lang$.
\begin{definition}[Product MC and Projection]
	Given a MC $\mc = (\pstates, \labelFunc, \pinit, \pmat)$ and a deterministic automaton $\aut = (\Sigma, \astates, \dinit, \amat, \ACC)$, the \emph{product Markov chain} is defined by $\mc \prodMC \aut \defeq (\pstates \times \astates, \labelFunc', \pinit', \pmat')$ where
	\begin{itemize}
	\item
		$\labelFunc'\big((m,d)\big) = \labelFunc(m)$;
	\item
		$\pinit'\big((m,d)\big) = \pinit(m)$ if $d = \amat(\dinit, \labelFunc(m))$, $0$ otherwise; and
	\item
		$\pmat'\big((m,d), (m',d')\big)$ equals $\pmat(m,m')$ if $d' = \amat(d, \labelFunc(m'))$, and is $0$ otherwise.
	\end{itemize}

	We denote by $\projaut((m,d), (m',d'))$ the \emph{projection} on $\aut$ of the given $((m,d), (m',d')) \in \pmat'$, i.e., $\projaut((m,d), (m',d')) = (d, \labelFunc(m'), d')$, and by $\projaut(\transtyle{B})$ its extension to a set of transitions $\transtyle{B} \subseteq \amat'$, i.e., $\projaut(\transtyle{B}) = \setcond{\projaut(p,p')}{(p,p') \in \transtyle{B}}$.
\end{definition}
As we have accepting transitions on the edges of the automata, we propose product Markov chains with accepting conditions on their edges.
\begin{definition}[GMC, RMC, and PMC]
	Given a MC $\mc$ and a deterministic automaton $\aut$ with accepting set $\ACC$, the product automaton is $\mc \prodAut \aut \defeq (\mc \prodMC \aut, \ACC')$ where
	\begin{itemize}
	\item
		if $\ACC = \abuechiset$, then $\ACC' \defeq \abuechiset'$ where $\afinal'_{i} = \setcond{(p,p') \in \pmat'}{\projaut(p,p') \in \afinal_{i}} \in \abuechiset'$ for each $i \in \natIntK$ (Generalised \buchi Markov chain, GMC);
	\item
		if $\ACC = \arabinset$, then $\ACC' \defeq (\arabinsetacc', \arabinsetrej')$ where $\aacc'_{i} = \setcond{(p,p') \in \pmat'}{\projaut(p,p') \in \aacc_{i}} \in \arabinsetacc'$ and $\arej'_{i} = \setcond{(p,p') \in \pmat'}{\projaut(p,p') \in \arej_{i}} \in \arabinsetrej'$ for each $i \in \natIntK$ (Rabin Markov chain, RMC); and
	\item
		if $\ACC = \pri \colon \pmat \to \natIntK$, then $\ACC' \defeq \pri' \colon \pmat' \to \natIntK$ where $\pri'(p,p') = \pri(\projaut(p,p'))$ for each $(p,p') \in \pmat'$ (Parity Markov chain, PMC).
	\end{itemize}
\end{definition}
Thus, RMC, PMC, and GMC are Markov chains extended with the corresponding accepting conditions.
We remark that the labelling of the initial states of the Markov chain is taken into account in the definition of $\pinit'$.

\begin{definition}[Bottom SCC]
\label{def:scc}
	A \emph{bottom strongly connected component} (BSCC) $\scc \subseteq V$ is an SCC in the underlying digraph $(V, E)$ of a MC $\mc$, where all edges with source in $\scc$ have only successors in $\scc$ (i.e., for each $(v, v') \in E$, $v \in \scc$ implies $v' \in \scc$).
	We assume that a (bottom) SCC does not contain any absorbing state.
	Given an SCC $\scc$, we denote by $\scctran$ the transitions of $\mc$ in $\scc$, i.e., $\scctran = \setcond{(m,m') \in \pmat}{m,m' \in \scc}$.
\end{definition}

\section{Lazy Determinisation}
\label{sec:lazy}

We fix an input MC $\mc$ and a NGBA $\baut = \bautElements$ as a specification.
Further, let $\aut = \det(\baut)$ be the deterministic Rabin automaton (DRA) constructed for $\baut$ (cf.~\cite{Safra/88/Safra,Schewe/09/determinise,Schewe+Varghese/12/generalisedBuchi}), and let $\mc \prodAut \aut = (\pstates \times \astates, \labelFunc, \pinit, \pmat, \ACC)$ be the product RMC.
We consider the problem of computing $\prob_{m_{0}}^{\mc}(\baut)$, i.e., the probability that a run of $\mc$ is accepted by $\baut$.

\subsection{Outline of our Methodology}

We first recall the classical approach for computing $\prob^{\mc}(\baut)$, see~\cite{BaierK08} for details.
It is well known~\cite{BiancoA95} that the computation of $\prob^{\mc}(\baut)$ reduces to the computation of the probabilistic reachability in the product RMC $\mc \prodAut \aut$ with $\aut = \det(\baut)$.
We first introduce the notion of accepting SCCs:

\begin{definition}[Accepting SCC]
\label{def:acceptingSCCOfRMC}
	Given a MC $\mc$ and the DRA $\aut = \det(\baut)$, let $\scc$ be a bottom SCC of the product RMC $\mc \prodAut \aut$.
	We say that $\scc$ is accepting if there exists an index $i \in \natIntK$ such that $\aacc_{i} \cap \projaut(\scctran) \neq \emptyset$ and $\arej_{i} \cap \projaut(\scctran) = \emptyset$;
	we call each $s \in \scc$ an \emph{accepting state}.
	Moreover, we call the union of all accepting BSCCs the \emph{accepting region}.
\end{definition}

Essentially, since a BSCC is an ergodic set, once a path enters an accepting BSCC $\scc$, with probability $1$ it will take transitions from $\aacc_{i}$ infinitely often;
since $\aacc_{i}$ is finite, at least one transition from $\aacc_{i}$ is taken infinitely often.
Now we have the following reduction:

\begin{theorem}[\kern-0.75ex\cite{BiancoA95}]
\label{thm:BiancoA95}
	Given a MC $\mc$ and a \buchi automaton $\baut$, consider $\aut = \det(\baut)$.
	Let $U \subseteq \pstates \times \astates$ be the accepting region and let $\Diamond U$ denote the set of paths containing a state of $U$.
	Then, $\prob^{\mc}(\baut) = \prob^{\mc \prodAut \aut}(\Diamond U)$.
\end{theorem}

When all bottom SCCs are evaluated, the evaluation of the Rabin MC is simple:
we abstract all accepting bottom SCCs to an absorbing goal state and perform a reachability analysis, which can be solved in polynomial time~\cite{BiancoA95,BaierK08}.
Thus, the outline of the traditional probabilistic model checking approach for LTL specifications is as follows:
\begin{enumerate}
\item
	translate the NGBA $\baut$ into an equivalent DRA $\aut = \det(\baut)$;
\item
	build (the reachable fragment of) the product automaton $\mc \prodAut \aut$;
\item
	for each BSCC $\scc$, check whether $\scc$ is accepting.
	Let $U$ be the union of these accepting SCCs;
\item
	infer the probability $\prob^{\mc \prodAut \aut}(\Diamond U)$.
\end{enumerate}

The construction of the deterministic Rabin automaton used in the classical approach is often the bottleneck of the approach, as one exploits some variant of the approach proposed by Safra~\cite{Safra/88/Safra}, which is rather involved.
The lazy determinisation technique we suggest in this paper follows a different approach.

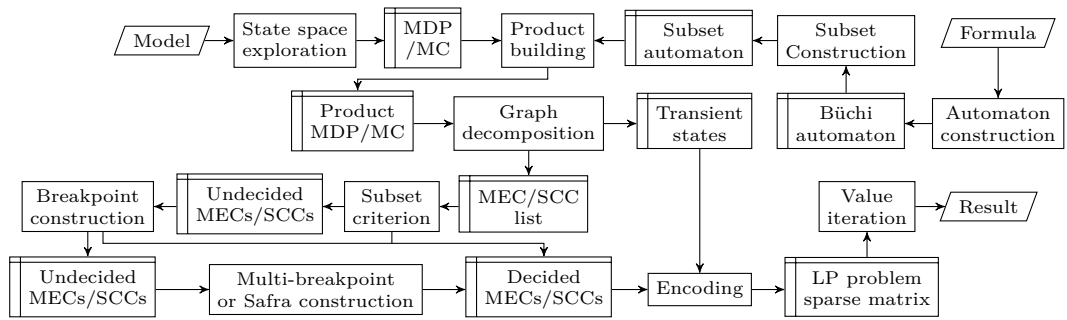
\begin{figure}
	\centering
\resizebox{1\linewidth}{!}{%
	\begin{tikzpicture}[->,>=stealth',auto]
		\scriptsize
		\tikzstyle{inputoutput} = [trapezium, draw, inner sep=3pt, trapezium right angle=-70pt, trapezium left angle=70pt, align=center];
		\tikzstyle{data} = [draw, internaldata, inner sep=3pt, align=center];
		\tikzstyle{code} = [draw, process, inner sep=3pt, align=center];
		
		\node[inputoutput] (model) at (-4,0) {Model};
		\node[code] (ssexp) at ($(model) + (1.75,0)$) {State space\neueZeile exploration};
		\node[data] (pmpg) at ($(ssexp) + (1.75,0)$) {MDP\neueZeile /MC};
		
		\node[inputoutput] (formula) at (7,0.1) {Formula};
		\node[code] (autconst) at ($(formula) + (0,-1.2)$) {Automaton\neueZeile construction};
		\node[data] (buechi) at ($(autconst) + (-2,0)$) {B\"{u}chi\neueZeile automaton};
		\node[code] (determ) at ($(buechi) + (0,1.1)$) {Subset\neueZeile Construction};
		\node[data] (parity) at ($(determ) + (-2,0)$) {Subset\neueZeile automaton};

		\node[code] (prodbuild) at ($0.45*(parity) + 0.55*(pmpg)$) {Product\neueZeile building};
		\node[data] (ppg) at ($(prodbuild) + (-2.5,-1.1)$) {Product\neueZeile MDP/MC};
		\node[code] (decalgo) at ($(ppg) + (2.25,0)$) {Graph\neueZeile decomposition};
		\node[data] (scclist) at ($(decalgo) + (0,-1.1)$) {MEC/SCC\neueZeile list};
		\node[data] (transientt) at ($(decalgo) + (2.25,0)$) {Transient\neueZeile states};
		\node[code] (subset) at ($(scclist) + (-1.8,0)$) {Subset\neueZeile criterion};
		\node[data] (subsetundecided) at ($(subset) + (-1.8,0)$) {Undecided\neueZeile MECs/SCCs};
		\node[code] (breakpoint) at ($(subsetundecided) + (-2.2,0)$) {Breakpoint\neueZeile construction};
		\node[data] (breakpointundecided) at ($(breakpoint) + (0,-1.1)$) {Undecided\neueZeile MECs/SCCs};
		\node[code] (rabin) at ($(breakpointundecided) + (3,0)$) {Multi-breakpoint\neueZeile or Safra construction};
		\node[data] (decided) at ($(rabin) + (3,0)$) {Decided\neueZeile MECs/SCCs};
		\node[code] (encoding) at ($(decided) + (2.05,0)$) {Encoding};
		\node[data] (lp) at ($(encoding) + (2.2,0)$) {LP problem\neueZeile sparse matrix};
		\node[code] (iteration) at ($(lp) + (0,1.1)$) {Value\neueZeile iteration};
		\node[inputoutput] (result) at ($(iteration) + (1.6,0)$) {Result};
		
		\draw (model) to (ssexp);
		\draw (ssexp) to (pmpg);
		
		\draw (formula) to (autconst);
		\draw (autconst) to (buechi);
		\draw (buechi) to (determ);
		\draw (determ) to (parity);
		
		\draw (pmpg) to (prodbuild);
		\draw (parity) to (prodbuild);
		
		\draw (prodbuild.south) to ($(prodbuild.south) + (0,-0.15)$) -| (ppg);
		\draw (ppg) to (decalgo);
		\draw (decalgo) to (scclist);
		\draw (decalgo) to (transientt);
		\draw (scclist) to (subset);
		\draw (subset) to (subsetundecided);
		\draw (subsetundecided) to (breakpoint);
		\draw (breakpoint) to (breakpointundecided);
		\draw (breakpointundecided) to (rabin);
		\draw (rabin) to (decided);
		\draw (decided) to (encoding);
		\draw (encoding) to (lp);
		\draw (lp) to (iteration);
		\draw (iteration) to (result);
		\draw[-] (subset.south) to ($(subset.south) + (0,-0.15)$);
		\draw ($(breakpoint.south)+(0.2,0.0)$) to ($(breakpoint.south) + (0.2,-0.15)$) -| (decided.north);
		\draw (transientt) -- (encoding);
	\end{tikzpicture}
}
	\caption{Overview of the model checking procedure}
	\label{fig:overview}
\end{figure}
We first transform the high-level specification (e.g., given in the \prism language~\cite{KwiatkowskaNP11}) into its MDP or MC semantics.
We then employ some tool (e.g., \ltlthreeba~\cite{BabiakKRS12} or \spot~\cite{DuretLutz11}) to construct a \buchi automaton equivalent to the LTL specification.
This nondeterministic automaton is used to obtain the deterministic \buchi over- and under-approximation subset automata $\ssaut^{u}$ and $\ssaut^{o}$, as described in Subsection~\ref{subset}.
The languages recognised by these two deterministic \buchi automata are such that $\lang(\ssaut^{u}) \subseteq \lang(\baut) \subseteq \lang(\ssaut^{o})$.
We build the product of these subset automata with the model MDP or MC (cf.~Lemma~\ref{lem:productSubsetIsomorphicQuotientMC}).
We then compute the maximal end components or bottom strongly connected components.
According to Lemma~\ref{lem:subset}, we try to decide these components of the product by using the acceptance conditions $\afinal^{o}_{i}$ and $\afinal^{u}_{i}$ of $\ssaut^{u}$ and $\ssaut^{o}$, respectively.

For each of those components where over- and under-approximation do not agree (and which we therefore cannot decide), we employ the breakpoint construction (cf.~Corollary~\ref{cor:breakpoint}), involving the deterministic Rabin over- and under-approximation breakpoint automata $\rbpaut^{u}$ and $\rbpaut^{o}$, such that $\lang(\rbpaut^{u}) \subseteq \lang(\baut) \subseteq \lang(\rbpaut^{o})$.
For this, we take one state of the component under consideration and start the breakpoint construction with this state as initial state.
This way, we obtain a product of a breakpoint automaton with parts of the model.
If the resulting product contains an accepting component (using the under-approximation), then the original component must be accepting, and if the resulting product contains a rejecting component (using the over-approximation), then the original component must be rejecting.

The remaining undecided components are decided either by using a Safra-based construction, restricted to the undecided component, or only by using $\rbpaut^{u}$, where we start from possibly different states of the subset product component under consideration; 
this approach always decides the remaining components, and we call it the multi-breakpoint construction.

For the model states that are part of an accepting component, or from which no accepting component is reachable, the probability to fulfil the specification is now already known to be $1$ or $0$, respectively.
To obtain the remaining state probabilities, we construct and solve a linear programming (LP) problem (or a linear equation system when we start with MCs).

Note that, even in case the multi-breakpoint (or Safra-based) procedure is necessary in some places, our method is usually still more efficient than direct Rabin determinisation, for instance based on some variation of~\cite{Schewe+Varghese/12/generalisedBuchi}.
The reason for this is twofold.
First, when starting the determinisation procedure from a component rather than from the initial state of the model, the number of states in the Rabin product will be smaller, and second, we only need the multi-breakpoint determinisation to decide MECs or bottom SCCs, such that the computation of transient probabilities can still be done in the smaller subset product.

The following optimisations can be used to speed up the model checking algorithm.
\begin{itemize}
\item
	We can compute the graph decomposition on the fly.
	Thus, we first compute one component, then decide it, compute the next component, etc.

\item
	If we have shown a state $(m,R)$ of the subset product to be accepting and $R' \supset R$, then $(m,R')$ is accepting.

\item
	We can treat all states, from which we find that we can reach an accepting component with probability 1, as accepting.

	Note that, if such a state is part of a MEC, this expands to the complete MEC, and if the state is initial, we can already terminate the algorithm.

\item
	Subset and breakpoint products can be effectively represented using BDDs~\cite{MSL08}.
\end{itemize}

In the remainder of this section, we detail the proposed approach:
we first introduce the theoretical background, and then present the incremental evaluation of the bottom SCCs.

\subsection{Acceptance Equivalence}
\label{thm:proof2}

In order to be able to apply our lazy approach, we exploit a number of acceptance equivalences in the RMC.
Given the DRA $\aut = \det(\baut)$ and a state $d$ of $\aut$, we denote by $\reached(d)$ the label of the root node $\epsilon$ of the labelled ordered tree associated to $d$ (cf.~\cite{Safra/88/Safra,Schewe/09/determinise,Schewe+Varghese/12/generalisedBuchi}).

\begin{theorem}
\label{thm:only_reach}
	Given a NGBA $\baut$ and $\aut = \det(\baut)$, let $d$ be an arbitrary state of $\aut$.
	Then, a word is accepted by $\aut_{d}$ if, and only if, it is accepted by $\baut_{\reached(d)}$.
\end{theorem}

Intuitively, a word $\word$ is accepted by $\aut_{d}$ if there is an accepting sequence $d_{0} d_{1} d_{2} \ldots$ with $d_{0} = d$ and $d_{i} = (\calT_{i}, l_{i}, h_{i})$ for each $i \in \omega$;
since each $l_{i+1}(\epsilon)$ is the set of states reached from $l_{i}(\epsilon)$ via $\word(i)$, then in $\baut_{\reached(d)} = \baut_{l_{0}(\epsilon)}$ there is a sequence of states $q_{0} q_{1} q_{2} \ldots$ such that each $q_{i+1} \in l_{i+1}(\epsilon)$ is reached from some $q_{i} \in l_{i}(\epsilon)$ via $\word(i)$;
such a sequence is accepting as well by the way $\det(\baut)$ is constructed.
A similar argument \linebreak applies for the other direction.
The formal proof is a mild generalisation of the correctness proof of the DRA construction~\cite{Schewe+Varghese/12/generalisedBuchi}.
Theorem~\ref{thm:only_reach} provides an immediate corollary.

\begin{corollary}
\label{cor:same}
	Given a NGBA $\baut$, a MC $\mc$, and the DRA $\aut = \det(\baut)$,
	\begin{inparaenum}[$(1.)$]
	\item
		a path $\arun$ in $\mc \prodAut \aut$ that starts from a state $(m,d)$ is accepted if, and only if, the word it defines is accepted by $\baut_{\reached(d)}$;
		and
	\item
		if $\reached(d) = \reached(d')$, then the probabilities of acceptance from a state $(m,d)$ and a state $(m,d')$ are equal, i.e., $\prob_{(m,d)}^{\mc \prodAut \aut}(\baut) = \prob_{(m,d')}^{\mc \prodAut \aut}(\baut)$.
	\end{inparaenum}
\end{corollary}

This property allows us to work on quotients \emph{and} to swap between states with the same reachability set.
If we ignore the accepting conditions, we have a product MC, and we can consider the quotient of such a product MC as follows.
\begin{definition}[Quotient MC]
	Given a MC $\mc$ and a DRA $\aut = \det(\baut)$, the \emph{quotient MC} $[\mc \prodMC \aut]$ of $\mc \prodMC \aut$ is the MC $([\pstates \times \astates], [\labelFunc], [\pinit], [\pmat])$ where
	\begin{itemize}
	\item
		$[\pstates \times \astates] = \setcond{(m,[d])}{(m,d) \in \pstates \times \astates,\ [d] = \setcond{d' \in \astates}{\reached(d') = \reached(d)}}$,
	\item
		$[\labelFunc](m,[d]) = \labelFunc(m,d)$,
	\item
		$[\pinit](m,[d]) = \pinit(m, d)$, and
	\item
		$[\pmat]\big((m,[d]), (m',[d'])\big) = \pmat\big((m,d),(m',d')\big)$.
	\end{itemize}
\end{definition}
By abuse of notation, we define $[(m,d)]= (m,[d])$ and $[C] = \setcond{[s]}{s \in C}$.
It is easy to see that, for each $d \in \astates$, $d \in [d]$ holds and that $[\pmat]$ is well defined:
for $(m, d_{1}), (m, d_2) \in [(m,d)]$, $\pmat\big((m, d_{1}), (m', [d'])\big) = \pmat\big((m,d),(m',d')\big) = \pmat\big((m, d_{2}), (m', [d'])\big)$ holds.

\begin{theorem}
\label{thm:bottom_component}
	For a MC $\mc$ and DRA $\aut = \det(\baut)$, it holds that
	\begin{enumerate}
	\item
		if $\scc$ is a bottom SCC of $\mc \prodMC \aut$ then $[\scc]$ is a bottom SCC of $[\mc \prodMC \aut]$,
	\item
		if $\scc'$ is a bottom SCC of $[\mc \prodMC \aut]$, then there is a bottom SCC $\scc$ of $\mc \prodMC \aut$ with $\scc' = [\scc]$.
	\end{enumerate}
\end{theorem}

Together with Definition~\ref{def:acceptingSCCOfRMC} and Theorem~\ref{thm:only_reach}, Theorem~\ref{thm:bottom_component} provides:

\begin{corollary}
\label{cor:quotients}
	Let $\scc$ be a bottom SCC of $[\mc \prodMC \aut]$.
	Then, either all states $s$ of $\mc \prodAut \aut$ with $[s] \in \scc$ are accepting, or all states $s$ of $\mc \prodAut \aut$ with $[s] \in \scc$ are rejecting.
\end{corollary}

Once all bottom SCCs are evaluated, we only need to perform a standard probabilistic reachability analysis on the quotient MC.

\subsection{Incremental Evaluation of Bottom SCCs}
\label{subset}

To evaluate each bottom SCC of the RMC, we use three techniques:
the first one is based on evaluating the subset construction directly.
We get two deterministic NGBAs that provide over- and under-approximations.
If this fails, we refine the corresponding bottom SCC by a breakpoint construction.
Only if both fail, a precise construction follows.

\subsubsection{Subset Construction}

For a given NGBA $\baut = \bautElements$, a simple way to over- and under-approximate its language by a subset construction is as follows.
We build two NGBAs $\ssaut^{o} = (\Sigma, 2^\astates, \setnocond{\ainit}, \amat', \abuechiset^{o})$ and $\ssaut^{u} = (\Sigma, 2^\astates, \setnocond{\ainit}, \amat', \abuechiset^{u})$, differing only for the accepting condition, where
\begin{itemize}
\item
	$\amat' = \setcond{ (R, \sigma, C)}{\emptyset \neq R \subseteq \astates, C = \amat(R, \sigma)}$,
\item
	$\afinal^{o}_{i} = \setcond{(R, \sigma, C) \in \amat'}{\exists (q, q') \in R \times C.\ (q, \sigma, q') \in \afinal_{i}} \in \abuechiset^{o}$ for each $i \in \natIntK$, and
\item
	$\afinal^{u}_{i} = \setcond{(R, \sigma, C) \in \amat'}{\forall (q, q') \in R \times C.\ (q, \sigma, q') \in \afinal_{i}} \in \abuechiset^{u}$ for each $i \in \natIntK$.
\end{itemize}

\begin{wrapfigure}[7]{right}{30mm}
	\centering
	\scriptsize
	\vskip-5mm
	\begin{tikzpicture}[->,>=stealth,auto]
        \path[use as bounding box] (-1.3,1.5) rectangle (1.5,0.15);

        \node (SUBA) at (0,0) {\normalsize$\ssaut_{\calE}$};
        \node (subx) at ($(SUBA) + (-1.1,0.9)$) {\normalsize$\setnocond{x}$};
        \node (subyz) at ($(SUBA) + (1.1,0.9)$) {\normalsize$\setnocond{y, z}$};

        \draw ($(subx.north) + (0,0.3)$) to node {} (subx.north);
        \draw (subx) to node[above, near end] {$a$} (subyz);
        \draw (subyz) to[bend right=30] node[above, very near end] {$~~b$} (subx);
        \draw (subyz) to[bend left=30] node[below, very near end] {$c$} (subx);

	\end{tikzpicture}
	\caption{The subset construction for $\baut_{\calE}$}
	\label{fig:exampleSubSet}
\end{wrapfigure}
Essentially, $\ssaut^{o}$ and $\ssaut^{u}$ are the subset automata that we use to over- and under-approximate the accepting conditions, respectively.
Figure~\ref{fig:exampleSubSet} shows the reachable fragment of the subset construction for the NGBA $\baut_{\calE}$ depicted in Figure~\ref{fig:exampleBuechi}.
The final sets of the two subset automata are $\afinal^{o}_{1} = \setnocond{(\setnocond{x}, a, \setnocond{yz})}$ and $\afinal^{o}_{2} = \setnocond{(\setnocond{yz}, b, \setnocond{x})}$ for $\ssaut^{o}$ and $\afinal^{u}_{1} = \afinal^{u}_{2} = \emptyset$ for $\ssaut^{u}$.
The following lemma holds:
\begin{lemma}
\label{lem:inclusions1}
	$\lang(\ssaut_{[d]}^{u}) \subseteq \lang(\aut_{d}) \subseteq \lang(\ssaut_{[d]}^{o})$.
\end{lemma}
The proof is easy as, in each $\afinal^{o}_{i}$ and $\afinal^{u}_{i}$, the accepting transitions are over- and under-approximated.
With this lemma, we are able to identify some accepting and rejecting bottom SCCs in the product.

We remark that $\ssaut^{o}$ and $\ssaut^{u}$ differ only in their accepting conditions.
Thus, the corresponding GMCs $\mc \prodAut \ssaut^{u}$ and $\mc \prodAut \ssaut^{o}$ also differ only for their accepting conditions.
If we ignore the accepting conditions, we have the following result:
\begin{lemma}
\label{lem:productSubsetIsomorphicQuotientMC}
	Let $\mc$ be a MC, $\baut$ a NGBA, $\aut = \det(\baut)$, and $\ssaut^{u}$ as defined above;
	let $\ssaut$ be $\ssaut^{u}$ without the accepting conditions.
	Then, $\mc \prodMC \ssaut$ and $[\mc \prodMC \aut]$ are isomorphic.
\end{lemma}
The proof is rather easy---it is based on the isomorphism identifying a state $(m, R)$ of $\mc \prodMC \ssaut$ with the state $(m,[d])$ of $[\mc \prodMC \aut]$ such that $\reached(d) = R$.

Considering the accepting conditions, we can classify some bottom SCCs.

\begin{lemma}
\label{lem:subset}
	Let $\mc$ be a MC and $\baut$ a NGBA.
	Let $\ssaut^{o}$ and $\ssaut^{u}$ be as defined above.
	Let $\scc$ be a bottom SCC of $\mc \prodAut \ssaut^{u}$.
	Then,
	\begin{itemize}
	\item
		$\scc$ is accepting if $\afinal^{u}_{i} \cap \proj{\ssaut^{u}}(\scctran) \neq \emptyset$ holds for all $i \in \natIntK$;
	\item
		$\scc$ is rejecting if $\afinal^{o}_{i} \cap \proj{\ssaut^{u}}(\scctran) = \emptyset$ holds for some $i \in \natIntK$.
	\end{itemize}
\end{lemma}

\begin{wrapfigure}[6]{right}{34mm}
	\centering
	\scriptsize
	\vskip-8.5mm
	\begin{tikzpicture}[->,>=stealth,auto]
		\path[use as bounding box] (-1.7,2) rectangle (1.7,0.15);
		
		\node (MBA) at (0,0) {\normalsize$\mc_{\calE} \prodMC \ssaut_{\calE}$};
		\node (mbaa) at ($(MBA) + (0,1.5)$) {\normalsize$a, \setnocond{y,z}$};
		\node (mbac) at ($(MBA) + (-1.3,0.25)$) {\normalsize$c, \setnocond{x}$};
		\node (mbab) at ($(MBA) + (1.3,0.25)$) {\normalsize$b, \setnocond{x}$};
		\coordinate (mbaatr) at ($(mbaa) + (0,-0.5)$);
		
		\draw ($(mbaa.north) + (0,0.3)$) to (mbaa.north);
		\draw (mbaa) -- (mbaatr) -- node [below, near end] {~~~~$2/3$} (mbac);
		\draw (mbaatr) to node [below, near end] {$1/3$~~~~} (mbab);
		\draw (mbac) to[bend left=30] node [above,very near end] {$1$~~} (mbaa);
		\draw (mbab) to[bend right=30] node [above,very near end] {~$1$} (mbaa);

	\end{tikzpicture}
	\caption{The product of $\mc_{\calE}$ and $\ssaut_{\calE}$}
	\label{fig:exampleMCTimesSubSet}
\end{wrapfigure}
The above result directly follows by Lemma~\ref{lem:inclusions1}.
Figure~\ref{fig:exampleMCTimesSubSet} shows the product of the MC $\mc_{\calE}$ depicted in Figure~\ref{fig:exampleMC} and the subset automaton $\ssaut_{\calE}$ in Figure~\ref{fig:exampleSubSet}. It is easy to check that the only bottom SCC is neither accepting nor rejecting.

For the bottom SCCs, for which we cannot conclude whether they are accepting or rejecting, we continue with the breakpoint construction.

\subsubsection{Breakpoint Construction}

For a given NGBA $\baut = \bautElements$, we denote with $\bp(\astates, k) = \setcond{(R, j, C)}{C \subsetneq R \subseteq \astates,\ j \in \natIntK}$ the breakpoint set, where
\begin{inparaenum}[$(1.)$]
\item
	$R$ intuitively refers to the set of currently reached states $l(\epsilon)$ in the root $\epsilon$ of an extended history tree $d = (\calT, l, h)$;
\item
	$j$ refers to the index $h(\epsilon)$ of the root; and
\item
	$C$ is the union of the $l$-labels of the children of the root.
\end{inparaenum}
That is, $(R, j, C) = \big(l(\epsilon), h(\epsilon), \bigcup_{i \in \calT \cap \naturals} l(i) \big)$, also denoted by $\langle d \rangle$, where $\calT \cap \naturals$ represents the nodes of $\calT$ that are children of $\epsilon$, i.e., $(R, j, C)$ is an abstraction of the tree $(\calT, l, h)$.

We build two DRAs $\rbpaut^{o} = (\Sigma, \bp(\astates, k), (\ainit, 1, \emptyset), \amat', \setnocond{(\aacc_{\epsilon},\emptyset), (\amat', \arej_{0})})$ and $\rbpaut^{u} = (\Sigma, \bp(\astates, k), (\ainit, 1, \emptyset), \amat', \setnocond{(\aacc_{\epsilon},\emptyset)})$, called the \emph{breakpoint automata}, as follows.

From the breakpoint state $(R, j, C)$, let $R' = \amat(R, \sigma)$
and $C' = \amat(C, \sigma) \cup \afinal_{j}(R, \sigma)$.
Then an accepting transition with letter $\sigma$ reaches $(R', j \oplus_{k} 1, \emptyset)$ if $C' = R'$.
This corresponds to the equivalence from Step~\ref{item:determinisation:accepting} that determines acceptance.
(Note that Step~\ref{item:determinisation:stealing} does not affect the union of the children's labels.)
Since Step~\ref{item:determinisation:accepting} removes all children, this is represented by using $\emptyset$ as label of the child $0$.
Formally,
\begin{align*}
	\aacc_{\epsilon} & = \{\,((R, j, C), \sigma, (R', j \oplus_{k} 1, \emptyset)) \mid (R, j, C) \in \bp(\astates, k),\ \sigma \in \Sigma,\\
	& \qquad \qquad \emptyset \neq R' = \amat(R, \sigma),\ C' = \amat(C, \sigma) \cup \afinal_{j}(R, \sigma),\ C' = R'\,\}\text{.} \\
\intertext{The remaining transitions, for which $C' \neq R'$, are obtained in a similar way, but now the transition reaches $(R', j, C')$, where $j$ remains unchanged; formally,}
	\amat'' & = \{\,((R, j, C), \sigma, (R', j, C')) \mid (R, j, C) \in \bp(\astates,k),\ \sigma \in \Sigma,\\
	& \qquad \qquad \emptyset \neq R' = \amat(R, \sigma),\ C' = \amat(C, \sigma) \cup \afinal_{j}(R, \sigma),\ C' \neq R'\,\}\text{.}\\
\intertext{The transition relation $\amat'$ is just $\amat'' \cup \aacc_{\epsilon}$.
Transitions that satisfy $C' = \emptyset$ are rejecting:}
	\arej_{0} & = \setcond{((R, j, C), \sigma, d) \in \amat''}{\amat(C, \sigma) = \emptyset}
	\text{.}
\end{align*}

\begin{wrapfigure}[7]{right}{43mm}
	\centering
	\scriptsize
	\vskip-12mm
	\begin{tikzpicture}[->,>=stealth,auto]
		\path[use as bounding box] (-1.95,0.15) rectangle (2.15,2.4);
		
		\node (BRKA) at (0,0) {\normalsize$\rbpaut_{\calE}$};
		\node (brk-x-2-O) at ($(BRKA) + (1.25,0.75)$) {\normalsize$\setnocond{x}, 2, \emptyset$};
		\node (brk-yz-1-y) at ($(BRKA) + (1.25,2)$) {\normalsize$\setnocond{y,z}, 1, \setnocond{y}$};
		\node (brk-x-1-O) at ($(BRKA) + (-1.25,2)$) {\normalsize$\setnocond{x}, 1, \emptyset$};
		\node (brk-yz-2-O) at ($(BRKA) + (-1.25,0.75)$) {\normalsize$\setnocond{y,z}, 2, \emptyset$};

		\draw ($(brk-x-1-O.north) + (0,0.3)$) to node {} (brk-x-1-O.north);
		\draw ($(brk-x-2-O.west) + (0,0.1)$) to node[above] {$a$} ($(brk-yz-2-O.east) + (0,0.1)$);
		\draw ($(brk-yz-2-O.east) + (0,-0.1)$) to node[below] {$c$} ($(brk-x-2-O.west) + (0,-0.1)$);

		\draw[double distance=0.7pt] (brk-yz-2-O) to node[left] {$b$} (brk-x-1-O);

		\draw ($(brk-x-1-O.east)+(0,-0.1)$) to node[below] {$a$} ($(brk-yz-1-y.west) + (0,-0.1)$);
		\draw ($(brk-yz-1-y.west) + (0,0.1)$) to node[above] {$b$} ($(brk-x-1-O.east)+(0,0.1)$);

		\draw[double distance=0.7pt] (brk-yz-1-y) to node[right] {$c$} (brk-x-2-O);
		
	\end{tikzpicture}
	\caption{The breakpoint construction for $\baut_{\calE}$ (fragment reachable from $(\setnocond{x}, 1, \emptyset)$)}
	\label{fig:exampleBreakpoint}
\end{wrapfigure}
Figure~\ref{fig:exampleBreakpoint} shows the reachable fragment of the breakpoint construction for the NGBA $\baut_{\calE}$ depicted in Figure~\ref{fig:exampleBuechi}.
The double arrow transitions are in $\aacc_{\epsilon}$ while the remaining transitions are in $\arej_{0}$.

\begin{theorem}
\label{theo:inclusions}
	The following inclusions hold:
	\[
		\lang(\ssaut_{[d]}^{u}) \subseteq \lang(\rbpaut_{\langle d \rangle}^{u}) \subseteq \lang(\aut_{d}) \subseteq \lang(\rbpaut_{\langle d \rangle}^{o}), \lang(\ssaut_{[d]}^{o})\text{.}
	\]
\end{theorem}

We remark that the breakpoint construction can be refined further such that it is finer than $\lang(\ssaut_{[d]}^{o})$.
However we leave it as future work to avoid heavy technical preparations.
Exploiting the above theorem, the following becomes clear.

\begin{corollary}
\label{cor:breakpoint}
	Let $\scc$ be a bottom SCC of the quotient MC.
	Let $(m, d) \in \scc$ be an arbitrary state of $\scc$.
	Moreover, let $\rbpaut^{o}$, $\rbpaut^{u}$ be the breakpoint automata.
	Then,
	\begin{itemize}
	\item
		$\scc$ is accepting if there exists a bottom SCC $\scc'$ in $\mc \prodAut \rbpaut^{u}_{\langle d \rangle}$ with $\scc = [\scc']$, which is accepting (i.e., $\scc'$ contains some transition in $\aacc_{\epsilon}$).
	\item
		$\scc$ is rejecting if there exists a bottom SCC $\scc'$ in $\mc \prodAut \rbpaut^{o}_{\langle d \rangle}$ with $\scc = [\scc']$, which is rejecting (i.e., $\scc'$ contains no transition in $\aacc_{\epsilon}$, but some transition in $\arej_{0}$).
	\end{itemize}
\end{corollary}

\begin{wrapfigure}[8]{right}{55mm}
	\centering
	\scriptsize
	\vskip-5mm
	\begin{tikzpicture}[->,>=stealth,auto]
		\path[use as bounding box] (-2.7,0.15) rectangle (2.8,2.5);
		
		\node (MBRKA) at (0,0) {\normalsize$\mc_{\calE} \prodAut \rbpaut^{u}_{\calE}$};
		\node (mbrk-x-2-O) at ($(MBRKA) + (1.6,0.75)$) {\normalsize$c,(\setnocond{x}, 2, \emptyset)$};
		\node (mbrk-yz-1-y) at ($(MBRKA) + (1.6,2)$) {\normalsize$a,(\setnocond{y,z}, 1, \setnocond{y})$};
		\node (mbrk-x-1-O) at ($(MBRKA) + (-1.6,2)$) {\normalsize$b,(\setnocond{x}, 1, \emptyset)$};
		\node (mbrk-yz-2-O) at ($(MBRKA) + (-1.6,0.75)$) {\normalsize$a,(\setnocond{y,z}, 2, \emptyset)$};
		
		\draw ($(mbrk-x-1-O.north)+(0,0.3)$) to node {} (mbrk-x-1-O.north);

		\draw ($(mbrk-x-2-O.west) + (0,0.1)$) to node[above, very near end] {$ 1$} ($(mbrk-yz-2-O.east) + (0,0.1)$);
		\draw ($(mbrk-yz-2-O.east) + (0,-0.1)$) to node[below, very near end] {$2/3$} ($(mbrk-x-2-O.west) + (0,-0.1)$);

		\draw[double distance=0.7pt] (mbrk-yz-2-O) to node[left, very near end] {$1/3$} (mbrk-x-1-O);

		\draw ($(mbrk-x-1-O.east)+(0,-0.1)$) to node[below, very near end] {$1$} ($(mbrk-yz-1-y.west) + (0,-0.1)$);
		\draw ($(mbrk-yz-1-y.west) + (0,0.1)$) to node[above, very near end] {$1/3$} ($(mbrk-x-1-O.east)+(0,0.1)$);

		\draw[double distance=0.7pt] (mbrk-yz-1-y) to node[right, very near end] {$2/3$} (mbrk-x-2-O);
		
	\end{tikzpicture}
	\caption{The product of $\mc_{\calE}$ and $\rbpaut^{u}_{\calE}$}
	\label{fig:exampleMCTimesBrkpnt}
\end{wrapfigure}
Note that the products $\mc \prodAut \rbpaut^{u}_{\langle d \rangle}$ and $\mc \prodAut \rbpaut^{o}_{\langle d \rangle}$ are the same RMCs except for their accepting conditions.
Figure~\ref{fig:exampleMCTimesBrkpnt} shows the product of the MC $\mc_{\calE}$ depicted in Figure~\ref{fig:exampleMC} and the breakpoint automaton $\rbpaut^{u}_{\calE}$ in Figure~\ref{fig:exampleBreakpoint}.
It is easy to see that the only bottom SCC is accepting.

Together with Corollary~\ref{cor:quotients}, Lemma~\ref{lem:subset} and Corollary~\ref{cor:breakpoint} immediately provide the following result, which justifies the incremental evaluations of the bottom SCCs.

\begin{corollary}
	Given a MC $\mc$, a NGBA $\baut$, and $\aut = \det(\baut)$, if $[(m,d)]$ is a state in a bottom SCC of the quotient MC and $[d] = [d']$, then
	\begin{itemize}
	\item
		$\prob_{(m,d)}^{\mc \prodAut \aut_{d}}(\baut) = 1$ if $\prob_{(m,[d])}^{\mc \prodAut \ssaut_{[d']}^{u}}(\baut) > 0$ or $\prob_{(m, \langle d \rangle)}^{\mc \prodAut \rbpaut_{\langle d' \rangle}^{u}}(\baut) > 0$, and
	\item
		$\prob_{(m,d)}^{\mc \prodAut \aut_{d}}(\baut) = 0$ if $\prob_{(m,[d])}^{\mc \prodAut \ssaut_{[d']}^{o}}(\baut) < 1$ or $\prob_{(m, \langle d \rangle)}^{\mc \prodAut \rbpaut_{\langle d' \rangle}^{o}}(\baut) < 1$.
	\end{itemize}
\end{corollary}

In case there are remaining bottom SCCs, for which we cannot conclude whether they are accepting or rejecting, we continue with a multi-breakpoint construction that is language-equivalent to the Rabin construction.

\subsubsection{Multi-Breakpoint Construction}

The multi-breakpoint construction we propose to decide the remaining bottom SCCs makes use of a combination of the subset and breakpoint constructions we have seen in the previous steps, but with different accepting conditions:
for the subset automaton $\ssaut = \bToSub(\baut) = (\Sigma, \ssstates, \ssinit, \ssmat, \ssfinal)$, we use the accepting condition $\ssfinal = \emptyset$, i.e., the automaton accepts no words;
for the breakpoint automaton $\bbpaut = \bToBp(\baut) = (\Sigma, \bpstates, \bpinit, \bpmat, \bpfinal)$, we consider $\bpfinal = \aacc_{\epsilon}$.
Note that the \buchi acceptance condition $\bpfinal = \aacc_{\epsilon}$ is trivially equivalent to the Rabin acceptance condition $\setnocond{(\aacc_{\epsilon}, \emptyset)}$, so $\bbpaut$ is essentially $\rbpaut^{u}$.
We remark that in general the languages accepted by $\ssaut$ and $\bbpaut$ are different from $\lang(\baut)$:
$\lang(\ssaut) = \emptyset$ by construction while $\lang(\bbpaut) \subseteq \lang(\baut)$, as shown in Theorem~\ref{theo:inclusions}.
To generate an automaton accepting the same language of $\baut$, we construct a \emph{semi-deterministic} automaton $\sdaut = \bToSd(\baut) = (\Sigma, \sdstates, \sdinit, \sdmat, \sdfinal)$ by merging $\ssaut$ and $\bbpaut$ as follows:
$\sdstates = \ssstates \cup \bpstates$, $\sdinit = \ssinit$, $\sdmat = \ssmat \cup \sdmattra \cup \bpmat$, and $\sdfinal = \bpfinal$, where $\sdmattra = \setcond{(R, \sigma, (R', j', C'))}{\text{$R \in \ssstates$, $(R', j', C') \in \bpstates$, and $R' \subseteq \ssmat(R, \sigma)$}}$.
The equivalence between the languages accepted by $\baut$ and $\sdaut$ will be established in Section~\ref{ssec:languageEquivalenceSemiDet}, more precisely by Proposition~\ref{pro:buechiLangEqualSemiDetLang}, but we point it out here as it is used in the proof of Theorem \ref{thm:acceptanceOfSubsetAutomaton}.

For $\aut = \det(\baut)$, it is known by Lemma~\ref{lem:productSubsetIsomorphicQuotientMC} that $\mc \prodMC \ssaut$ and $\mc \prodMC \aut$ are strictly related, so we can define the accepting SCC of $\mc \prodMC \ssaut$ by means of the accepting states of $\mc \prodMC \aut$.
\begin{definition}
\label{def:acceptingSccOfSubset}
	Given a MC $\mc$ and a NGBA $\baut$, for $\ssaut = \bToSub(\baut)$ and $\aut = \det(\baut)$, we say that a bottom SCC $\scc$ of $\mc \prodMC \ssaut$ is accepting if, and only if, there exists a state $s = (m, d)$ in an accepting bottom SCC $\scc'$ of $\mc \prodMC \aut$ such that $(m, \reached(d)) \in \scc$.
\end{definition}
Note that Corollary~\ref{cor:same} ensures that the accepting SCCs of $\mc \prodMC \ssaut$ are well defined.

\begin{theorem}
\label{thm:acceptanceOfSubsetAutomaton}
	Given a MC $\mc$ and a NGBA $\baut$, for $\sdaut = \bToSd(\baut)$ and  $\ssaut = \bToSub(\baut)$, the following facts are equivalent:
	\begin{enumerate}
	\item
		$\scc$ is an accepting bottom SCC of $\mc \prodMC \ssaut$;
	\item
		there exist $(m,R) \in \scc$ and $R' \subseteq R$ such that $(m, (R', j, \emptyset))$ belongs to an accepting SCC of $\mc \prodAut \sdaut_{(m, (R', j, \emptyset))}$ for some $j \in \natIntK$;
	\item
		there exist $(m,R) \in \scc$ and $q \in R$ such that $(m,(\setnocond{q}, 1, \emptyset))$ reaches with probability $1$ an accepting SCC of $\mc \prodAut \sdaut_{(m, (\setnocond{q}, 1, \emptyset))}$.
	\end{enumerate}
\end{theorem}

Theorem~\ref{thm:acceptanceOfSubsetAutomaton} provides a practical way to check whether an SCC $\scc$ of $\mc \prodMC \ssaut$ is accepting:
it is enough to check whether some state $(m, R)$ of $\scc$ has $R \supseteq R'$ for some $(m, (R', j, \emptyset))$ in the accepting region of $\mc \prodAut \sdaut$, or whether, for a state $q \in R$, $(m, (\setnocond{q}, 1, \emptyset))$ reaches with probability $1$ the accepting region.
We remark that, by construction of $\sdaut$, if we change the initial state of $\sdaut$ to $(R, j ,C)$---i.e., if we consider $\sdaut_{(R, j, C)}$---then the run can only visit breakpoint states; i.e., it is actually a run of $\bbpaut_{(R, j, C)}$.

\section{Semi-Determinisation}
\label{sec:motivation}

Based on the Theorem \ref{thm:BiancoA95}, the classical approach for evaluating MCs for LTL specifications is sketched as
follows:
\begin{enumerate}
\item
	translate the NGBA $\baut$ into an equivalent DPA $\daut = \det(\baut)$;
\item
	build (the reachable fragment of) the product automaton $\mc \prodAut \daut$;
\item
	for each bottom SCC $\scc$, check whether $\scc$ is accepting.
	Let $U$ be the union of these accepting SCCs;
\item
	abstract all accepting bottom SCCs to an absorbing goal state and perform a reachability analysis to infer $\prob^{\mc \prodAut \daut}(\Diamond U)$, which can be solved in polynomial time~\cite{BiancoA95,BaierK08}.
\end{enumerate}

The classical approach is to construct a deterministic Rabin automaton in step~1 and thus to evaluate Rabin acceptance conditions in step~3~\cite{Safra/88/Safra,Schewe+Varghese/12/generalisedBuchi}.
The size of such deterministic Rabin automaton is $m \cdot n^{\bigO(k \cdot n)}$ where $n$ and $k$ are the number of states and accepting sets of $\baut$, respectively, and $m$ the number of states of $\mc$.
\begin{figure}[t]
	\centering
\resizebox{\textwidth}{!}{%
	\begin{tikzpicture}
		\node[draw, text height=2ex, minimum width=55mm, inner sep=0pt, minimum height=3ex] (headingDP) at (0,0) {\progHeader{$\ComputeAcceptingScc(\mc, \baut)$}};
		\node[below=-0.1ex of headingDP] (bodyDP) {%
		\begin{minipage}[h]{55mm}
			\begin{algorithmic}[1]
				\STATE{$\mathit{Acc}_{\mathit{Scc}} = \emptyset$}
				\STATE{$\ssaut = \BuildSubset(\baut)$}
				\STATE{$\mc \prodMC \ssaut = \BuildProduct(\mc,\ssaut)$}
				\STATE{$\mathit{Scc} = \ComputeScc(\mc \prodMC \ssaut)$}
				\FORALL{$\scc \in \mathit{Scc}$}
					\IF{$\IsAccepting(\mc, \baut, \scc)$}
						\STATE{$\mathit{Acc}_{\mathit{Scc}} = \mathit{Acc}_{\mathit{Scc}} \cup \setnocond{\scc}$}
					\ENDIF
				\ENDFOR
				\RETURN{$\mathit{Acc}_{\mathit{Scc}}$}
			\end{algorithmic}
		\end{minipage}%
		};
		\coordinate (endDP) at ($(headingDP.south west) + (0.007,-3.7)$);
		\draw ($(headingDP.south west) + (0.007,0)$) -- (endDP) -- ($(endDP)+(2em,0)$) {};

		\node[draw, text height=2ex, minimum width=78mm, inner sep=0pt, minimum height=3ex] (headingISACC) at (7,0) {\progHeader{$\IsAccepting(\mc, \baut, \scc)$}};
		\node[below=-0.1ex of headingISACC] (bodyISACC) {%
		\begin{minipage}[h]{78mm}
			\begin{algorithmic}[1]
				\STATE{take a state $(m, S) \in \scc$}
				\FORALL{$q \in S$}
					\STATE{$\bbpaut = \BuildBreakpoint(\baut, \setnocond{q})$}
					\STATE{$\mc \prodAut \bbpaut {=} \BuildProduct(\mc, \bbpaut, (m,\setnocond{q}))$}
					\IF{$(m,\setnocond{q})$ is accepting in $\mc \prodAut \bbpaut$}
						\RETURN{\TRUE}
					\ENDIF
				\ENDFOR
				\RETURN{\FALSE}
			\end{algorithmic}
		\end{minipage}
		};
		\coordinate (endISACC) at ($(headingISACC.south west) + (0.007,-3.25)$);
		\draw ($(headingISACC.south west) + (0.007,0)$) -- (endISACC) -- ($(endISACC)+(2em,0)$) {};
	\end{tikzpicture}
}
	\vskip-1mm
	\caption{Algorithm to compute accepting SCCs of $\mc \prodMC \bToSub(\baut)$}
	\label{fig:ComputeAcceptingScc}
\end{figure}
By using the isomorphism between the product MC $\mc \prodMC \ssaut$ of $\mc$ and $\ssaut = \bToSub(\baut)$ and the quotient MC $[\mc \prodMC \daut]$, in Lemma~\ref{lem:productSubsetIsomorphicQuotientMC} we have established that it is enough to check whether each SCC of $\mc \prodMC \ssaut$ is accepting.
Then, computing the probability $\prob^{\mc}(\baut)$ simply reduces to computing the probability of reaching the accepting SCCs in $\mc \prodMC \ssaut$.
The latter step is analogous to the classical one, so let us focus on the former.
The corresponding pseudocode makes use of the procedures $\ComputeAcceptingScc(\mc, \baut)$ and $\IsAccepting(\mc, \baut, \scc)$, computing the accepting SCCs of $\mc \prodMC \ssaut$ and whether the SCC $\scc$ is accepting, respectively.

The procedure $\ComputeAcceptingScc(\mc, \baut)$ is the same as the one depicted in Figure~\ref{fig:ComputeAcceptingScc} and works as follows:
\begin{inparaenum}[1.)]
\item
	we build the subset automaton $\ssaut$ and its product with the MC $\mc$, $\mc \prodMC \ssaut$;
\item
	we compute the SCCs of $\mc \prodMC \ssaut$;
\item
	for each SCC $\scc$, we decide whether it is accepting and we collect into $\mathit{Acc}_{\mathit{Scc}}$ all accepting SCCs.
	By Corollary~\ref{cor:same}, we have that $\mathit{Acc}_{\mathit{Scc}}$ contains all states of $\mc \prodMC \ssaut$ corresponding to the accepting states of $\mc \prodMC \daut$.
\end{inparaenum}
Note that we just need to work with the reachable fragment of $\ssaut$ and of $\mc \prodMC \ssaut$, since the unreachable parts do not contribute to the evaluation of $\prob^{\mc}(\baut)$.

Regarding the procedure $\IsAccepting(\mc, \baut, \scc)$, we first verify whether $\scc$ is accepting via over- and under-approximating acceptance conditions on the subset construction itself.
This technique is in not complete; if $\scc$ has not been decided, we proceed with over- and under-approximating acceptance conditions on the breakpoint construction.
If $\scc$ is still not decided, we finally fall back to the original Rabin construction, but only on the states of~$\scc$.

In this work we prove that it is sufficient to use subset and breakpoint constructions for identifying accepting SCCs.
The new construction avoids the Rabin (or parity) determinisation of the \buchi automaton completely and gives an improved complexity.
The pseudocode of our semi-deterministic construction is the one depicted in
Figure~\ref{fig:ComputeAcceptingScc}:
deciding the accepting SCCs of $\mc \prodMC \ssaut$ is based on breakpoint automata $\bbpaut$ with \buchi acceptance conditions;
the method is conclusive, but different initial states might have to be considered.
The evaluation of the probability of reaching such accepting SCCs requires only $\mc \prodMC \ssaut$.
The algorithm we propose is rather simple, but its correctness is much more involved and we devote the next section to show that our novel approach is correct.
The correctness is based on the equivalence of the given NGBA $\baut$ and a semi-deterministic \buchi automaton whose initial part is generated via subset construction while for the final part the breakpoint construction is used.
The transit transitions between these two parts connect each subset state $S$ to each possible breakpoint state $(R',j',C')$ having $R' \subseteq \amat(R,\sigma)$.
As Theorem~\ref{thm:acceptanceOfSubsetAutomaton} will show, an SCC of $\mc \prodMC \ssaut$ is accepting if and only if by performing one of such transit transitions we land directly inside an accepting SCC of $\mc \prodAut \bbpaut$.
Alternatively, by the same theorem, an SCC of $\mc \prodMC \ssaut$ is accepting if and only if by performing one of such transit transitions we land to a state $(m, (\setnocond{q}, 1, \emptyset))$ that reaches accepting SCCs of $\mc \prodAut \bbpaut$ with probability $1$.
Note that such SCCs may be unreachable from the usual initial state of $\mc \prodAut \bbpaut$, so this does not contradict the fact that the usual breakpoint construction is not enough to decide whether $\word \in \lang(\baut)$ for a given word $\word \in \Sigma^{\omega}$.

\section{Correctness}
\label{sec:determinisation}

In this section we show that the semi-determinisation construction is correct.
This result is achieved in several steps.
We first describe how a NGBA $\baut$ can be converted into a semi-deterministic \buchi automaton $\sdaut$ via subset and breakpoint constructions;
then we show that $\baut$ and $\sdaut$ recognise the same language;
next, we consider the parity determinisation $\daut$ of $\sdaut$ and we show that again the recognised language is preserved;
finally, we relate the accepting SCCs of $\mc \prodMC \ssaut$ with the accepting SCCs of $\mc \prodAut \sdaut$ and of $\mc \prodAut \bbpaut$, where $\ssaut = \bToSub(\baut)$ and $\bbpaut = \bToBp(\baut)$.
We conclude the section with the complexity analysis of the semi-determinisation construction.

\subsection{Semi-Determinisation of NGBAs}

The first step for proving the correctness of our semi-determinisation construction is the generation of a semi-deterministic automaton corresponding to the given $\baut$, by using the subset and breakpoint constructions.

\begin{definition}
\label{def:semi-determinisation}
	Given a NGBA $\baut$, consider the subset and \buchi breakpoint automata $\bToSub(\baut) = (\Sigma, \ssstates, \ssinit, \ssmat, \ssfinal)$ and $\bToBp(\baut) = (\Sigma, \bpstates, \bpinit, \bpmat, \bpfinal)$, respectively.
	The \emph{semi-determinisation} of $\baut$ is the \buchi automaton $\bToSd(\baut) = (\Sigma, \sdstates, \sdinit, \sdmat, \sdfinal)$ where $\sdstates = \ssstates \cup \bpstates$ is the set of states, $\sdinit = \ssinit$ is the initial state, $\sdmat = \ssmat \cup \sdmattra \cup \bpmat$ is the transition relation, $\sdfinal = \bpfinal$ is the accepting set, and $\sdmattra$ is defined as $\sdmattra = \setcond{(R, \sigma, (R', j', C'))}{\text{$R \in \ssstates$, $(R', j', C') \in \bpstates$, and $R' \subseteq \ssmat(R, \sigma)$}}$.
\end{definition}

Thus, the semi-deterministic automaton $\bToSd(\baut)$ consists of two deterministic parts:
an initial part with the states $\sdstatesini = \ssstates$, where the automaton follows the subset construction, and a final part with the states $\sdstatesfin = \bpstates$, where the automaton follows the breakpoint construction.
Within a run of an automaton, there is only (or: at most) a single step that is not following this deterministic pattern:
the transition taken from $\sdmattra$ from the initial to the final part.
By construction, it is clear that the resulting automaton is semi-deterministic.

\begin{figure*}
	\centering
	\scriptsize
	\begin{tikzpicture}[->,>=stealth,shorten >=1pt,auto,xscale=1.1,yscale=0.9]
		\path[use as bounding box] (-5.2,0.4) rectangle (5.75,3.5);

		\node (SD) at (0,0) {};
		\node (SUBA) at ($(SD) + (-4,0.25)$) {Subset construction};
		\coordinate (SUBAC) at ($(SUBA) + (0,2)$);
		\node (subx) at ($(SUBAC) + (-1,0)$) {$\setnocond{x}$};
		\node (subyz) at ($(SUBAC) + (1,0)$) {$\setnocond{y,z}$};
		
		\draw ($(subx) + (0,0.5)$) to node {} ($(subx) + (0,0.2)$);
		\draw (subx) to node[above] {$a$} (subyz);
		\draw (subyz) to[bend left=30] node[below] {$b$} (subx);
		\draw (subyz) to[bend right=30] node[above] {$c$} (subx);

		\node (BRKA) at ($(SD) + (2.5,0.25)$) {Breakpoint construction};
		\coordinate (BRKAC) at ($(BRKA) + (-1.5,2)$);
		\node (brk-x-O-2) at ($(BRKAC) + (0,-1)$) {$\setnocond{x}, 2, \emptyset$};
		\node (brk-x-O-1) at ($(BRKAC) + (0,1)$) {$\setnocond{x}, 1, \emptyset$};
		\node (brk-yz-O-2) at ($(BRKAC) + (-1.5,0)$) {$\setnocond{y,z}, 2, \emptyset$};
		\node (brk-yz-y-1) at ($(BRKAC) + (1.5,0)$) {$\setnocond{y,z}, 1, \setnocond{y}$};

		\node (brk-yz-z-1) at ($(BRKAC) + (4,1)$) {$\setnocond{y,z}, 1, \setnocond{z}$};
		\node (brk-y-O-2) at ($(BRKAC) + (4,-1)$) {$\setnocond{y}, 2, \emptyset$};

		\draw (brk-x-O-2) to[bend left=15] node {$a$} (brk-yz-O-2);
		\draw (brk-yz-O-2) to[bend left=15] node {$c$} (brk-x-O-2);

		\draw[double distance=0.7pt] (brk-yz-O-2) to node {$b$} (brk-x-O-1);

		\draw (brk-x-O-1) to[bend left=15] node {$a$} (brk-yz-y-1);
		\draw (brk-yz-y-1) to[bend left=15] node {$b$} (brk-x-O-1);

		\draw[double distance=0.7pt] (brk-yz-y-1) to node {$c$} (brk-x-O-2);
		
		\draw (brk-y-O-2) to node[above] {$c$} (brk-x-O-2);
		\draw[double distance=0.7pt] (brk-yz-z-1) to[bend left=15] node {$b$} (brk-x-O-2);
		\draw (brk-yz-z-1) to node {$c$} (brk-x-O-1);

		\draw[dotted] (subyz) to[bend left=15] node {$b$} (brk-x-O-1);
		\draw[dotted] (subyz) to[bend right=15] node {$b$} (brk-x-O-2);
		\draw[dotted] (subyz) to[bend left=30] node {$c$} (brk-x-O-1);
		\draw[dotted] (subyz) to[bend right=30] node {$c$} (brk-x-O-2);
		\draw[dotted] (subx) .. controls ($(subx) + (0.75,-1.5)$) and ($(brk-yz-O-2) + (-1,-0.4)$) .. node {$a$} (brk-yz-O-2);
		\draw[dotted] (subx) .. controls ($(subx) + (1,-2)$) and ($(brk-y-O-2) + (-4,-0.75)$) .. node[near end] {$a$} (brk-y-O-2);
		\draw[dotted] (subx) .. controls ($(subx) + (1,2)$) and ($(brk-yz-z-1) + (-4,0.75)$) .. node[near end] {$a$} (brk-yz-z-1);
		\draw[dotted] (subx) .. controls ($(subx) + (1.5,2)$) and ($(brk-yz-y-1) + (-3,0)$) .. node[above,near start] {\kern10mm$a$} (brk-yz-y-1);

	\end{tikzpicture}
	\caption{Semi-determinisation $\sdaut_{\calE}$ of $\baut_{\calE}$ in Figure~\ref{fig:exampleBuechi} (fragment)}
	\label{fig:exampleSemiDet}
\end{figure*}

Figure~\ref{fig:exampleSemiDet} shows the semi-determinisation of the \buchi automaton $\baut_{\calE}$ in Figure~\ref{fig:exampleBuechi}.
The left hand side is the initial part, obtained via the subset construction;
the right hand side is a fragment of the final part, generated via the breakpoint construction.
We remark that the breakpoint construction has 38 states (of which only 12 states are reachable via transit transitions) while we have depicted only 6 of them.
Double arrows are transitions belonging to $\sdfinal = \bpfinal$ while dotted arrows are (some of) the transit transitions in $\sdmattra$.

\subsection{Determinising and Applying SDAs}
Given a SDBA $\sdaut = (\Sigma, \sdstates, \sdinit, \sdmat, \sdfinal)$ with $\sdstates = \sdstatesini \cup \sdstatesfin$ and $\sdmat = \sdmatini \cup \sdmattra \cup \sdmatfin$, we can construct a DPA $\daut = (\Sigma, \dstates, (\sdinit, \emptyset), \dmat, \pri)$ (where $\emptyset$ represents the function with an empty domain) as follows.
Let $\sdmat^{\blank}$ be the completion of $\sdmat$ that maps every element of $\sdstates \times \Sigma$ not in the domain of $\sdmat$ to a fresh symbol $\blank$.
For the components $\sdmatini^{\blank}$ and $\sdmatfin^{\blank}$, we use according definitions.

A state in $\dstates$ is a pair $(r,f)$, consisting of the state $r$ reached through the extended initial transitions $\sdmatini^{\blank}$ and a bijection $f \colon \natIntK[m] \to R$ for a set $R \subseteq \sdstatesfin$ with $m = |R|$.

The transition $\dmat \colon \big((r,f), \sigma\big) \mapsto (r',f')$ is defined as follows:
\begin{itemize}
\item
	update of subset part:
	$r' = \sdmatini^{\blank}(r,\sigma)$;
\item
	updating breakpoint states:
	let $g \colon \natIntK[m] \to R'$ be a surjection with $R' \subseteq \sdstatesfin \cup \setnocond{\blank}$ defined as $g(j) = \sdmatfin^{\blank}\big(f(j), \sigma\big)$ for each $j \in \natIntK[m]$;
\item
	minimal acceptance number:
	let $a$ be the minimal integer such that $\big(f(a), \sigma, g(a)\big) \in \sdfinal$ is an accepting transition if such an integer exists, and $a = |\sdstatesfin| + 1$ otherwise;
\item
	removing duplicate breakpoint states:
	let $g' \colon \natIntK[m] \to R'$ be obtained from $g$ by replacing, for every $h,j \in \natIntK[m]$ such that $j > h$ and $g(j) = g(h)$, $g(j)$ by $\blank$;
\item
	minimal rejecting number:
	let $d$ be the minimal integer with $g'(d) = \blank$ if such an integer exists, and $d = |\sdstatesfin| + 1$ otherwise;
\item
	removing blanks:
	let $g'' \colon \natIntK[m''] \to R''$ be a bijection with $R'' = R' \setminus \setnocond{\blank}$ where $m'' = |R''|$;
	$g''$ is obtained from $g'$ by removing the $\blank$ signs while preserving the order, that is, if $h < j$, $g''(h) = g'(h')$, and $g''(j) = g'(j')$, then $h' < j'$;
\item
	restoring transit transitions:
	$f' \colon \natIntK[m'] \to S$ is a bijection with $f'(h) = g''(h)$ for all $h \leq m''$, $S = R'' \cup \setcond{q \in \sdstatesfin}{(r, \sigma, q) \in \sdmattra}$, and $m' = |S|$; and
\item
	transition priority:
	the priority of this transition $\big((r,f), \sigma, (r',f')\big)$ is $2d-1$ if $d \leq a$ and $2a$ if $a < d$.
\end{itemize}
Note that in the above definition, for $f'$ the assignment of numbers $i$ with $m'' < i \leq m'$ to elements of $\setcond{q \in \sdstatesfin}{(r, \sigma, q) \in \sdmattra}$ is arbitrary as long as $f'$ is a bijection.
We denote by $\sdToDet(\sdaut)$ the DPA $\daut$ constructed as above from $\sdaut$.

Given a NGBA $\baut$, we write $\daut = \det(\baut)$ to denote the automaton $\daut = \sdToDet(\bToSd(\baut))$ and for a state $d = (r,f)$ of $\daut$, we denote by $\reached(d)$ the states reached in $d$, i.e., $\reached(d) = r$.

The parity automaton follows the initial subset part of the semi-deterministic automaton in the part $r$ of a state $(r,f)$.
It simulates the final breakpoint part via the function $f$ that stores the nondeterministic choice of where to start in the breakpoint part by assigning them to the entries $f(i)$ while preserving the previous choices.

\begin{figure}
	\centering
	\scriptsize
	\begin{tikzpicture}[->,>=stealth,shorten >=1pt,auto]
	\path[use as bounding box] (-2.55,0.25) rectangle (11,4.5);
		
		\node (P) at (0,0) {};
		\node (Pl) at ($(P) + (-2,3.5)$) {\normalsize $\daut_{\calE}$};
		\node (blank) at ($(P) + (2,3.5)$) {$(\blank, \emptyset)$};
		\node (xO) at ($(P) + (0,3.5)$) {$(\setnocond{x}, \emptyset)$};
		\node (yz1) at ($(P) + (0,2.5)$) {$(\setnocond{y,z}, f_{1})$};
		\node (x2) at ($(P) + (2,1.5)$) {$(\setnocond{x}, f_{2})$};
		\node (yz4) at ($(P) + (0,0.5)$) {$(\setnocond{y,z}, f_{4})$};
		\node (x3) at ($(P) + (-2,1.5)$) {$(\setnocond{x}, f_{3})$};
		
		\draw (blank) edge[loop above] node {$a,77$} (blank);
		\draw (blank) to[in=30, out=60, distance=5mm] node {$b,77$} (blank);
		\draw (blank) edge[loop right, distance=5mm] node {$c,77$} (blank);

		\draw ($(xO) + (0,0.5)$) to node {} ($(xO) + (0,0.2)$);
		\draw (xO) to node[left] {$a,77$} (yz1);
		\draw[dotted] (xO) to[bend left=15] node[above] {$b,77$} (blank);
		\draw[dotted] (xO) to[bend right=15] node[above] {$c,77$} (blank);
		
		\draw (yz1) to[bend left=15] node[right] {~~$b,3$} (x2);
		\draw[dotted] (yz1) to[bend right=15] node[above, near start] {$a,1$} (blank);

		\draw (x2) to[bend left=15] node[below] {$a,77$~~~~~~} (yz1);
		\draw[dotted] (x2) to[bend right=15] node[left] {$b,1$} (blank);
		\draw[dotted] (x2) to[bend right=30] node[right] {$c,1$} (blank);
		
		\draw (yz1) to node[below] {~~$c,2$} (x3);

		\draw (x3) to[bend left=15] node[above] {~~$a,77$} (yz4);
		\draw[dotted] (x3) .. controls ($(x3) + (2,4)$) and ($(blank) + (-0.7,0.8)$) .. node[left, near start] {$a,1$} (blank);

		\draw (yz4) to[bend left=15] node[below] {$c,3$~~~~} (x3);
		\draw (yz4) to node[above] {$b,2$~~} (x2);
		\draw[dotted] (yz4) .. controls ($(yz4) + (3.75,0.5)$) and ($(blank) + (1.25,-0.8)$) .. node[below, near start] {$a,1$} (blank);
		
		\node (table) at (7.5,2.5) {%
			\setlength{\tabcolsep}{2pt}
			\begin{tabular}{c|cccc}
				$j$ & $f_{1}(j)$ & $f_{2}(j)$ & $f_{3}(j)$ & $f_{4}(j)$\\
				\hline
				1 & $\setnocond{y,z},1,\setnocond{y}$ & $\setnocond{x},1,\emptyset$ & $\setnocond{x},2,\emptyset$ & $\setnocond{y,z},2,\emptyset$ \\
				2 & $\setnocond{y,z},2,\emptyset$ & $\setnocond{x},2,\emptyset$ & $\setnocond{x},1,\emptyset$ & $\setnocond{y,z},1,\setnocond{y}$ \\
				3 & $\setnocond{y,z},2,\setnocond{y}$ & & & $\setnocond{y,z},2,\setnocond{y}$ \\
				4 & $\setnocond{y,z},1,\emptyset$ & & & $\setnocond{y,z},1,\emptyset$ \\
				5 & $\setnocond{y,z},1,\setnocond{z}$ & & & $\setnocond{y,z},1,\setnocond{z}$ \\
				6 & $\setnocond{y,z},2,\setnocond{z}$ & & & $\setnocond{y,z},2,\setnocond{z}$ \\
				7 & $\setnocond{y},1,\emptyset$ & & & $\setnocond{y},1,\emptyset$ \\
				8 & $\setnocond{y},2,\emptyset$ & & & $\setnocond{y},2,\emptyset$ \\
				9 & $\setnocond{z},1,\emptyset$ & & & $\setnocond{z},1,\emptyset$ \\
				10 & $\setnocond{z},2,\emptyset$ & & & $\setnocond{z},2,\emptyset$ \\
			\end{tabular}
		};
	\end{tikzpicture}
	\caption{Parity automaton $\daut_{\calE}$ corresponding to $\sdaut_{\calE}$ in Figure~\ref{fig:exampleSemiDet}}
	\label{fig:exampleParity}
\end{figure}
Figure~\ref{fig:exampleParity} shows the parity automaton $\daut_{\calE}$ obtained by applying the above determinisation to the semi-deterministic automaton $\bToSd(\baut_{\calE})$ depicted in Figure~\ref{fig:exampleSemiDet}.
Function $f_{1}$ is completely arbitrary since functions $g$, $g'$, and $g''$ are all the empty function, so let us detail how to obtain the transition from $(\setnocond{y,z},f_{1})$ to $(\setnocond{x},f_{2})$ via action $b$.
Table~\ref{tab:exampleParityFunctions} shows the functions $g$, $g'$, $g''$, and $f'$ we compute and whether $(f(j), b, g(j))$ is accepting (i.e., $(f(j), b, g(j)) \in \sdfinal$).
As we can see, for the transition from $(\setnocond{y,z},f_{1})$ to $(\setnocond{x},f_{2})$ via action $b$ we have that both $a$ and $d$ have value $2$ since $(f_{1}(2), b, g(2)) \in \sdfinal$ and $g'(2) = \blank$, so the resulting transition has priority $2d - 1 = 3$ as $d \leq a$.
Instead, for the transition from $(\setnocond{x},f_{2})$ to $(\setnocond{y,z},f_{1})$ via action $a$, we have that both $a$ and $d$ have value $|\sdstatesfin| + 1 = 39$ since there is no accepting transition and no blank in $g'(\functionGenericArgument)$, so the resulting transition has priority $2d - 1 = 77$.
Note that in $f'(\functionGenericArgument)$ only positions $1$ and $2$ are determined by $g''(\functionGenericArgument)$;
the remaining positions are again arbitrary and having $f' = f_{1}$ is the result of a deliberate choice.
In fact, a different choice would just make the resulting parity automaton larger than $\daut_{\calE}$ while accepting the same language.

\begin{table}
	\caption{Examples of the construction of the transitions of $\daut_{\calE}$ in
	Figure~\ref{fig:exampleParity}}
	\label{tab:exampleParityFunctions}
	\centering
	\small
	\setlength{\tabcolsep}{4pt}
\resizebox{\textwidth}{!}{%
	\begin{tabular}{c|ccccc|ccccc}
		& \multicolumn{5}{c|}{Transition $((\setnocond{y,z},f_{1}), b, (\setnocond{x}, f_{2}))$} & \multicolumn{5}{c}{Transition $((\setnocond{x},f_{2}), a, (\setnocond{y,z}, f_{1}))$} \\
		$j$ & $g(j)$ & acc. & $g'(j)$ & $g''(j)$ & $f'(j)$ & $g(j)$ & acc. & $g'(j)$ & $g''(j)$ & $f'(j)$ \\
		\hline
		1 & $\setnocond{x},1,\emptyset$ & no & $\setnocond{x},1,\emptyset$ & $\setnocond{x},1,\emptyset$ & $\setnocond{x},1,\emptyset$ & $\setnocond{y,z},1,\setnocond{y}$ & no & $\setnocond{y,z},1,\setnocond{y}$ & $\setnocond{y,z},1,\setnocond{y}$ & $\setnocond{y,z},1,\setnocond{y}$ \\
		2 & $\setnocond{x},1,\emptyset$ & yes & $\blank$ & $\setnocond{x},2,\emptyset$ & $\setnocond{x},2,\emptyset$ & $\setnocond{y,z},2,\emptyset$ & no & $\setnocond{y,z},2,\emptyset$ & $\setnocond{y,z},2,\emptyset$ & $\setnocond{y,z},2,\emptyset$ \\
		3 & $\setnocond{x},1,\emptyset$ & yes & $\blank$ & & & & & & & $\setnocond{y,z},2,\setnocond{y}$ \\
		4 & $\setnocond{x},1,\emptyset$ & no & $\blank$ & & & & & & & $\setnocond{y,z},1,\emptyset$ \\
		5 & $\setnocond{x},2,\emptyset$ & yes & $\setnocond{x},2,\emptyset$ & & & & & & & $\setnocond{y,z},1,\setnocond{z}$ \\
		6 & $\setnocond{x},1,\emptyset$ & yes & $\blank$ & & & & & & & $\setnocond{y,z},2,\setnocond{z}$ \\
		7 & $\blank$ & no & $\blank$ & & & & & & & $\setnocond{y},1,\emptyset$ \\
		8 & $\blank$ & no & $\blank$ & & & & & & & $\setnocond{y},2,\emptyset$ \\
		9 & $\setnocond{x},1,\emptyset$ & no & $\blank$ & & & & & & & $\setnocond{z},1,\emptyset$ \\
		10 & $\setnocond{x},1,\emptyset$ & yes & $\blank$ & & & & & & & $\setnocond{z},2,\emptyset$
	\end{tabular}
}
\end{table}

Consider a word $\word \in \setnocond{a,b,c}^{\omega}$ and the associated run $\arun$:
if $\word \in \setnocond{a,b,c}^{\omega} \setminus (ab|ac)^{\omega}$, then the corresponding run of $\daut_{\calE}$ has $77$ as limiting minimum priority since $\word \notin (ab|ac)^{\omega}$ means that there exists $i \in \omega$ such that either $\word(i) = \word(i+1) = a$, or $\word(i) = b$ and $\word(i+1) = c$, or $\word(i) = c$ and $\word(i+1) = b$, thus the state $(\blank, \emptyset)$ is reached via the transition $\tran{\arun}(i+1)$.
Since $(\blank, \emptyset)$ enables only self-loops each one with priority $77$, this is also the minimum priority appearing infinitely often.
Now, suppose that $\word \in (ab|ac)^{\omega} \setminus \lang(\baut)$.
This means that either $\word \in (ab|ac)^{*}(ab)^{\omega}$, or $\word \in (ab|ac)^{*}(ac)^{\omega}$;
in the former case, the automaton repeatedly switches between states $(\setnocond{x},f_{2})$ and $(\setnocond{y,z}, f_{1})$, and in the latter case between states $(\setnocond{x},f_{3})$ and $(\setnocond{y,z}, f_{4})$.
In both cases, it is immediate to see that the minimum priority appearing infinitely often is $3$ that is odd, thus  $\alpha$ is rejected.

\subsection{Language Equivalence}
\label{ssec:languageEquivalenceSemiDet}

The semi-deterministic construction we presented in Definition~\ref{def:semi-determinisation} preserves the accepted language, that is, a NGBA $\baut$ and the resulting semi-deterministic automaton $\bToSd(\baut)$ accept the same language;
moreover, the language accepted by $\bToSd(\baut)$ starting from a state $(R, j, C) \in \sdstatesfin$ depends only on $R$, the subset component.

\begin{proposition}
\label{pro:buechiLangEqualSemiDetLang}
	Given a NGBA $\baut$, let $\sdaut$ be constructed as above.
	Then, $\lang(\sdaut) = \lang(\baut)$.
\end{proposition}

\begin{proposition}
\label{pro:semiDetLanguageIgnoreBandi}
	Given a NGBA $\baut$ and two states $(R, j, C),$ $(R, j', C') \in \sdstatesfin$ of $\bToSd(\baut)$, $\lang(\bToSd(\baut)_{(R, j, C)}) = \lang(\bToSd(\baut)_{(R, j', C')})$.
\end{proposition}

Similarly, for a given NGBA $\baut$, also $\sdaut = \bToSd(\baut)$ and the corresponding parity automaton $\daut = \sdToDet(\sdaut)$ are language equivalent, thus $\lang(\baut) = \lang(\sdaut) = \lang(\daut)$.

\begin{proposition}
\label{pro:semiDetLangEqParityLang}
	Given a SDBA $\sdaut$ and $\daut = \sdToDet(\sdaut)$, $\lang(\sdaut) = \lang(\daut)$ holds.
\end{proposition}

Given a semi-deterministic \buchi automaton $\sdaut$, $\daut = \sdToDet(\sdaut)$, and a state $(r,f)$ of $\daut$, we remark that for $i \in \omega$ we have $f(i) \in \sdstatesfin$.
Since $\sdaut$ is semi-deterministic, by Definition~\ref{def:semiDetAut} the reachable fragment of $\sdaut_{f(i)}$ is a deterministic automaton so we can consider the product $\mc \prodAut \sdaut_{f(i)}$ that is a MC extended with accepting conditions.
In particular, the accepting SCCs of $\mc \prodAut \daut$ and $\mc \prodAut \sdaut$ are strictly related by the function $f$ of states $(r,f) \in \daut$ and the smallest priority occurring in the considered SCC.
In the following, we say that $\mc \prodAut \sdaut_{(m,q)}$ (or $\mc \prodAut \daut_{(m,(r,f))}$) is accepting if the probability to eventually being trapped into an accepting SCC is $1$.
\begin{lemma}
\label{lem:parityOfSccAndAcceptingSccOfSemidet}
	Given a SDBA $\sdaut$ and $\daut = \sdToDet(\sdaut)$, if $\mc \prodAut \daut_{(m,(r,f))}$ forms an SCC where the smallest priority of the transitions in the SCC is $2a$, then $\mc \prodAut \sdaut_{(m,f(a))}$ is accepting.
\end{lemma}

It is known by Lemma~\ref{lem:productSubsetIsomorphicQuotientMC} that $\mc \prodMC \ssaut$ and $\mc \prodMC \daut$ are strictly related, so we can define the accepting SCC of $\mc \prodMC \ssaut$ by means of the accepting states of $\mc \prodMC \daut$.
\begin{definition}
\label{def:acceptingSccOfSubsetApp}
	Given a MC $\mc$ and a NGBA $\baut$, for $\ssaut = \bToSub(\baut)$ and $\daut = \det(\baut)$, we say that a bottom SCC $\scc$ of $\mc \prodMC \ssaut$ is accepting if, and only if, there exists a state $s = (m, (R, f))$ in an accepting bottom SCC $\scc'$ of $\mc \prodMC \daut$ such that $(m, R) \in \scc$.
\end{definition}
Note that Corollary~\ref{cor:same} ensures that the accepting SCCs of $\mc \prodMC \ssaut$ are well defined.

\begin{theorem}
\label{thm:acceptanceOfSubsetAutomatonApp}
	Given a MC $\mc$ and a NGBA $\baut$, for $\sdaut = \bToSd(\baut)$ and  $\ssaut = \bToSub(\baut)$, the following facts are equivalent:
	\begin{enumerate}
	\item
		$\scc$ is an accepting bottom SCC of $\mc \prodMC \ssaut$;
	\item
		there exist $(m,R) \in \scc$ and $R' \subseteq R$ such that $(m, (R', j, \emptyset))$ belongs to an accepting SCC of $\mc \prodAut \sdaut_{(m, (R', j, \emptyset))}$ for some $j \in \natIntK$;
	\item
		there exist $(m,R) \in \scc$ and $q \in R$ such that $(m,(\setnocond{q}, 1, \emptyset))$ reaches with probability $1$ an accepting SCC of $\mc \prodAut \sdaut_{(m, (\setnocond{q}, 1, \emptyset))}$.
	\end{enumerate}
\end{theorem}

Theorem~\ref{thm:acceptanceOfSubsetAutomatonApp} provides a practical way to check whether an SCC $\scc$ of $\mc \prodMC \ssaut$ is accepting:
it is enough to check whether some state $(m, R)$ of $\scc$ has $R \supseteq R'$ for some $(m, (R', j, \emptyset))$ in the accepting region of $\mc \prodAut \sdaut$, or whether, for a state $q \in R$, $(m, (\setnocond{q}, 1, \emptyset))$ reaches with probability $1$ the accepting region.
We remark that, by construction of $\sdaut$, if we change the initial state of $\sdaut$ to $(R, j ,C)$---i.e., if we consider $\sdaut_{(R, j, C)}$---then the run can only visit breakpoint states; i.e., it is actually a run of $\bbpaut_{(R, j, C)}$.

\begin{theorem}
\label{thm:probabilityViaSubsetOnlyMC}
	Given a MC $\mc$ and a \buchi automaton $\baut$, consider $\ssaut = \bToSub(\baut)$ and $\mc \prodMC \ssaut$ with accepting SCCs according to Definition~\ref{def:acceptingSccOfSubset}.
	Let $U$ be the accepting region and let $\Diamond U$ denote the set of paths containing a state of $U$.
	Then, $\prob^{\mc}(\baut) = \prob^{\mc \prodMC \ssaut}(\Diamond U)$.
\end{theorem}

\subsection{Complexity of Semi-Determinisation}

The complexity of the procedure $\ComputeAcceptingScc(\mc, \baut)$ in Figure~\ref{fig:ComputeAcceptingScc} is $\bigO(m^{2} \cdot k \cdot 3^{n})$, where $n$ is the number of states of the \buchi automaton $\baut$, $k$ the number of accepting sets in $\abuechiset$, and $m$ the number of states of the MC $\mc$.
Note that the actual runtime can be improved by caching the positive results of $\IsAccepting(\mc, \baut, \scc)$:
if $\scc$ is accepting with witness $(m,q)$ with $q \in R$ and $(m,R) \in \scc$, and we have to compute $\IsAccepting(\mc, \baut, \scc')$, we can first verify whether there exists $(m,R') \in \scc'$ such that $q \in R'$;
if this is the case, we can immediately return a positive answer without constructing the \buchi breakpoint automaton.

\section{Markov Decision Processes}
\label{sec:mdp}

The lazy determinisation approach proposed in this paper extends to Markov decision processes (MDPs) after minor adaptation;
Markov chains have mainly been used for ease of notation.
While the details of the extension to MDPs have been moved to Appendix~\ref{app:MDP}, we give here an outline of the adaptation with a focus on the differences and particularities that need to be taken into consideration when we are dealing with MDPs.

An MDP is a tuple $\mdp = (\pstates, \labelFunc, \Aset, \pinit, \pmat)$ where $\pstates$, $\labelFunc$, and $\pinit$ are as for Markov chains, $\Aset$ is a finite set of actions, and $\pmat \colon \pstates \times \Aset \to \dist(\pstates)$ is the transition probability function where $\dist(\pstates)$ is the set of distributions over $\pstates$.
The nondeterministic choices are resolved by a scheduler $\sched$ that chooses the next action to be executed depending on a finite path.
Like for Markov chains, the principal technique to analyse MDPs against a specification $\phi$ is to construct a deterministic Rabin automaton $\aut$, build the product $\mdp \prodAut \aut$, and analyse it.
This product will be referred to as a \emph{Rabin MDP} (RMDP).
According to~\cite{BiancoA95}, for a RMDP, it suffices to consider memoryless deterministic schedulers of the form $\sched \colon \pstates \times \astates \to \Aset$, where $\astates$ is the set of states of $\aut$.
Given a NGBA specification $\baut_{\phi}$, we are interested in $\sup_{\sched} \prob^{\mdp, \sched}(\baut_{\phi})$.
In particular, one can use finite memory schedulers on $\mdp$.
(Schedulers that control $\mdp$ can be used to control $\mdp \prodAut \aut$ for all deterministic automata $\aut$.)
The superscript $\mdp$ is omitted when it is clear from the context.
We remark that the infimum can be treated accordingly, as $\inf_{\sched} \prob^{\sched}(\baut_{\phi}) = 1 - \sup_{\sched} \prob^{\sched}(\baut_{\neg \phi})$.

As Theorem~\ref{thm:only_reach} operates on the word level, Corollary \ref{cor:same} immediately extends to MDPs.
Under the corresponding equivalence relation we obtain a \emph{quotient MDP}.
From here, it is clear that we can use the estimation of the word languages provided in Theorem~\ref{theo:inclusions} to estimate $\sup_{\sched} \prob^{\sched}(\baut_{\phi})$.

\begin{corollary}
\label{cor:MDPestimate}
	Given an MDP $\mdp$ and a NGBA $\baut$, let $m$ be a state of $\mdp$ and $d,d'$ be states of $\aut = \det(\baut)$ with $[d]=[d']$.
	Then
	$\sup_{\sched}\prob_{(m,[d])}^{\sched}(\ssaut^{u}_{[d]})
	\leq
	\sup_{\sched}\prob_{(m,\langle d \rangle)}^{\sched}(\rbpaut^{u}_{\langle d \rangle})
	\leq
	\sup_{\sched}\prob_{(m,d)}^{\sched}(\aut_{d})
	=                                                   
	\sup_{\sched}\prob_{(m,d')}^{\sched}(\aut_{d'})     
	\leq
	\sup_{\sched}\prob_{(m,\langle d \rangle)}^{\sched}(\rbpaut^{o}_{\langle d \rangle}),
	\ \sup_{\sched}\prob_{(m,[d])}^{\sched}(\ssaut^{o}_{[d]})$ holds.
\end{corollary}

In the standard evaluation of RMDP, the \emph{end components} of the product $\mdp \prodAut \aut$ play a role comparable to the one played by bottom SCCs in MCs.
An end component (EC) is simply a sub-MDP, which is closed in the sense that there exists a memoryless scheduler $\sched$ such that the induced Markov chain is a bottom SCC.
If there is a scheduler that additionally guarantees that a run that contains all possible transitions infinitely often is accepting, then the EC is \emph{accepting}. 
Thus, one can stay in the EC and traverse all of its transitions (that the scheduler allows) infinitely often, where acceptance is defined as for BSCCs in MCs.

\begin{theorem}
\label{thm:acceptanceOfSubsetAutomatonEC}
	Given an MDP $\mdp$ and a NGBA $\baut$, for $\aut = \det(\baut)$,  $\sdaut = \bToSd(\baut)$, and $\ssaut = \bToSub(\baut)$, if $\mec$ is an accepting EC of $\mdp {\prodAut} \aut$, then
	\begin{inparaenum}[(1.)]
	\item $[\mec]$ is an EC of $\mdp {\prodMDP} \ssaut$ and
	\item $\mec' = \langle \mec \rangle$ is an accepting EC of $\mdp {\prodAut} \sdaut$. $\mec'$ contains a state $(m,(R, 1, \emptyset))$ with $R \subseteq [d]$ and $(m,d) \in \mec$.
	\end{inparaenum}
\end{theorem}
Note that, since each EC $\mec$ of $\mdp \prodAut \aut$ is either accepting or rejecting, finding an accepting EC $\mec' = \langle \mec \rangle$ of $\mdp \prodAut \sdaut$ allows us to derive that $\mec$ is accepting as well.

For RMDPs, it suffices to analyse maximal end components (MEC). 
We define a MEC as accepting if it contains an accepting EC.
MECs are easy to construct and, for each accepting pair, they are easy to evaluate:
it suffices to remove the rejecting transitions, repeat the construction of MECs on the remainder, and check if there is any that contains an accepting transition.
Once accepting MECs are determined, their states are assigned a winning probability of $1$, and evaluating the complete MDP reduces to a maximal reachability analysis, which 
reduces to solving an LP problem. 
It can therefore be solved in polynomial~time.

These two theorems allow us to use a layered approach of lazy determinisation for MDPs, which is rather similar to the one described for Markov chains.
We start with the quotient MDP, and consider an arbitrary MEC $\mec$.
By using the accepting conditions of the subset automata $\ssaut^{u}$ and $\ssaut^{o}$, we check whether $\mec$ is accepting or rejecting, respectively.
If this test is inconclusive, we first refine $\mec$ by a breakpoint construction, and finally by a multi-breakpoint construction.
We remark that, as for Markov chains, the breakpoint and multi-breakpoint constructions can be considered as oracles:
when we have identified the accepting MECs, a plain reachability analysis is performed on the quotient MDP.

Theorem~\ref{thm:acceptanceOfSubsetAutomatonEC} makes clear what needs to be calculated in order to classify an EC---and thus a MEC---as accepting, while Corollary~\ref{cor:MDPestimate} allows for applying this observation in the quantitative analysis of an MDP, and also to smoothly combine this style of reasoning with the lazy approach.
This completes the picture of~\cite{Courcoubetis+Yannakakis/95/Markov} for the quantitative analysis of MDPs, which is technically the same as their analysis of concurrent probabilistic programs~\cite{Courcoubetis+Yannakakis/95/Markov}.

In general, it is possible to compute $\inf_{\sched} \prob^{\mdp, \sched}(\baut_{\neg\phi})$ instead of $\sup_{\sched} \prob^{\mdp, \sched}(\baut_{\phi})$, and then use $\sup_{\sched} \prob^{\mdp, \sched}(\baut_{\phi}) = 1 - \inf_{\sched} \prob^{\mdp, \sched}(\baut_{\neg\phi})$.
For computing the infimum, we would have to invert the terms: an EC is rejecting, if it is not accepting, and a MEC is rejecting, if it contains a rejecting EC.
This detour has a principle computational disadvantage: a witness that an EC is rejecting is more involved than a witness that it is accepting, as it requires to try all combinations of subsets (see below).

We have not implemented this detour.
However, this approach may be worthwhile trying when $\baut_{\neg\phi}$ appears to be better suited for analysis.
This could, e.g., be the case, when  $\baut_{\neg\phi}$ is deterministic, or when  $\baut_{\neg\phi})$ is significantly smaller than $\baut_{\phi}$.
Given that the construction of $\baut_{\phi}$ and $\baut_{\neg\phi}$ are cheap in practice, it might be worth computing both.

It is worthwhile to point out that, in principle, the qualitative analysis from~\cite{Courcoubetis+Yannakakis/95/Markov} could replace Theorem~\ref{thm:acceptanceOfSubsetAutomatonEC} when using this detour.
This would imply accepting the computational drawbacks and losing the choice between the two automata to start with.
The techniques from Section 4.2 of~\cite{Courcoubetis+Yannakakis/95/Markov} also do not directly allow for analysing MECs only, and some further work would have to be invested to allow for focussing on rejecting MECs, and to facilitate it when using the detour through $\baut_{\neg\phi}$.

As said before, from a theoretical point of view, to compute $\sup_{\sched} \prob^{\mdp, \sched}(\baut_{\phi})$ it is enough to consider the maximal end components of $\mdp \prodAut \det(\baut_{\phi})$ instead of the bottom SCCs as in the Markov chain case and then compute the maximal probability to reach the accepting MECs.
From the practical point of view, the algorithms $\ComputeAcceptingScc$ and $\IsAccepting$ shown in Figure~\ref{fig:ComputeAcceptingScc} can be easily adapted to MDPs as follows:
in $\ComputeAcceptingScc$, at line 4 $\ComputeScc$ is replaced by $\ComputeMec$ that computes the MECs of the product between the MDP $\mdp$ and the subset automaton $\ssaut$;
in $\IsAccepting(\mdp, \baut, \mec)$, line 1 has to be replaced by a loop on all $(m,R) \in \mec$ since checking the acceptance from a single state of the MEC $\mec$ does not suffice.
However, if $(m,R)$ is known to be not accepting, we can exclude all states that cannot avoid reaching $(m,R)$, thus the breakpoint construction can be performed on a reduced number of states.

\begin{figure}
	\centering
	\begin{tikzpicture}[->,>=stealth,auto,xscale=0.8]
		\path[use as bounding box] (-8.2,4.6) rectangle (8.4,0);
		\node (mdp) at (-6,3) {\begin{minipage}{3.5cm}\centering$\mdp$\\$\forall m \in \pstates.\,\lang(m) = m$\end{minipage}};
		\node (mdpc) at ($(mdp) + (0,0.7)$) {};
		\node (mdp_b) at ($(mdpc) + (-1.5,0)$) {$b$};
		\node (mdp_c) at ($(mdpc) + (0,0.75)$) {$c$};
		\node (mdp_a) at ($(mdpc) + (1.5,0)$) {$a$};
		
		\draw ($(mdp_b.north) + (0,0.3)$) to (mdp_b.north);
		\draw (mdp_b) to node[above, very near end] {\scriptsize$1$} (mdp_c);
		\draw (mdp_c) to node[above, very near end] {\scriptsize$1$} (mdp_a);
		\draw (mdp_a) to node[below, very near end] {\scriptsize$1$} (mdp_b);
		\draw (mdp_a) edge [loop right, distance=7mm] node [below, very near end] {\scriptsize$1$} (mdp_a);

		\node (baut) at (-0.25,3) {$\baut$};
		\node (bautc) at ($(baut) + (0,1)$) {};
		\node (baut_x) at ($(bautc) + (-1,0)$) {$x$};
		\node (baut_y) at ($(bautc) + (1,0)$) {$y$};
		
		\draw ($(baut_x.north) + (0,0.3)$) to (baut_x.north);
		\draw (baut_x) to node {$a$} (baut_y);
		\draw (baut_x) edge [loop left, distance=7mm] node [left] {$a$} (baut_x);
		\draw (baut_x) to[in=30, out=60, distance=7mm] node[right] {$b$} (baut_x);
		\draw (baut_x) to[in=-50, out=-20, distance=7mm] node[right] {$c$} (baut_x);
		\draw (baut_y) edge [loop right, distance=7mm, double] node [right] {$a$} (baut_y);

		\node (mdpss) at (5.5,3) {$\mdp \prodMDP \ssaut$};
		\node (mdpssc) at ($(mdpss) + (0,0.7)$) {};
		\node (mdpss_b) at ($(mdpssc) + (-2,0)$) {$b,\setnocond{x}$};
		\node (mdpss_c) at ($(mdpssc) + (0,0.75)$) {$c,\setnocond{x}$};
		\node (mdpss_a) at ($(mdpssc) + (2,0)$) {$a,\setnocond{x,y}$};
		
		\draw ($(mdpss_b.north) + (0,0.3)$) to (mdpss_b.north);
		\draw (mdpss_b) to[above, very near end] node {\scriptsize$1$} (mdpss_c);
		\draw (mdpss_c) to[above, very near end] node {\scriptsize$1$} (mdpss_a);
		\draw (mdpss_a) to[below, very near end] node {\scriptsize$1$} (mdpss_b);
		\draw (mdpss_a) edge [loop above, distance=7mm] node [right, very near end] {\scriptsize$1$} (mdpss_a);

		\node (mdpbp) at (-1,-0.2) {};
		\node (mdpbpl) at ($(mdpbp) + (6.5,0.5)$) {$\mdp \prodAut \bbpaut$};
		\node (mdpbpc) at ($(mdpbp) + (0,1.2)$) {};
		
		\node (mdpbp_b_x_) at ($(mdpbpc) + (-6,0)$) {$b,(\setnocond{x}, 1, \emptyset)$};
		\node (mdpbp_c_x_) at ($(mdpbpc) + (-2.7,0)$) {$c,(\setnocond{x}, 1, \emptyset)$};
		\node (mdpbp_a_xy_) at ($(mdpbpc) + (0.8,0)$) {$a,(\setnocond{x,y}, 1, \emptyset)$};
		\node (mdpbp_a_xy_y) at ($(mdpbpc) + (4.7,0)$) {$a,(\setnocond{x,y}, 1, \setnocond{y})$};
		\node (mdpbp_a_y_) at ($(mdpbpc) + (8.25,0)$) {$a,(\setnocond{y}, 1, \emptyset)$};
		\node (mdpbp_a_x_) at ($(mdpbpc) + (-2.7,0.75)$) {$a,(\setnocond{x}, 1, \emptyset)$};
				
		\draw (mdpbp_b_x_) to node[above, very near end] {\scriptsize$1$} (mdpbp_c_x_);
		\draw (mdpbp_c_x_) to node[above, very near end] {\scriptsize$1$} (mdpbp_a_xy_);
		\draw (mdpbp_a_xy_) to node[above, very near end] {\scriptsize$1$} (mdpbp_a_xy_y);
		\draw (mdpbp_a_xy_) .. controls ($(mdpbp_a_xy_) + (-1,-0.75)$) and ($(mdpbp_b_x_) + (1,-0.75)$) .. node [below, very near end] {\scriptsize$1$} (mdpbp_b_x_);
		\draw (mdpbp_a_xy_y) edge [loop above, distance=7mm] node [right, very near end] {\scriptsize$1$} (mdpbp_a_xy_y);
		\draw[-] (mdpbp_a_xy_y) .. controls ($(mdpbp_a_xy_y) + (0,-1)$) and ($0.5*(mdpbp_a_xy_y) + 0.5*(mdpbp_b_x_) + (2,-1)$) .. ($0.5*(mdpbp_a_xy_y) + 0.5*(mdpbp_b_x_) + (0,-1)$); 
		\draw ($0.5*(mdpbp_a_xy_y) + 0.5*(mdpbp_b_x_) + (0,-1)$) .. controls ($0.5*(mdpbp_a_xy_y) + 0.5*(mdpbp_b_x_) + (-2,-1)$) and ($(mdpbp_b_x_) + (0,-1)$) .. node [below, very near end] {\scriptsize$1$} (mdpbp_b_x_);
		\draw (mdpbp_a_x_) to[bend right=10] node[above, very near end] {\scriptsize$1$} (mdpbp_b_x_);
		\draw (mdpbp_a_x_) to[bend left=10] node[right, very near end] {\scriptsize$1$} (mdpbp_a_xy_);
		\draw (mdpbp_a_y_) edge [loop above, distance=7mm,double] node [right, very near end] {\scriptsize$1$} (mdpbp_a_y_);
	\end{tikzpicture}
	\caption{Finding accepting ECs in MDPs: MDP $\mdp$, NGBA $\baut$, $\ssaut = \bToSub(\baut)$, $\bbpaut = \bToBp(\baut)$}
	\label{fig:exampleMDPs}
\end{figure}
For MDPs, differently from the subset and breakpoint construction, for the multi-breakpoint case testing only one $(m, R) \in \mec$ in general is not sufficient; 
consider the MDP $\mdp$ and the NGBA $\baut$ depicted in Figure~\ref{fig:exampleMDPs}.
We first consider the product MDP $\mdp \prodMDP \ssaut$, containing one MEC.
We first try to decide whether it is accepting by considering the state $(c,\setnocond{x})$.
The only nonempty subset of $\setnocond{x}$ is the set itself, thus we look for accepting MECs in $\mdp \prodAut \bbpaut_{(c,(\setnocond{x}, 1, \emptyset))}$.
It is clear that from $(c,(\setnocond{x}, 1, \emptyset))$ no accepting MECs can be reached.
In contrast to the MC setting, we cannot conclude that the original MEC is not accepting.
Instead, we remove $(c,\setnocond{x})$ from the set of states to consider, as well as $(b,\setnocond{x})$, from which we cannot avoid reaching $(c,\setnocond{x})$.
The state left to try is $(a,\setnocond{x,y})$, where we have two transitions available.
Indeed, in $\mdp \prodAut \bbpaut$ the singleton MEC $\setnocond{(a,(\setnocond{y}, 1, \emptyset))}$ is accepting.
Thus the MEC of $\mdp \prodMDP \ssaut$ is accepting, though only one of its states---$\setnocond{(a,\setnocond{x,y})}$---allows us to conclude this, and we need to select the correct subset, $\setnocond{y}$, to start with.

\section{Implementation and Results}
\label{sec:experiments}

We have implemented our approach in our \iscasmc tool~\cite{iscasmc} in both explicit and BDD-based symbolic versions.
We use LTL formulas to specify properties, and apply \spot~\cite{DuretLutz11} to translate them to NGBAs.
Our experimental results suggest that our technique provides a practical approach for checking LTL properties for probabilistic systems.
A web interface to \iscasmc can be found at \url{http://iscasmc.ios.ac.cn/}.
For our experiments, we used a 3.6 GHz Intel Core i7-4790 with 16GB 1600 MHz DDR3 RAM.

We consider a model~\cite{ClothH05} of a distributed file server system used by Google in the \prism model version of~\cite{BaierHHHK13}.
This model is a continuous-time Markov chain, but we can apply our methods on its embedded (discrete time) MC.

\begin{figure}[t]
	\centering
	\begin{tikzpicture}
	\begin{axis}
	[
	width=60mm,
	height=50mm,
	title={\begin{minipage}{29mm}Average runtimes (Google file server)\end{minipage}}
	]
	\addplot[
	point meta=explicit symbolic,
	blue,
	mark=+
	]
	table[meta=label] {
	x y label
	5 6.55 prism
	6 6.6 prism
	7 9.74 prism
	8 3.96 prism
	9 9.75 prism
	10 12.99 prism
	11 15.87 prism
	12 28.07 prism
	13 10.19 prism
	14 28.25 prism
	15 40.74 prism
	16 52.08 prism
	17 54.92 prism
	18 64.53 prism
	19 76.59 prism
	20 79.71 prism
	21 115.09 prism
	22 93.9 prism
	23 107.1 prism
	24 140.73 prism
	25 150.28 prism
	26 133.26 prism
	27 174.03 prism
	28 170.21 prism
	29 164.62 prism
	30 193.98 prism
	};
	\addplot[
	point meta=explicit symbolic,
	red,
	mark=x
	]
	table[meta=label] {
	x y label
	5 1.64 iscasmc
	6 7.64 iscasmc
	7 1.64 iscasmc
	8 1.77 iscasmc
	9 1.84 iscasmc
	10 1.88 iscasmc
	11 1.86 iscasmc
	12 5.05 iscasmc
	13 1.99 iscasmc
	14 2.13 iscasmc
	15 2.12 iscasmc
	16 2.04 iscasmc
	17 2.68 iscasmc
	18 2.32 iscasmc
	19 6.36 iscasmc
	20 8.21 iscasmc
	21 2.54 iscasmc
	22 2.94 iscasmc
	23 6.34 iscasmc
	24 8.76 iscasmc
	25 10.3 iscasmc
	26 4.95 iscasmc
	27 8.29 iscasmc
	28 11.67 iscasmc
	29 9.11 iscasmc
	30 12.54 iscasmc
	};
	\end{axis}
	\end{tikzpicture}
	\ \
	\begin{tikzpicture}
	\begin{axis}
	[
	width=60mm,
	height=50mm,
	title={\begin{minipage}{29mm}Average runtimes (Mutual exclusion)\end{minipage}}
	]
	\addplot[
	point meta=explicit symbolic,
	blue,
	mark=+
	]
	table[meta=label] {
	x y label
	5 11.51 prism
	6 18.28 prism
	7 27.07 prism
	8 25.22 prism
	9 22.84 prism
	10 26.2 prism
	11 26.22 prism
	12 34.28 prism
	13 33.26 prism
	14 41.3 prism
	15 50.48 prism
	16 65.57 prism
	17 86.71 prism
	18 86.62 prism
	19 86.42 prism
	20 84.3 prism
	21 108.63 prism
	22 102.33 prism
	23 153.44 prism
	24 134.03 prism
	25 175.07 prism
	26 128.92 prism
	27 180.07 prism
	28 160.48 prism
	29 182.4 prism
	30 190.28 prism
	};
	\addplot[
	point meta=explicit symbolic,
	red,
	mark=x
	]
	table[meta=label] {
	x y label
	5 2.19 iscasmc
	6 6.1 iscasmc
	7 9.74 iscasmc
	8 8.12 iscasmc
	9 2.6 iscasmc
	10 5.37 iscasmc
	11 5.9 iscasmc
	12 13.08 iscasmc
	13 6.99 iscasmc
	14 3.66 iscasmc
	15 8.15 iscasmc
	16 11.58 iscasmc
	17 8.67 iscasmc
	18 13.9 iscasmc
	19 19.4 iscasmc
	20 12.35 iscasmc
	21 9.9 iscasmc
	22 21.16 iscasmc
	23 25.68 iscasmc
	24 11.3 iscasmc
	25 15.51 iscasmc
	26 22.15 iscasmc
	27 27.05 iscasmc
	28 22.3 iscasmc
	29 22.93 iscasmc
	30 34.51 iscasmc
	};
	\end{axis}
	\end{tikzpicture}
	\ \
	\begin{tikzpicture}
		\node at (0,0) {\textcolor{blue}{$+$ \prism}};
		\node at (0,-0.5) {\textcolor{red}{$\times$ \iscasmc}};
		\node at (0,-2) {};
	\end{tikzpicture}
	\caption{Plots for random formulas}
	\label{fig:random}
\end{figure}
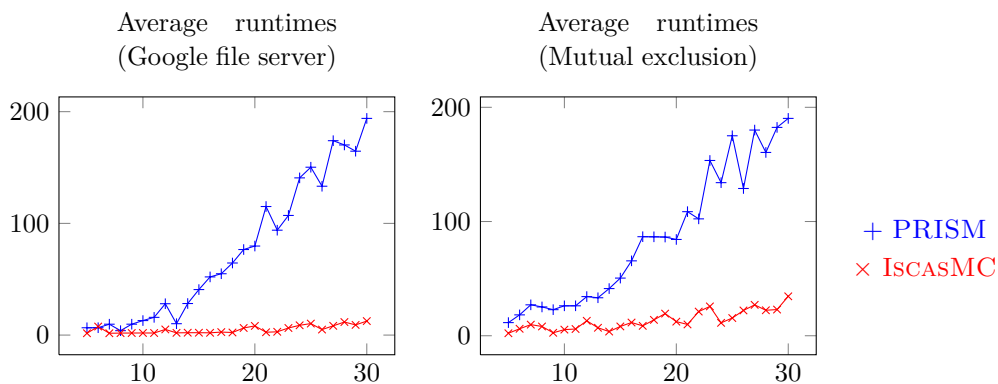

To estimate the scalability of our approach, we have applied it on increasingly larger formulas.
For this, we have applied \spot to generate 100 random formulas for each formula size from 2 to 30 based on the atomic propositions of~\cite{ClothH05}.
On these formulas we have then applied our model checker as well as \prism.
We used a timeout of 5 minutes and represent model checking runs which timed out as if they were performed in 5 minutes.
In the left part of Figure~\ref{fig:random} we plot the averages of the different runs for both tools (blue ``$+$'' for \prism and red ``$\times$'' for \iscasmc).
As shown, \iscasmc is very close to \prism for smaller formula sizes and its relative performance improves for larger ones.
Unfortunately, we were not able to compare our approach with \rabinizerthree~\cite{KomarkovaK14}, because the tool failed in the considered random formulas, such that we could not produce a plot. 
The problem seems to be caused by a preprocessing step on the formula before translating it to an automaton.

Next, we consider a mutual exclusion protocol~\cite{PZ86} (with four processes) which is an MDP.
Again, we consider random formulas of given lengths.
As seen in the right part of Figure~\ref{fig:random}, the general picture is similar to the previous case.

As a second case study, we consider a set of properties analysed previously in~\cite{ChatterjeeGK13}.
As there, we aborted tool runs when they took more than 30 minutes or needed more than 4GB of RAM.
The comparison with the results from~\cite{ChatterjeeGK13} cannot be completely accurate:
unfortunately, their implementation is not available on request to the authors, and for their results they did not state the exact speed of the machine used.
By comparing the runtimes stated for \prism in~\cite{ChatterjeeGK13} with the corresponding runtimes we obtained on our machine, we estimate that our machine is faster than theirs by about a factor of 1.6.
Thus, we have included the values from~\cite{ChatterjeeGK13} divided by 1.6 to take into account the estimated effect of the machine.
\begin{table}[t]
	\centering
	\caption{Runtime comparison for the randomised mutual exclusion protocol}\label{tab:mutual}
\resizebox{1\textwidth}{!}{
	\renewcommand{\tabcolsep}{2mm}
	\begin{tabular}{clddddddd}
	\toprule
	& & \multicolumn{7}{c}{time} \\
	\multicolumn{1}{c}{property} & \multicolumn{1}{c}{n} & \multicolumn{1}{c}{BP expl.} & \multicolumn{1}{c}{BP BDD} & \multicolumn{1}{c}{RB expl.} & \multicolumn{1}{c}{RB BDD} & \multicolumn{1}{c}{\prism} & \multicolumn{1}{c}{Rabinizer3} & \multicolumn{1}{c}{scaled \cite{ChatterjeeGK13}} \\
	\cmidrule(r{2pt}){1-2} \cmidrule(lr){3-9}
	\multirow{2}{6.9cm}{\begin{minipage}{6.2cm}$\ltlPmin ( \ltlGF p_1{=}10 \wedge \ltlGF p_2{=}10$\\\hspace*{0.5cm}$\wedge \ltlGF p_3{=}10 \wedge \ltlGF p_4{=}10 )$\end{minipage}\ (3)} & 4 & \textbf{3} & 5 & 15 & 28 & \multicolumn{1}{c}{--} & 104 & 23 \\
	& 5 & \textbf{19} & 21 & \multicolumn{1}{c}{--} & 104 & \multicolumn{1}{c}{--} & 1478 & 380 \\
	\midrule
	\multirow{3}{6.9cm}{\begin{minipage}{6.2cm}$\ltlPmax ( (\ltlGF p_1{=}0 \vee \ltlFG p_2{\neq} 0)$\\\hspace*{0.5cm}$\wedge (\ltlGF p_2{=}0 \vee \ltlFG p_3{\neq} 0) )$\end{minipage}\ (4)}& 3 & \textbf{1} & 2 & 2 & 4 & 138 & 2 & \textbf{1} \\
	& 4 & \textbf{3} & 7 & 4 & 15 & \multicolumn{1}{c}{--} & 20 & 18 \\
	& 5 & \textbf{19} & 32 & 35 & 76 & \multicolumn{1}{c}{--} & 319 & 299 \\
	\midrule
	\multirow{3}{6.9cm}{\begin{minipage}{6.2cm}$\ltlPmax ( (\ltlGF p_1{=}0 \vee \ltlFG p_1{\neq} 0)$\\\hspace*{0.5cm}$\wedge (\ltlGF p_2{=}0 \vee \ltlFG p_2{\neq} 0) )$\end{minipage}\ (5)}& 3 & 2 & 2 & 2 & 4 & 41 & 2 & \textbf{1}\\
	& 4 & \textbf{3} & 8 & 4 & 17 & 336 & 19 & 18\\
	& 5 & \textbf{19} & 34 & 45 & 68 & \multicolumn{1}{c}{--} & 314 & 289 \\
	\midrule
	\multirow{3}{6.9cm}{\begin{minipage}{6.2cm}$\ltlPmax ( (\ltlGF p_1{=}0 \vee \ltlFG p_2{\neq} 0)$\\\hspace*{0.5cm}$\wedge (\ltlGF p_2{=}0 \vee \ltlFG p_3{\neq} 0)$\\\hspace*{0.5cm}$\wedge (\ltlGF p_3{=}0 \vee \ltlFG p_1{\neq} 0) )$\end{minipage}\ (6)}& 3 & \textbf{1} & 2 & 2 & 6 & \multicolumn{1}{c}{--} & 5 & 4 \\
	& 4 & \textbf{3} & 9 & 7 & 27 & \multicolumn{1}{c}{--} & 52 & 47 \\
	& 5 & \textbf{29} & 38 & 99 & 124 & \multicolumn{1}{c}{--} & 871 & 762\\
	\midrule
	\multirow{3}{6.9cm}{\begin{minipage}{6.2cm}$\ltlPmax ( (\ltlGF p_1{=}0 \vee \ltlFG p_1{\neq} 0)$\\\hspace*{0.5cm}$\wedge (\ltlGF p_2{=}0 \vee \ltlFG p_2{\neq} 0)$\\\hspace*{0.5cm}$ \wedge (\ltlGF p_3{=}0 \vee \ltlFG p_3{\neq} 0) )$\end{minipage}\ (7)}& 3 & \textbf{1} & 2 & 2 & 9 & \multicolumn{1}{c}{--} & 5 & 5 \\
	& 4 & \textbf{3} & 9 & 12 & 41 & \multicolumn{1}{c}{--} & 50 & 49 \\
	& 5 & \textbf{29} & 38 & \multicolumn{1}{c}{--} & 171 & \multicolumn{1}{c}{--} & 849 & 792 \\
	\midrule
	\multirow{3}{6.9cm}{\begin{minipage}{6.2cm}$\ltlPmin ( (\ltlGF p_1{\neq} 10 \vee \ltlGF p_1{=}0 \vee \ltlFG p_1{=}1)$\\\hspace*{0.5cm}$\wedge \ltlGF p_1{\neq} 0 \wedge \ltlGF p_1{=}1 )$\end{minipage}\ (8)}& 3 & \textbf{1} & 2 & \textbf{1} & 3 & \textbf{1} & \textbf{1} & \textbf{1} \\
	& 4 & \textbf{3} & 6 & \textbf{3} & 10 & 8  & 13 & 6\\
	& 5 & \textbf{17} & 25 & \textbf{17} & 41 & 123 & 208 & 91  \\
	\midrule
	\multirow{3}{6.9cm}{\begin{minipage}{6.2cm}$\ltlPmax ( (\ltlG p_1{\neq} 10 \vee \ltlG p_2{\neq} 10 \vee \ltlG p_3{\neq} 10)$\\\hspace*{0.5cm}$ \wedge (\ltlFG p_1{\neq} 1 \vee \ltlGF p_2{=}1 \vee \ltlGF p_3{=}1)$\\\hspace*{0.5cm}$ \wedge (\ltlFG p_2{\neq} 1 \vee \ltlGF p_1{=}1 \vee \ltlGF p_3{=}1) )$\end{minipage}\ (9)}& 3 & \textbf{2} & 6 & \textbf{2} & 4 & \multicolumn{1}{c}{--} & 982 & 50 \\
	& 4 & 9 & 16 & \textbf{7} & 14 & \multicolumn{1}{c}{--} & 1718 & 440 \\
	& 5 & 136 & 60 & 91 & \textbf{56} & \multicolumn{1}{c}{--} & \multicolumn{1}{c}{--} & \multicolumn{1}{c}{--} \\
	\midrule
	\multirow{3}{6.9cm}{\begin{minipage}{6.2cm}$\ltlPmin ( (\ltlFG p_1{\neq} 0 \vee \ltlFG p_2{\neq} 0 \vee \ltlGF p_3{=}0)$\\\hspace*{0.5cm}$ \vee (\ltlFG p_1{\neq} 10 \wedge \ltlGF p_2{=}10 \wedge \ltlGF p_3{=}10)$\end{minipage}\ (10)}& 3 & \textbf{2} & 3 & \textbf{2} & 5 & 169 & 3 & \textbf{2}\\
	& 4 & 79 & 12 & \textbf{4} & 18 & \multicolumn{1}{c}{--} & 32 & 21 \\
	& 5 & \multicolumn{1}{c}{--} & 48 & \textbf{44} & 69 & \multicolumn{1}{c}{--} & 480 & 339 \\
	\bottomrule
	\multicolumn{9}{p{19cm}}{Note: The entries in column ``scaled \cite{ChatterjeeGK13}'' are the runtimes from \cite{ChatterjeeGK13} divided by $1.6$ as we used an estimated $1.6$ times faster machine.}
	\end{tabular}
}
\end{table}
In Table~\ref{tab:mutual} we provide the results obtained.
Here, ``property'' and ``n'' are as in~\cite{ChatterjeeGK13} and depict the property and the size of the model under consideration.
We report the total runtime in seconds (``time'') for the explicit-state (``BP expl.'') and the BDD-based symbolic (``BP BDD'') implementations of the multi-breakpoint construction, as well as the explicit and symbolic (``RB expl.'', ``RB BDD'') of the Rabin-based implementation.
In both BP and RB cases, we first apply the subset and breakpoint steps.
We also include the runtimes of \prism (``\prism'') and of the tool used in~\cite{ChatterjeeGK13} (``scaled~\cite{ChatterjeeGK13}'') developed for a subclass of LTL formulas and its generalisation to full LTL~\cite{EsparzaK14} implemented in \rabinizerthree~\cite{KomarkovaK14} (``\rabinizerthree'', for which we thank the authors for providing the source code).
We mark the best running times (obtained by rounding the actual times) with bold font.

The runtime of our new approaches is almost always better than the running time of other methods.
In many cases, the multi-breakpoint approach performs better than new the Rabin-based approach (restricted to the single undecided end component), but not always.
Broadly speaking, this can happen when the breakpoint construction has to consider many subsets as starting points for one end component, while the Rabin determinisation does not lead to a significant overhead compared to the breakpoint construction. 
(E.g., when the history trees have at most three or four nodes.)
Thus, both methods are of value.
Both of them are faster than the specialised algorithm of~\cite{ChatterjeeGK13} and \rabinizerthree.
We assume that one reason for this is that this method is not based on the evaluation of end components in the subset product, and also its implementation might not involve some of the optimisations we apply.
In most cases, the explicit-state implementation is faster than the BDD-based approach, which is, however, more memory-efficient.

\begin{table}[t]
	\centering
\renewcommand{\tabcolsep}{2mm}
	\caption{Runtime comparison for the workstation cluster protocol.
		In the table, we use the following shortcuts for the formulas:\newline
		$
		\mathit{prop}\ltlU_{\!k} := \ltlP (\mathit{left\_{n}}{=}n \ltlU (\mathit{left\_{n}}{=}n{-}1 \ltlU (\ldots \ltlU (\mathit{left\_{n}}{=}n{-}k \ltlU\mathit{right\_{n}}{\neq}n) \ldots )
		$%
		, \newline
		$
		\mathit{prop}{\ltlGF}{\wedge}_{\!k} := \ltlP ( \ltlGF \mathit{left\_{n}}{=}n \wedge \bigvee_{i=0}^k \ltlFG \mathit{right\_{n}}{=}n{-}i )
		$%
		, and \newline
		$
		\mathit{prop}{\ltlGF}{\vee}_{\!k} := \ltlP ( \ltlGF \mathit{left\_{n}}{=}n \vee \bigvee_{i=0}^k \ltlFG \mathit{right\_{n}}{=}n{-}i )
		$%
		.
	}
	\label{tab:cluster}
	\begin{tabular}{ldddddd}
	\toprule
		& 
		\multicolumn{6}{c}{time} \\
		\multicolumn{1}{c}{property} & 
		\multicolumn{1}{c}{BP expl.} & \multicolumn{1}{c}{BP BDD} & \multicolumn{1}{c}{RB expl.} & \multicolumn{1}{c}{RB BDD} & \multicolumn{1}{c}{\prism} & \multicolumn{1}{c}{Rabinizer3} \\
	\cmidrule(r{2pt}){1-1} \cmidrule(lr){2-7}
		$\mathit{prop\ltlU}_{\!9}$ 
		& \textbf{2} & \textbf{2} & \textbf{2} & \textbf{2} & 6 & 25 \\
		$\mathit{prop\ltlU}_{\!10}$ 
		& \textbf{2} & 3 & \textbf{2} & 3 & 23 & 121 \\
		$\mathit{prop\ltlU}_{\!11}$ 
		& \textbf{3} & 4 & \textbf{3} & 4 & 95 & 686 \\
		$\mathit{prop\ltlU}_{\!12}$ 
		& \textbf{4} & 5 & \textbf{4} & 5 & \multicolumn{1}{c}{--} & \multicolumn{1}{c}{--} \\
		$\mathit{prop\ltlU}_{\!13}$ 
		& \textbf{7} & 8 & \textbf{7} & 8 & \multicolumn{1}{c}{--} & \multicolumn{1}{c}{--} \\
		$\mathit{prop\ltlU}_{\!14}$ 
		& \textbf{7} & 9 & \textbf{7} & 9 & \multicolumn{1}{c}{--} & \multicolumn{1}{c}{--} \\
	\midrule
		$\mathit{prop{\ltlGF}{\wedge}}_{\!2}$ 
		& \textbf{1} & \textbf{1} & \textbf{1} & \textbf{1} & 2 & \textbf{1} \\
		$\mathit{prop{\ltlGF}{\wedge}}_{\!3}$ 
		& \textbf{1} & \textbf{1} & \textbf{1} & \textbf{1} & 48 & \textbf{1} \\
		$\mathit{prop{\ltlGF}{\wedge}}_{\!4}$ 
		& 2 & \textbf{1} & \textbf{1} & \textbf{1} & \multicolumn{1}{c}{--} & 2 \\
		$\mathit{prop{\ltlGF}{\wedge}}_{\!5}$ 
		& \textbf{1} & \textbf{1} & \textbf{1} & 2 & \multicolumn{1}{c}{--} & 14 \\
		$\mathit{prop{\ltlGF}{\wedge}}_{\!6}$ 
		& \textbf{1} & \textbf{1} & \textbf{1} & \textbf{1} & \multicolumn{1}{c}{--} & 177 \\
		$\mathit{prop{\ltlGF}{\wedge}}_{\!7}$ 
		& \textbf{2} & \textbf{2} & \textbf{2} & \textbf{2} & \multicolumn{1}{c}{--} & \multicolumn{1}{c}{--} \\
	\midrule
		$\mathit{prop{\ltlGF}{\vee}}_{\!2}$ 
		& \textbf{1} & \textbf{1} & \textbf{1} & 2 & 2 & \textbf{1} \\
		$\mathit{prop{\ltlGF}{\vee}}_{\!3}$ 
		& \textbf{1} & \textbf{1} & 2 & 3 & 233 & \textbf{1} \\
		$\mathit{prop{\ltlGF}{\vee}}_{\!4}$ 
		& \textbf{1} & \textbf{1} & 2 & 3 & \multicolumn{1}{c}{--} & 2 \\
		$\mathit{prop{\ltlGF}{\vee}}_{\!5}$ 
		& \textbf{1} & 2 & \textbf{1} & 2 & \multicolumn{1}{c}{--} & 14 \\
		$\mathit{prop{\ltlGF}{\vee}}_{\!6}$ 
		& \textbf{2} & \textbf{2} & \textbf{2} & 4 & \multicolumn{1}{c}{--} & 180 \\
		$\mathit{prop{\ltlGF}{\vee}}_{\!7}$  
		& \textbf{3} & \textbf{3} & \textbf{3} & 6 & \multicolumn{1}{c}{--} & \multicolumn{1}{c}{--} \\
	\bottomrule
	\end{tabular}%
\end{table}

As third case study, we consider a model~\cite{HHK00} of two clusters of $n {=} 16$ workstations each, so that the two clusters are connected by a backbone.
Each of the workstations may fail with a given rate, as may the backbone.
Though this case study is a continuous-time Markov chain, we focused on time-unbounded properties, such that we could use discrete-time Markov chains to analyse them.
We give the results in Table~\ref{tab:cluster}, where the meaning of the columns is as for the mutual exclusion case in Table~\ref{tab:mutual}.
As before, we mark the best (rounded) running times with bold font.
The properties $\mathit{prop}\ltlU_{\!k}$ are probabilities of the event of component failures with respect to the order (first $k$ failures on left before right) while $\mathit{prop}{\ltlGF}{\wedge}_{\!k}$ and $\mathit{prop}{\ltlGF}{\vee}_{\!k}$ describe the long-run number of workstations functional.
As clearly shown from the results in the table, \iscasmc outperforms \prism and \rabinizerthree all cases, in particular for large PLTL formulas. 
It is worthwhile to analyse in details the three properties and how they have been checked:
for the $\mathit{prop}\ltlU_{\!k}$ case, the subset construction suffices and returns a (rounded) probability value of $0.509642$;
for $\mathit{prop}{\ltlGF}{\wedge}_{\!k}$, the breakpoint construction is enough to determine that the property holds with probability $0$. 
This explains why the BP and RB columns are essentially the same (we remark that the reported times are the rounded actual runtimes).
Property $\mathit{prop}{\ltlGF}{\vee}_{\!k}$, instead, requires to use the multi-breakpoint or the Safra-based construction to complete the model checking analysis and obtaining a probability value of $1$.

\begin{table}[t]
	\centering
\renewcommand{\tabcolsep}{2mm}
	\caption{Runtime comparison for the self-stabilising protocol. In the table, $\mathit{prop}_{k}$ is the formula\newline
	$\ltlPmin[\ltlG ( (\sum_{i=1}^{k}q_{i} = k) \implies (((\sum_{i=1}^{k}q_{i} = k) \ltlU (\sum_{i=1}^{k}q_{i} = k-1)) \ltlU \dots \ltlU (\sum_{i=1}^{k}q_{i} = 1)))]$
	}
	\label{tab:ij}
	\begin{tabular}{lcdddd}
	\toprule
	& & \multicolumn{3}{c}{time} \\
	\multicolumn{1}{c}{property} & \multicolumn{1}{c}{n} & \multicolumn{1}{c}{\iscasmc expl.} & \multicolumn{1}{c}{\iscasmc BDD} & \multicolumn{1}{c}{\prism} & \multicolumn{1}{c}{\rabinizerthree}\\
	\cmidrule(r{2pt}){1-2} \cmidrule(lr){3-6}
	\multirow{6}{*}{$\mathit{prop}_6$} 
	& 10 & \textbf{1} & 2 & 56 & 360 \\
	& 11 & \textbf{1} & 4 & 56 & 359 \\
	& 12 & \textbf{2} & 7 & 56 & 360 \\
	& 13 & \textbf{3} & 14 & 57 & 361 \\
	& 14 & \textbf{4} & 29 & 56 & 363 \\
	& 15 & \textbf{6} & 71 & 59 & 364 \\
	\midrule
	\multirow{6}{*}{$\mathit{prop}_7$} 
	& 10 & \textbf{1} & 2 & \multicolumn{1}{c}{--} & \multicolumn{1}{c}{--} \\
	& 11 & \textbf{2} & 4 & \multicolumn{1}{c}{--} & \multicolumn{1}{c}{--} \\
	& 12 & \textbf{2} & 7 & \multicolumn{1}{c}{--} & \multicolumn{1}{c}{--} \\
	& 13 & \textbf{3} & 14 & \multicolumn{1}{c}{--} & \multicolumn{1}{c}{--} \\
	& 14 & \textbf{4} & 31 & \multicolumn{1}{c}{--} & \multicolumn{1}{c}{--} \\
	& 15 & \textbf{6} & 72 & \multicolumn{1}{c}{--} & \multicolumn{1}{c}{--} \\
	\midrule
	\multirow{6}{*}{$\mathit{prop}_8$} 
	& 10 & \textbf{1} & 2 & \multicolumn{1}{c}{--} & \multicolumn{1}{c}{--} \\
	& 11 & \textbf{2} & 4 & \multicolumn{1}{c}{--} & \multicolumn{1}{c}{--} \\
	& 12 & \textbf{2} & 7 & \multicolumn{1}{c}{--} & \multicolumn{1}{c}{--} \\
	& 13 & \textbf{3} & 14 & \multicolumn{1}{c}{--} & \multicolumn{1}{c}{--} \\
	& 14 & \textbf{4} & 30 & \multicolumn{1}{c}{--} & \multicolumn{1}{c}{--} \\
	& 15 & \textbf{7} & 69 & \multicolumn{1}{c}{--} & \multicolumn{1}{c}{--} \\
	\bottomrule
	\end{tabular}%
\end{table}

Finally, as fourth case study, we consider a self-stabilising protocol originally proposed by Israeli and Jalfon \cite{IJ90}; 
as our \prism model, we consider the one adopted in \cite{KNP12a}.
The model is parametric in the number $n$ of participants.
For this case study, we analyse three instances of the PLTL formula
\[
	\ltlPmin\left[\ltlG \left( (\sum_{i=1}^{k}q_{i} = k) \implies \Big(\big((\sum_{i=1}^{k}q_{i} = k) \ltlU (\sum_{i=1}^{k}q_{i} = k-1)\big) \ltlU \dots \ltlU (\sum_{i=1}^{k}q_{i} = 1)\Big) \right)\right]
\]
providing a statement about the number of tokens still existing during the executing of the self-stabilising algorithm.
In Table~\ref{tab:ij} we provide performance comparisms for $k=6,7,8$.
\iscasmc was able to decide all properties by the subset criterion, such that applying breakpoint, multi-breakpoint, or Safra determinisation was not necessary.
Therefore, we only provide two entries for the \iscasmc explicit (``\iscasmc expl.'') and symbolic (``\iscasmc BDD'') implementation.
From the results, it is clear that the explicit version of \iscasmc, by using the subset criterion, outperforms both \prism and \rabinizerthree; 
if we consider the symbolic version, \iscasmc BDD is always faster than \rabinizerthree and only in one case ($k=6$, $n=15$) it is slower than \prism.
By analysing the causes that made \prism fail in all cases for $k=7$ and $k=8$, we find that the construction of the DRA for the given formula took too long or too much memory, such that \prism failed before even starting to construct the product of the model and the Rabin automaton.

It is interesting to observe, together with the runtimes for the mutual exclusion protocol shown in Table~\ref{tab:mutual} and the runtimes for the workstation cluster shown in Table~\ref{tab:cluster}, how \rabinizerthree performs much better on formulas involving mainly the nested $\ltlG$ operators than the $\ltlU$ operator.
This seems to be caused by the master-slave automata construction underlying \rabinizerthree (cf.~\cite{EsparzaK14,KomarkovaK14}):
given an LTL formula $\phi$, let $\phi_{1}$, \dots, $\phi_{n}$ be LTL formulas such that for each $\phi_{i}$, $\ltlG \phi_{i}$ is a sub-formula of $\phi$, occurring in the scope of an $\ltlF$ temporal operator.
For instance, for $\phi = \ltlF(\ltlG (a \wedge \ltlFG b) \vee \ltlG c)$, examples of such formulas are $\phi_{1} = a \wedge \ltlFG b$, $\phi_{2} = b$ and $\phi_{3} = c$.
Checking whether a subformula $\ltlG \phi_{i}$ finally holds, i.e., it occurs in the scope of $\ltlF$, is delegated to a slave automaton, while the remaining tasks are directly managed by the master automaton.
For the formula $\phi$, three slaves are created, one for each formula $\ltlG\phi_{i}$.
This means that a formula like $\mathit{prop}{\ltlGF}{\wedge}_{\!k}$ and $\mathit{prop}{\ltlGF}{\vee}_{\!k}$ can be split among multiple slaves, while a formula like $\mathit{prop}\ltlU_{\!k}$ is essentially managed by a single automaton, thus taking no advantage from the master-slave construction.

\clearpage
\appendix

\section{Markov Decision Processes}
\label{app:MDP}

\begin{wrapfigure}[6]{right}{30mm}
	\centering
	\begin{tikzpicture}[->,>=stealth,auto]
		\path[use as bounding box] (-1,0) rectangle (1,1);
		
		\node (M) at (0,0) {$\mdp_{\calE}$};
		\node (m0) at ($(M) + (-1,1)$) {$m_{0}^{a}$};
		\node (m1) at ($(M) + (1,1)$) {$m_{1}^{b}$};
		
		\draw ($(m0.north) + (0,0.3)$) to node {} (m0.north);
		\draw (m0) to[bend left=15] node[above, very near end] {$1$} (m1);
		\draw (m1) to[bend left=15] node[below, very near end] {$1$} (m0);
		\draw (m1) edge[loop below, distance=7mm] node[left]{$1$} (m1);

	\end{tikzpicture}
	\caption{An MDP}
	\label{fig:exampleMDP}
\end{wrapfigure}
In this section, we give more details on how we evaluate MDPs.
Let $\mdp = (\pstates, \labelFunc, \Aset, \pinit, \pmat)$ be an MDP.
Figure~\ref{fig:exampleMDP} shows an example of an MDP;
the labels of the states are in superscript and we omit the actions of the transitions.

We recall the computation of the probability $\sup_{\sched} \prob^{\sched}(\baut_{\phi})$ where $\sched \colon \pstates \to \Aset$ is a deterministic scheduler, based on~\cite[Lemma 10.125]{BaierK08} but adapted to the accepting condition on transitions instead of on states.
Let $\aut = \det(\baut_{\phi})$ denote the deterministic Rabin automaton for $\baut_{\phi}$.

\subsection{Product of MDPs with Deterministic Automata}
As for MC, we have to define products of MDP with deterministic automata.
We first define the product MDP and then the product automata and finally the quotient MDP.

\begin{definition}[Product MDP]
	Given an MDP $\mdp = (\pstates, \labelFunc, \Aset, \pinit, \pmat)$ and a deterministic automaton
	$\aut = (\Sigma, \astates, \dinit, \amat, \ACC)$, the product MDP is defined
	by $\mdp \prodMDP \aut \defeq (\pstates \times \astates, \labelFunc', \Aset, \pinit', \pmat')$ where
	\begin{itemize}
	\item
		$\labelFunc'((m,d)) = \labelFunc(m)$;
	\item
		$\pinit'((m,d)) = \pinit(m)$ if $d = \amat(\dinit, \labelFunc(m))$, and $0$ otherwise; and
	\item
		$\pmat'((m,d), a)(m',d')$ equals $\pmat(m,a)(m')$ if $d' = \amat(d, \labelFunc(m'))$, and is $0$ otherwise.
	\end{itemize}
\end{definition}
We recall that $(p,a,p') \in \pmat'$ means $\pmat'(p,a)(p') > 0$ and we define $\projaut(p,a,p')$ to be the projection on $\aut$ of the given $(p,a,p') = ((m,d), a, (m',d'))$, i.e., $\projaut(p,a,p') = (d, \labelFunc(m'), d')$.

\begin{figure}[h!]
	\centering
	\begin{tikzpicture}[->,>=stealth,auto]
		\path[use as bounding box] (-4,-0.25) rectangle (4.75,4);
		\node (B) at (-2.5,0) {$\baut_{\calE}$};
		\node (MS) at (2.5,0) {$\mdp_{\calE} \prodMDP \ssaut_{\calE}$};
		
		\node (b0) at ($(B) + (-1,1)$) {$q_{0}$};
		\node (b1) at ($(B) + (1,1)$) {$q_{1}$};
		\draw ($(b0.south) + (0,-0.3)$) to node {} (b0);
		\draw (b0) edge[loop above, distance=7mm] node[above] {$\mathit{true}$} (b0);
		\draw (b0) to node[above] {$a$} (b1);
		\draw (b1) edge[loop above, double, distance=7mm] node[above] {$a$} (b1);
		\draw (b1) edge[loop right, distance=7mm] node[above] {$\neg a$} (b1);

		\node (p0) at ($(MS) + (-1.5,2)$) {$m_{0}, \setnocond{q_{0}}$};
		\node[rectangle, dashed, draw, rounded corners, minimum width=20mm, minimum height=26mm] (MEC) at ($(MS) + (1,2)$) {};
		\node (p1) at ($(MEC) + (0,1)$) {$m_{1}, \setnocond{q_{0}, q_{1}}$};
		\node (p2) at ($(MEC) + (0,-1)$) {$m_{0}, \setnocond{q_{0},q_{1}}$};

		\draw (p0) to[bend left=15] node[above, near end] {$1$} (p1);
		\draw (p1) to[bend left=15] node[right, near end] {$1$} (p2);
		\draw (p1) edge[loop above, distance=7mm] node[right, near end] {$1$} (p1);
		\draw (p2) to[bend left=15] node[left, near end] {$1$} (p1);
	\end{tikzpicture}
	
	\caption{A \buchi automaton $\baut_{\calE}$ and the product MDP of $\mdp_{\calE}$ and $\ssaut_{\calE}$}
	\label{fig:exampleMDPTimesSubSet}
\end{figure}
Figure~\ref{fig:exampleMDPTimesSubSet} shows in the left hand side a \buchi automaton whose language contains all words from $\Sigma^{\omega}$ where $a$ occurs infinitely often;
the labels $\mathit{true}$ and $\neg a$ are shortcuts for all labels in $\Sigma$ and in $\Sigma \setminus \setnocond{a}$, respectively.
This means, for instance, that if $\Sigma = \setnocond{a, b}$, then $q_{0}$ has two transitions with label $a$ and $b$, respectively, to $q_{0}$ itself.
The right hand side of Figure~\ref{fig:exampleMDPTimesSubSet} shows the product MDP between $\mdp_{\calE}$ in Figure~\ref{fig:exampleMDP} and the subset construction $\ssaut_{\calE}$ of $\baut_{\calE}$.
The dashed box encloses a maximal end component we will formally define later;
intuitively, a maximal end component is the MDP counterpart of a bottom SCC of a Markov chain.

\begin{definition}[RMDP, PMDP, and GMDP]
	Given an MDP $\mdp = (\pstates, \labelFunc, \Aset, \pinit, \pmat)$ and a deterministic automaton
	$\aut = (\Sigma, \astates, \dinit, \amat, \ACC)$, the product automaton is defined
	by $\mdp \prodAut \aut \defeq (\mdp \prodMDP \aut, \ACC')$ where
	\begin{itemize}
	\item
		if $\ACC = \setcond{(\aacc_{i}, \arej_{i})}{i \in \natIntK}$, then $\ACC' \defeq \setcond{(\aacc'_{i}, \arej'_{i})}{i \in \natIntK}$ where $\aacc'_{i} = \setcond{(p,a,p') \in \pmat'}{\projaut(p, a, p') \in \aacc_{i}, a \in \Aset}$ and $\arej'_{i} = \setcond{(p,a,p') \in \pmat'}{\projaut(p, a, p')  \in \arej_{i}, a \in \Aset}$ (Rabin Markov Decision Process, RMDP);
	\item
		if $\ACC = \pri \colon \pmat \to \natIntK$, then $\ACC' \defeq \pri' \colon \pmat' \to \natIntK$ where $\pri'(p,a,p') = \pri(\projaut(p,a,p'))$ for each $(p,a,p') \in \pmat'$ (Parity Markov Decision Process, PMDP); or
	\item
		if $\ACC = \setcond{\afinal_{i}}{i \in \natIntK}$, then $\ACC' \defeq \setcond{\afinal'_{i}}{i \in \natIntK}$ where $\afinal'_{i} =  \setcond{(p,a,p') \in \pmat'}{\projaut(p,a,p') \in \afinal_{i}, a \in \Sigma}$ (Generalised \buchi Markov Decision Process, GMDP).
	\end{itemize}
\end{definition}

\begin{definition}[Quotient MDP]
	Given an MDP $\mdp$ and a DRA $\aut = \det(\baut)$, the \emph{quotient MDP} $[\mdp \prodMDP \aut]$ is the MDP $([\pstates \times \astates], [\labelFunc], \Aset, [\pinit], [\pmat])$ where
	\begin{itemize}
	\item
		$[\pstates \times \astates] = \setcond{(m,[d])}{(m,d) \in \pstates \times \astates,\ [d] = \setcond{d' \in \astates}{\reached(d') = \reached(d)}}$,
	\item
		$[\labelFunc](m,[d]) = \labelFunc(m,d)$,
	\item
		$[\pinit](m,[d]) = \pinit(m, d)$, and
	\item
		$[\pmat]\big((m,[d]), a, (m',[d'])\big) = \pmat\big((m,d), a, (m',d')\big)$.
	\end{itemize}
\end{definition}
As for the quotient MC, the quotient MDP is again well defined:
it is immediate to see that $d \in [d]$ and that for each $(m, d_{1}), (m, d_2) \in [(m,d)]$ and each $a \in \Aset$, it holds
$\pmat\big((m, d_{1}), a, (m', [d'])\big) = \pmat\big((m, d), a, (m', d')\big) = \pmat\big((m, d_{2}), a, (m', [d'])\big)$.

\subsection{Reduction to Probabilistic Reachability}
For the reader's convenience, we recall the reduction of $\sup_{\sched}\prob^{\mdp, \sched}(\baut_{\phi})$ to probabilistic reachability in the RMDP $\mdp \prodAut \aut$.

First, we introduce some concept and corresponding notation, starting with the concept of maximal end component (MEC) that is the MDP counterpart of the SCC for a MC;
the formal definition is not so immediate as we have to take care of the role of the actions.
\begin{definition}[MEC]
	Given an MDP $\mdp = (\pstates, \labelFunc, \Aset, \pinit, \pmat)$,
	\begin{itemize}
	\item
		a sub-MDP is a pair $(T, \en)$ such that $\emptyset \neq T \subseteq \pstates$ and $\en \colon \pstates \to 2^{\Aset}$ satisfying:
		\begin{inparaenum}[$(1.)$]
		\item
			$\emptyset \neq \en(t) \subseteq \Aset(t)$ for each $r \in T$ where $\Aset(t)$ denotes the enabled actions of $t$, and
		\item
			$t \in T$ and $a \in \en(t)$ implies $\pmat(t, a) \in \dist(T)$.
		\end{inparaenum}
	\item
		An \emph{end component} of $\mdp$ is a sub-MDP $(T, \en)$ such that the digraph induced by $(T, \en)$ is strongly connected.
	\item
		An end component $\mec = (T, \en)$ is a \emph{maximal end component} (MEC) if it is not contained in some other end component $(T', \en') \neq (T, \en)$ with $T \subseteq T'$ and $\en(t) \subseteq \en'(t)$ for all $t \in T$.
	\end{itemize}
\end{definition}
As for SCCs, we define the transitions $\mectran$ of the MEC $\mec = (T, \en)$ of the MDP $\mdp$ as $\mectran = \setcond{(t,a,t') \in \pmat}{t, t' \in T, a \in \en(t)}$.
\begin{definition}[Accepting MEC]
	Given an MDP $\mdp$ and a DRA $\aut = \det(\baut)$, let $\mec$ be a MEC of the product RMDP $\mdp \prodAut \aut$.
	We say that $\mec$ is accepting if there exists an index $i \in \natIntK$ such that $\aacc_{i} \cap \projaut(\mectran) \neq \emptyset$ and $\arej_{i} \cap \projaut(\mectran) = \emptyset$.
\end{definition}
Note that Theorem~\ref{thm:bottom_component} and Corollary~\ref{cor:quotients} extend easily to the MECs of the quotient MDP and to accepting MECs of the product automaton, respectively.

We now describe how to simplify the identification of the accepting MECs.
Given an MDP $\mdp$, we first define a sub-MDP restricting to a subset of transitions $\transtyle{B} \subseteq \pmat$.
Let $\mdp_{\Box \transtyle{B}}$ denote the sub-MDP with transition space $\transtyle{B}$.
Moreover, let $\widehat{\transtyle{B}}$ be the completion of the set $\transtyle{B}$, i.e., $\widehat{\transtyle{B}} = \setcond{(m,a,m') \in \pmat}{\exists m'' \in \pstates.\ (m,a,m'') \in \transtyle{B}}$.

With these notations, we recall that the computation of the probability
$\sup_{\sched}\prob^{\sched}(\baut_{\phi})$ can be reduced to a reachability
probability:
\begin{itemize}
\item
	for a Rabin pair $(\aacc, \arej)$, define $U_{(\aacc, \arej)} \subseteq \pstates$ such that $m \in U_{(\aacc, \arej)}$ if there exists a MEC $\mec = (T, \en)$ in $\mdp_{\Box (\pmat \setminus \widehat{\arej})}$ such that $m \in T$ and $\mectran \cap \aacc \neq \emptyset$.
	In this case we say $m$ is accepting with respect to $(\aacc, \arej)$.
	Note that if $m$ is accepting, then all states in $T$ are accepting as well, thus $T \subseteq U_{(\aacc, \arej)}$.
\item
	Let $U \defeq \bigcup_{i \in \natIntK} U_{(\aacc_{i}, \arej_{i})}$ be the accepting region and $\aut = \det(\baut)$.
	Then, $\sup_{\sched} \prob^{\mdp, \sched}(\baut_{\phi}) $ reduces to a probabilistic reachability:
	\[
		\sup_{\sched}\prob^{\mdp,\sched}(\baut_{\phi}) = \sup_{\sched}\prob^{\mdp \prodAut \aut,\sched}(\Diamond U)\text{.}
	\]
\end{itemize}

Obviously, if a MEC $(T, \en)$ contains a state $m \in U$ from the accepting region, then the probability of accepting the language of $\baut_{\phi}$ is the same for all of them, so it does not matter the particular state $m' \in T$ we reach when we enter $(T, \en)$.

\subsection{The Incremental Evaluation of MECs}

As argued in the body of the paper, the likelihood $\sup_{\sched} \prob^{\mdp \prodAut \aut,\sched}_{(m,d)}(\baut_{\phi})$ does not depend on $d$ itself, but only on $[d]$.
Together with the observations from Corollary~\ref{cor:MDPestimate}, we can therefore follow the principle layered approach.
As for Markov chains, we can use the results from the previous layers to avoid parts of the construction.
The incremental evaluation is described in details as follows.

Construct the quotient MDP $\mdp \prodAut \ssaut^{u}$ (which is the same as $\mdp \prodAut \ssaut^{o}$ except the accepting conditions).
For each MEC $\mec$, evaluate as follows:
\begin{enumerate}
\item
	Similar to Lemma~\ref{lem:subset} for Markov chains, $\mec$ is accepting if $\mec$ contains some accepting transition $(p,a,p')$ with $\proj{\ssaut^{u}}(p,a,p') \in \afinal^{u}_{i}$ for each $i \in \natIntK$;
	$\mec$ is rejecting if $\mec$ does not contain transitions $(p,a,p')$ with $\proj{\ssaut^{u}}(p,a,p') \in \afinal^{o}_{i}$ for some $i \in \natIntK$.

\item
	If we could neither establish that $\mec$ is accepting nor that $\mec$ is rejecting, we refine $\mec$ by a breakpoint construction (only for this MEC).
	Let $(\mdp \prodAut \rbpaut^{u})|_{\mec}$, $(\mdp \prodAut \rbpaut^{o})|_{\mec}$ denote the breakpoint automata for $\mec$ (again, the difference is only reflected by the accepting conditions).

	These RMDPs are to be read as the MDPs that exist if one expands the MDP restricted to $\mec$.
	Let $\mec'$ be an arbitrary MEC in the resulting breakpoint automata;
	we recall that the Rabin pairs for $\rbpaut^{o}$ and $\rbpaut^{u}$ are $\setnocond{(\aacc_{\epsilon},\emptyset), (\amat', \arej_{0})}$ and $\setnocond{(\aacc_{\epsilon},\emptyset)}$, respectively.

\item
	$\mec'$ is accepting if $\mec'$ contains some accepting transition $(p,a,p')$ with $\proj{\rbpaut^{u}}(p,a,p') \in \aacc_{\epsilon}$, i.e., it is accepted by $\rbpaut^{u}$.
	Otherwise, $\mec' = (T, \en)$ is for sure rejecting if it is rejected by $\rbpaut^{o}$, i.e., there exists $(p, a, p') \in \mecprimetran$ with $\proj{\rbpaut^{u}}(p, a, p') \in \arej_{0}$.

\item
	If $\mec$ remains inconclusive, we finally evaluate the MECs of $\mdp \prodAut \aut$.
	For this, we follow a similar procedure as for the breakpoints.
	In particular, it suffices to refine the MECs individually.
\end{enumerate}

After the above procedure, we make accepting MECs in the quotient MDP absorbing, whereas non-accepting MECs are untouched.
We remark that a MEC cannot be accepting if it contains a state $(m, \langle d \rangle)$ such that $\sup_{\sched} \prob^{\mdp \prodAut \aut, \sched}_{(m,d)}(\baut_{\phi}) < 1$ has been established by a previous estimation.
As a final step, a probabilistic reachability analysis---through solving a linear programming problem---is needed for the evaluation.
Note that this can always be performed on the quotient MDP.

\begin{figure}
	\centering
	\begin{tikzpicture}[->,>=stealth,auto]
		\path[use as bounding box] (-3,0) rectangle (3,3);
		
		\node (MR) at (0,0) {$\mdp_{\calE} \prodAut \aut_{\calE}$};

		\node[draw, rectangle, inner sep=0pt, rounded corners, minimum width=18mm, minimum height=9mm] (p0) at ($(MR) + (-4.5,1.5)$) {$m_{1},\ \ \phantom{\setnocond{q_{0}}}$};
		\node[astateq] (p0t0) at ($(p0) + (0.3,0)$) {$\setnocond{q_{0}}$};
		\node[dashed, rectangle, draw, rounded corners, minimum width=20mm, minimum height=17mm] at ($(p0) + (0,0.3)$) {};
		\draw (p0) edge[loop above, distance=7mm] node[right, very near end] {$1$} (p0);
		
		\draw ($(p0.south) + (0,-0.5)$) to node {} (p0.south);
		
		\node[draw, rectangle, inner sep=0pt, rounded corners, minimum width=22mm, minimum height=9mm] (p1) at ($(MR) + (-1.6,1.5)$) {$m_{0},\ \ \phantom{\setnocond{q_{0}, q_{1}}}$};
		\node[astateq] (p1t0) at ($(p1) + (0.3,0)$) {$\setnocond{q_{0},q_{1}}$};

		\node[dashed, rectangle, draw, rounded corners, minimum width=54mm, minimum height=20mm] (MEC) at ($(MR) + (3,1.5)$) {};
		\node[draw, rectangle, inner sep=0pt, rounded corners, minimum width=22mm, minimum height=9mm] (p2) at ($(MEC) + (-1.5,0)$) {$m_{1},\ \ \phantom{\setnocond{q_{0},q_{1}}}$};
		\node[astateq] (p2t0) at ($(p2) + (0.3,0)$) {$\setnocond{q_{0},q_{1}}$};
		
		\node[draw, rectangle, inner sep=0pt, rounded corners, minimum width=22mm, minimum height=18mm] (p3) at ($(MEC) + (1.5,0)$) {$m_{0},\ \ \phantom{\setnocond{q_{0},q_{1}}}$};
		\node[astateq] (p3t0) at ($(p3) + (0.3,0.5)$) {$\setnocond{q_{0},q_{1}}$};
		\node[astateq] (p3t1) at ($(p3) + (0.3,-0.5)$) {$\setnocond{q_{1}}$};
		\draw[-] (p3t0) to node {} (p3t1);

		\draw (p0) to node[above, near end] {$1$} (p1);

		\draw (p1) to node[above, near end] {$1$} (p2);

		\draw (p2) to[bend left=15] node[above, near end] {$1$} (p3);

		\draw[double] (p3) to[bend left=15] node[below, near end] {$1$} (p2);
	
	\end{tikzpicture}
	\caption{The product automaton of $\mdp_{\calE}$ and $\aut_{\calE} = \det(\baut_{\calE})$}
	\label{fig:exampleMDPTimesRabin}
\end{figure}

Figure~\ref{fig:exampleMDPTimesRabin} shows the product automaton of $\mdp_{\calE}$ and $\aut_{\calE} = \det(\baut_{\calE})$, i.e., the RMDP obtained with the DRA corresponding to $\baut_{\calE}$ depicted in Figure~\ref{fig:exampleMDPTimesSubSet}.
This example is particularly interesting since it remarks that finding a rejecting MEC is inconclusive.
In fact, consider the MEC enclosed in the left hand dashed box of the product automaton.
This MEC is rejecting but we can not conclude that the language of the product automaton is empty since there is another scheduler that eventually leaves such MEC reaching the right hand MEC that is indeed accepting, since it contains an accepting transition, depicted with the double arrow.

\section{Symbolic Implementation of Our Approach}
\label{sec:bdd}
Binary decision diagrams (BDDs)~\cite{Lee:1959:RSC} are a well-known mechanism to represent binary functions $f \colon \booleans^{V} \to \booleans$ (where $\booleans$ is the Boolean set and $v \in \booleans^{V}$ means $v \colon V \to \booleans$) by using a specific form of directed acyclic graphs.
Given a NGBA $\baut$, we describe first how $\baut$ can be represented symbolically.
By taking $V = \astates$, we can construct the indicator BDDs $v_{R}$ for any set $R \subseteq \astates$, such that $v_{R}(q) = 1$ if, and only if, $q \in R$.
We can represent multiple states of the subset automaton by building a disjunction over their indicator functions.
This is then used to encode states of subset automata $\ssaut^{o}$ and $\ssaut^{u}$ (which are subsets of the states $\astates$ of $\baut$).
We encode the transition relation by introducing additional variables $\Sigma$ representing the transition labels and a copy $\successor{\astates}$ of $\astates$ as successor variables.
For this, we construct a function $t \colon \booleans^{\astates} \times \booleans^{\Sigma} \times \booleans^{\successor{\astates}} \to \booleans$ where $\amat(R, \sigma) = C$ if, and only if, $t(v_{R}, v_{\sigma}, v_{C}) = 1$.
The construction for the breakpoint automata is similar;
for instance, denoted by $\natIntK$ the set of variables encoding $\natIntK$, the states can be represented by using functions from $\booleans^{\astates} \times \booleans^{\natIntK} \times \booleans^{\astates}$ to $\booleans$.

This idea resembles~\cite{MSL08}, but there additional variables are introduced to enumerate states of subset (or breakpoint) automata.
For our purposes this is not needed.
The construction of the single-breakpoint and multi-breakpoint automata is almost identical in terms of their BDD representations.
According to~\cite{MSL08}, Rabin automata are not well suited to be be constructed using BDDs directly.
It is however possible to construct them in an explicit way and convert them to a symbolic representation afterwards.
For this, we assign a number to each of the explicit states of the Rabin automaton.
Afterwards, we can refer to the state using BDD variables encoding this number.

We emphasise that we can still compute Rabin automata on-the-fly when using the BDD-based approach, so as to avoid having to construct parts of the Rabin automaton which are not needed in the product with the MC or MDP.
It might happen that in the symbolic computation of the reachable states of the product we note that a certain state of the Rabin automaton is required.
In this case, we compute this successor state in the explicit representation of the automaton, assign to it a new number, encode this number using BDDs and then use this BDD as part of the reachable states.

MDPs can be represented similarly.
To represent exact transition probabilities, one can involve multi-terminal BDDs (MTBDDs)~\cite{ClarkeFZ96,AlfaroKNPS00}.
If they are not required, BDDs are sufficient.
Products of symbolic model and automata can then be computed using (MT)BDD operations, allowing for effective symbolic analyses.
To compute the (bottom) SCCs, we employ a slight variantion of~\cite{GentiliniPP03}.
For MDPs, we then employ a symbolic variant of the classical algorithm~\cite{Courcoubetis+Yannakakis/95/Markov,Alfaro97} to obtain the set of MECs from the set of SCCs.
The acceptance of an SCC/MEC can be decided by a few BDD operations.

\section{Proofs}

\subsection{Proof of Theorem~\ref{thm:only_reach}}
To establish Theorem~\ref{thm:only_reach}, we show inclusion in both directions.
The proof is the same as the correctness proof for the determinisation construction in \cite{Schewe+Varghese/12/generalisedBuchi}, but the claim is different, and the proof is therefore included for completeness.
The difference in the claim is that $\lang(\aut_{d}) = \lang(\baut_{\reached(d)})$ is shown for $d = (\calT, l, h)$ with a singleton set $\calT = \setnocond{\epsilon}$ that contains only the root and $h(\epsilon)=1$.
The proof, however, does not use either of these properties.

\begin{lemma}
	$\lang(\aut_{d}) \subseteq \lang(\baut_{\reached(d)})$
\end{lemma}

\paragraph*{Notation}

For an $\omega$-word $\word$ and $j \geq i$, we denote with $\word[i,j[$ the word $\word(i) \cdot \word(i+1) \cdot \word(i+2) \cdot \ldots \cdot \word(j-1)$.
We denote with $\astates_{1} \rightarrow^\word \astates_{2}$ for a finite word $\word = \word_{1} \cdot \ldots \cdot \word_{j-1}$ that there is, for all $q_{j} \in \astates_{2}$ a sequence
$q_{1} \cdot \ldots \cdot q_{j}$ with $q_{1} \in \astates_{1}$ and $(q_{i}, \word_{i}, q_{i+1}) \in \amat$ for all $1 \leq i < j$.
If one of these transitions is guaranteed to be in $\afinal_{a}$, we write $\astates_{1} \Rightarrow^{\word}_{a} \astates_{2}$.

For an extended history tree $d = (\calT, l, h)$, we denote for a state $q \in \reached(d)$ the node of $\vartheta \in \calT$, such that $q \in l(\vartheta)$, but not in any child of $\vartheta$ ($\nexists i \in \naturals.\ q \in l(\vartheta i)$), by $\host(q,d)$.

For an input word $\word \colon \omega \to \Sigma$, let $\arun = d_{0} \cdot d_{1} \cdot \ldots$ be the run of the DRA $\aut$ on $\word$.
A node $v$ in the history tree $d_{i+1}$ is called \emph{stable} if $\rename(v) = v$ and \emph{accepting} if it is accepting in the transition $(d_{i}, \word(i), d_{i+1})$.

\begin{proof}
Let $\word \in \lang(\aut_{d})$.
Then there is a $v$ that is eventually always stable and always eventually accepting in the run $\arun = d_{0} \cdot d_{1} \cdot d_2 \cdot \ldots$ (with $d_{0} = d$) of $\aut_{d}$ on $\word$.
We pick such a $v$.

Let $i_{0} < i_{1} < i_{2} < \ldots$ be an infinite ascending chain of indices such that
\begin{itemize}
\item
	$v$ is stable for all $d_{j}$ with $j \geq i_{0}$, and
\item
	the chain $i_{0} < i_{1} < i_{2} < \ldots$ contains exactly those indices $i \geq i_{0}$ such that $d_{i}$ is accepting; this implies that $h$ is updated exactly at these indices.
\end{itemize}

Let $d_{i} = (\calT_{i}, l_{i}, h_{i})$ for all $i \in \omega$. By construction, we have
\begin{itemize}
\item
	$\reached(d) \rightarrow^{\word[0,i_{0}[} l_{i_{0}}(v)$, and
\item
	$l_{i_{j}}(v) \Rightarrow^{\word[i_{j},i_{j+1}[}_{h_{i_{j}}} l_{i_{j+1}}(v)$.
\end{itemize}
Exploiting K\"onig's lemma, this provides us with the existence of a run that visits all accepting sets $\afinal_{i}$ of $\baut_{\reached(d)}$ infinitely often.
(Note that the value of $h$ is circulating in the successive sequences of the run.)
This run is accepting, and $\word$ therefore belongs to the language of $\baut_{\reached(d)}$.
\end{proof}

\begin{lemma}
	$\lang(\aut_{d}) \supseteq \lang(\baut_{\reached(d)})$.
\end{lemma}

\begin{proof}
Let $\word \in \lang(\baut_{\reached(d)})$ and $\arun = q_{0} \cdot q_{1} \cdot \ldots$ be the run of $\baut_{\reached(d)}$ on the input word $\word$; 
let $\arun_{\aut} = d_{0} \cdot d_{1} \cdot \ldots$ be the run of $\aut_{d}$ on $\word$.
We then define the related sequence of host nodes $\vartheta = v_{0} \cdot v_{1} \cdot v_{2} \cdot \ldots = \host(q_{0}, d_{0}) \cdot \host(q_{1}, d_{1}) \cdot \host(q_{2}, d_{2}) \cdot \ldots$.
Let $l$ be the shortest length $|v_{i}|$ of these nodes of the trees $d_{i}$ hosting $q_{i}$ that occurs infinitely many times.

We follow the run and see that the initial sequence of length $l$ of the nodes in $\vartheta$ eventually stabilises.
Let $i_{0} < i_{1} < i_{2} < \ldots$ be an infinite ascending chain of indices such that the length $|v_{j}| \geq l$ of the $j$-th node is not smaller than $l$ for all $j \geq i_{0}$, and equal to $l = |v_{i}|$ for all indices $i \in \setnocond{i_{0}, i_{1}, i_{2}, \ldots}$ in this chain.
This implies that $v_{i_{0}}$, $v_{i_{1}}$, $v_{i_{2}}$, \ldots is a descending chain when the single nodes $v_{i}$ are compared by lexicographic order.
As the domain is finite, almost all elements of the descending chain are equal, say $v_{i}:=\pi$.
In particular, $\pi$ is eventually always stable.

Let us assume for contradicting that this stable prefix $\pi$ is accepting only finitely many times.
We choose an index $i$ from the chain $i_{0} < i_{1} < i_2 < \ldots$ such that $\pi$ is stable for all $j \geq i$.
(Note that $\pi$ is the host of $q_{i}$ for $d_{i}$, and $q_{j} \in l_{j}(\pi)$ holds for all $j \geq i$.)

As $\arun$ is accepting, there is a smallest index $j > i$ such that $(q_{j-1}, \word(j-1), q_{j}) \in \afinal_{h_{i}(\pi)}$.
Now, as $\pi$ is not accepting, $q_{i}$ must henceforth be in the label of a child of $\pi$, which contradicts the assumption that infinitely many nodes in $\vartheta$ have length~$|\pi|$.

Thus, $\pi$ is eventually always stable and always eventually accepting.
\end{proof}

The lemmas in this appendix imply Theorem~\ref{thm:only_reach}.

\subsection{Proof of Theorem~\ref{thm:bottom_component}}

Before proving Theorem~\ref{thm:bottom_component} we establish the following result about the relation between paths and bottom SCCs of the product DRA.

\begin{lemma}
\label{lem:paths}
	For a MC $\mc$, a DRA $\aut = \det(\baut)$, and the product MC $\mc \prodMC \aut$, the following holds:
	\begin{enumerate}
	\item
		if there is a path from $(m,d)$ to $(m',d')$ in $\mc \prodMC \aut$, then there is a path from $[(m,d)]$ to $[(m',d')]$ in $[\mc \prodMC \aut]$; and
	\item
		if there is a path from $[(m,d)]$ to $[(m',d')]$ in $[\mc \prodMC \aut]$ and $(m,d)$ is reachable in $\mc \prodMC \aut$, then there is a path from $(m,d)$ to some $(m',d'')$ with $\reached(d') = \reached(d'')$ in $\mc \prodMC \aut$.
	\end{enumerate}
	In particular, if $[(m,d)]$ is reachable in $[\mc \prodMC \aut]$, then $(m,d')$ is reachable in $\mc \prodMC \aut$ for some $d'$ with $\reached(d) = \reached(d')$.
\end{lemma}

Both of the claims are easy to establish by induction over the length of the path.
They prepare the relevant theorem about bottom SCCs.
\begin{proof}
To show (1), we have to run through the properties of a bottom SCC.
First, all states in $[\scc]$ are reachable and connected by Lemma~\ref{lem:paths}.
It remains to show that no state $[(m',d')] \notin [\scc]$ is the successor of any state in the quotient MC.
Let us assume for contradiction that there is a state $(m,d) \in \scc$ such that $[(m',d')]$ is the successor of $[(m,d)']$
By Lemma~\ref{lem:paths}, this implies that there is a state $(m',d'')$ with $\reached(d'')=\reached(d')$ reachable from $(m,d)$.
As $\scc$ is a bottom SCC, $(m',d'') \in \scc$ holds, and $[(m',d')] = [(m',d'')] \in [\scc]$ follows, which is a contradiction.

To show (2), let us start with the set $\scc'' = \setcond{s \in V}{[s] \in \scc}$.
$[\scc''] = \scc$ follows from the point (2) of Lemma~\ref{lem:paths}, and $\scc''$ is closed under successors:
assuming by contradiction that this is not the case provides a $(m,d) \in \scc''$ with successor $(m',d') \notin \scc''$.
(Note that the reachability of $(m,d)$ implies the reachability of $(m',d')$.)
But this implies that $[(m',d')]$ is a successor of $[(m,d)]$, and by construction of $\scc''$, $[(m',d')] \in \scc$.
For the bottom SCC $\scc$, the construction of $\scc''$ then implies $(m',d') \in \scc''$, which is a contradiction.

As $\scc''$ is closed under successors, it contains some bottom SCC, and we select $\scc'$ to be such a bottom SCC.
We have shown that $[\scc']$ is a bottom SCC in the quotient MC in the first half of this proof.
Consequently, $[\scc']$ is a bottom SCC that is contained in $\scc$, and hence $[\scc'] = \scc$ holds.
\end{proof}

\subsection{Proof of Theorem~\ref{theo:inclusions}}
In this appendix, we show Theorem~\ref{theo:inclusions}, that is, the inclusions $\lang(\ssaut_{[d]}^{u}) \subseteq \lang(\rbpaut_{\langle d \rangle}^{u}) \subseteq \lang(\aut_{d}) \subseteq \lang(\rbpaut_{\langle d \rangle}^{o}), \lang(\ssaut_{[d]}^{o})$.
The subset automata $\ssaut^{u}$, $\ssaut^{o}$, and the breakpoint automata $\rbpaut^{u}$, $\rbpaut^{o}$ are defined in Section~\ref{subset}.

Given a triple $(R, j, C)$, we refer to $R$ by $\set\big((R, j, C)\big)$, to $j$ by $\ind\big((R, j, C)\big)$, and to $C$ by $\children\big((R, j, C)\big)$.

\begin{lemma}
	$\lang(\ssaut_{[d]}^{u}) \subseteq \lang(\rbpaut_{\langle d \rangle}^{u})$.
\end{lemma}

\begin{proof}
Let us assume for contradiction that $\word$ is an infinite word such that the run $\arun$ of $\ssaut_{[d]}^{u}$ on $\word$ is accepting, while the run $\arun'$ of $\lang(\rbpaut_{\langle d \rangle}^{u})$ on $\word$ is rejecting.
Then there is a position $p \in \omega$ such that the following statements hold:
\begin{itemize}
\item
	no transition after position $p$ is accepting in $\tran{\arun'}$, and, consequently,
\item
	there is an index $i$ such that, in the run $\arun'$ of $\rbpaut_{\langle d \rangle}^{u}$, for all $n \geq p$, $\ind(\arun'(n)) = i$.
\end{itemize}

It is easy to see that, by construction, $\set(\arun'(n)) = \arun(n)$ holds for all $n \in \omega$.
As $\arun$ is accepting, there is a position $n \geq p$ such that, for all $q \in \arun(n)$ and all $q' \in \arun(n+1)$, $(q, \word(n), q') \in \amat$ implies $(q, \word(n), q') \in \afinal_{i}$.
(Otherwise $\afinal_{i}$ would not be accepting in $\ssaut_{[d]}^{u}$.)
But then, transition $n$ in $\tran{\arun'}$ is accepting, which is a contradiction.
\end{proof}

\begin{lemma}
	$\lang(\rbpaut_{\langle d \rangle}^{u}) \subseteq \lang(\aut_{d})$
\end{lemma}

\begin{proof}
Let us consider an accepting run $\arun$ of $\rbpaut_{\langle d \rangle}^{u}$ and a run $\arun'$ of $\aut_{d}$ on a given input word $\word$.
It is easy to show by induction that, for $\arun'(n) = (\calT_{n}, l_{n}, h_{n})$,
\begin{itemize}
\item
	$\set(\arun(n)) = \reached(\arun'(n))$,
\item
	$\children(\arun(n)) = \bigcup_{i \in \calT_{n} \cap \naturals} l_{n}(i)$, and
\item
	$\ind(\arun(n)) = h_{n}(\epsilon)$
\end{itemize}
hold for all $n \in \omega$, and that if $\tran{\arun}(n)$ is accepting, then $\tran{\arun'}(n)$ is accepting with accepting pair with index $\epsilon$.
As the root cannot be rejecting, this implies that $\arun'$ is accepting, too.
\end{proof}

\begin{lemma}
	$\lang(\aut_{d}) \subseteq \lang(\rbpaut_{\langle d \rangle}^{o})$
\end{lemma}

\begin{proof}
Let us consider an accepting run $\arun'$ of $\aut_{d}$ and a run $\arun$ of $\rbpaut_{\langle d \rangle}^{o}$ on a given input word $\word$.
It is easy to show by induction that, for $\arun'(n) = (\calT_{n}, l_{n}, h_{n})$,
\begin{itemize}
\item
	$\set(\arun(n)) = \reached(\arun'(n))$,
\item
	$\children(\arun(n)) = \bigcup_{i \in \calT_{n} \cap \naturals} l_{n}(i)$, and
\item
	$\ind(\arun(n)) = h_{n}(\epsilon)$
\end{itemize}
hold for all $n \in \omega$, and that if $\tran{\arun'}(n)$ is accepting with accepting pair with index $\epsilon$ (recall that the root cannot be rejecting), then $\tran{\arun}(n)$ is accepting.
If $\tran{\arun'}(n)$ is accepting, but not with index $\epsilon$, then the node with position $0$ in the history is eventually always stable (note that, whenever $0$ is not stable, no other node than $\epsilon$ is) say from position $p \in \omega$ onwards.
But then $\children(\arun(n)) \supseteq l_{n}(0)$ cannot be empty for any $n \geq p$, thus $\tran{\arun}(n) \notin \arej_{0}$.

Consequently, $\tran{\arun}$ is rejecting only finitely many times, and $\word$ is accepted by $\rbpaut_{\langle d \rangle}^{o}$.
\end{proof}

\begin{lemma}
	$\lang(\aut_{d}) \subseteq \lang(\ssaut_{[d]}^{o})$.
\end{lemma}

\begin{proof}
Recall that $\aut$ is obtained by determinisation of the NGBA $\baut$, i.e., $\aut = \det(\baut)$.
Let $\arun$ be the run of a word $\word$ that is rejected by $\ssaut_{[d]}^{o}$, and $i$ an index such that $\tran{\arun}$ contains only finitely many transitions in $\afinal_{i}$.
Then there is a position $n \in \omega$ such that no transition in $\afinal_{i}$ may occur from $n$ onwards.

But then there is no $m \geq n$ for which a position $q \in \arun(m)$ is reachable such that $(q, \word(n), q')$ is in the set $\afinal_{i}$ of the NGBA $\baut$.
Thus no run of $\baut$ on $\word$ can have more than $n$ transitions from $\afinal_{i}$, hence $\word \notin \lang(\baut_{[d]})$ and the claim follows with since $\lang(\baut_{[d]}) = \lang(\aut_{d})$ by Theorem~\ref{thm:only_reach}, thus $\word \notin \lang(\aut_{d})$.
\end{proof}

Together, the lemmas in this appendix imply Theorem~\ref{theo:inclusions}.

\subsection{Proof of Proposition~\ref{pro:buechiLangEqualSemiDetLang}}

We split Proposition~\ref{pro:buechiLangEqualSemiDetLang}, which says that given a NGBA $\baut$ we have $\lang(\bToSd(\baut)) = \lang( \baut)$, in two lemmas corresponding to the inclusions $\lang(\bToSd(\baut)) \subseteq \lang( \baut)$ and $\lang(\baut) \subseteq \lang(\bToSd(\baut))$, respectively.
\begin{lemma}
\label{lem:semiDetLangSubseteqBuechiLang}
	Given a NGBA $\baut$, $\lang(\bToSd(\baut)) \subseteq \lang( \baut)$.
\end{lemma}
\begin{proof}
Let $\arun = R_{0} \cdot \ldots \cdot R_{n-1} \cdot (R_{n}, j_{n}, C_{n}) \cdot (R_{n+1}, j_{n+1}, C_{n+1}) \ldots$ be an accepting run of $\bToSd(\baut)$ on a word $\word$.
Thus, for each $i \geq 0$, $R_{i+1} = \sdmatini(R_{i}, \word(i))$ holds;
$(R_{n}, j_{n}, C_{n}) \in \sdmattra(R_{n-1}, \word(n-1))$; and, for each $i \geq n$, $(R_{i+1}, j_{i+1}, C_{i+1}) = \sdmatfin((R_{i}, j_{i}, C_{i}), \word(i))$ holds.

Let $b_{0} < b_{1} < b_{2} \ldots$ be the ascending chain of breakpoints.
That is, the chain that contains exactly the positions $b_{l}$ such that $\arun(b_{l}) = (R_{b_{l}}, j_{b_{l}}, \emptyset)$ for some $R_{b_{l}} \subseteq \astates$ and $j_{b_{l}} \in \natIntK$ that are reached after some accepting transition has been performed.

We build the largest prefix closed tree of initial sequences $q_{0} \cdot q_{1} \cdot \ldots \cdot q_{h} \in \astates^{*}$ of runs, for which the following conditions hold:
\begin{itemize}
\item
    $q_{0} \cdot q_{1} \cdot \ldots \cdot q_{h-1} \in \astates^{*}$ is included in the tree,
\item
    it is an initial sequence of a run of $\baut$,
\item
    $q_{h} \in R_{h}$, and
\item
    for all indices $h = b_{l+1}$ (with $l \in \naturals$) from the chain of breakpoints there has to be a $j_{b_{l}}$ accepting transition $(q_{m}, \word(m), q_{m+1}) \in \afinal_{j_{b_{l}}}$ for some $b_{l} \leq m < h$.
\end{itemize}

As usual with the breakpoint construction, it is easy to show that there exists a run $q_{0} \cdot q_{1} \cdot \ldots \cdot q_{h}$ for all $h \in \omega$ and $q_{h} \in R_{h}$.
Thus, we are left with an infinite and finitely branching tree.
Invoking K\"onig's lemma, this tree contains an infinite path, which is an accepting run of $\baut$ by construction.
\end{proof}

\begin{lemma}
\label{lem:buechiLangSubseteqSemiDetLang}
	Given a NGBA $\baut$, it is $\lang(\baut) \subseteq \lang(\bToSd(\baut))$.
\end{lemma}
Before proving Lemma~\ref{lem:buechiLangSubseteqSemiDetLang}, we introduce some terminology we use in the proof.

Given two sequences $\pi$ and $\pi'$, we write $\pi \trianglelefteq \pi'$ if $\pi$ is a prefix of $\pi'$.

Let $\arun = q_{0} \cdot q_{1} \cdot q_{2} \cdot \ldots$ be a run of $\baut$ on an $\omega$-word $\word \in \Sigma^\omega$.
For this run $\arun$, we denote by $[\arun, n, \word]$ the set of runs on $\word$ with the same initial sequence $\arun_{n} = q_{0} \cdot \ldots \cdot q_{n}$, i.e., $[\arun, n, \word] = \setcond{\arun' \in \runs(\alpha)}{\arun_{n} \trianglelefteq \arun'}$.
Moreover, let $\langle \arun, n, \word \rangle$ be the sequence $R_n \cdot R_{n+1} \cdot R_{n+2} \cdot \ldots$ where $R_{i} = \setcond{q \in \astates}{\exists \arun' \in [\arun, n, \word].\ \arun'(i) = q}$ for all $i \geq n$.
Essentially, the sets $R_{i}$ are the sets from a subset construction that starts with the singleton $R_{n} = \setnocond{q_{n} \in \astates}$ at position $n$.
Note that $R_{n}$ is indeed a singleton since it is the last state of $\arun_{n}$.

With $\langle \arun, n, \word \rangle^{h}$, we define the sequence $C^{h}_{n} \cdot C^{h}_{n+1} \cdot C^{h}_{n+2} \ldots$  where $C^{h}_{i} = \setcond{q \in \astates}{\exists \arun' \in [\arun, n, \word].\ \arun'(i) = q \text{ and } \exists l, n < l \leq i.\ \overline{\arun}'(l) \in \afinal_{h}}$.

Essentially, $(R_{n}, C^{h}_{n})$ are initially the pairs from a breakpoint construction (relative to the accepting set $\afinal_{h}$) that starts with $(R_{n}, C^{h}_{n}) = (\setnocond{q_{n}}, \emptyset)$ at position $n$, but does not reset when a breakpoint is met.
In this case, the sequence continues with sets $R_{m} = C_{m}$ for all positions $m$ from the breakpoint onwards:
$R_{m} = C_{m}$ implies $\amat(C_{m},\word(m)) = \amat(R_{m},\word(m)) \supseteq \afinal_{h}(R_{m}, \word(m))$, and thus $R_{m+1} = \amat(R_{m},\word(m)) =  \amat(C_{m},\word(m)) \cup \afinal_{h}(R_{m}, \word(m)) = C_{m+1}$.

Note that, for each $h \in \natIntK$, each $j \geq n$, $\langle \arun, n, \word \rangle^{h} = C^{h}_{n} \cdot C^{h}_{n+1} \cdot C^{h}_{n+2} \cdot \ldots$, and $\langle \arun, n, \word \rangle = R_{n} \cdot R_{n+1} \cdot R_{n+2} \cdot \ldots$, by construction it follows that $C^{h}_{j} \subseteq R_{j}$.

We first define the width of a position $n$ as $\width(\arun, n, \word) = \max \setcond{|R_{j}|}{j \geq n}$ where $\langle \arun, n, \word \rangle = R_{n} \cdot R_{n+1} \cdot R_{n+2} \cdot \ldots$ and, for each $h \in \natIntK$, $\width_{h}(\arun, n, \word) = \max \setcond{|C^{h}_{j}|}{j \geq n}$ where $\langle \arun, n, \word \rangle^h = C^{h}_{n} \cdot C^{h}_{n+1} \cdot C^{h}_{n+2} \cdot \ldots$.

Note that, for each $h \in \natIntK$ and each $n \in \omega$, we have the following relations between the widths:
$\width_{h}(\arun, n+1, \word) \leq \width_{h}(\arun, n, \word) \leq \width(\arun, n, \word)$ and $\width(\arun, n+1, \word) \leq \width(\arun, n, \word)$.
In fact, by definition of $\width$, we have that $\width_{h}(\arun, n, \word) = \max \setcond{|C_{j}|}{j \geq n} = \max \setnocond{|C_{n}|, \max \setcond{|C_{j}|}{j \geq n+1}} = \max \setnocond{|C_{n}|, \width_{h}(\arun, n+1, \word)} \geq \width_{h}(\arun, n+1, \word)$, and similarly for $\width(\arun, n+1, \word) \leq \width(\arun, n, \word)$.
Since $C^{h}_{j} \subseteq R_{j}$ holds for each $h \in \natIntK$ and $j \geq n$, it is immediate to derive $\width_{h}(\arun, n, \word) \leq \width(\arun, n, \word)$.

Since the widths are monotone, we can therefore define the width of a run as the limit of the width of its positions, i.e., we define $\width(\arun, \word)= \lim_{n \to \infty} \width(\arun, n, \word)$ and $\width_{h}(\arun, \word)= \lim_{n \to \infty} \width_{h}(\arun, n, \word)$.
We are now ready to prove Lemma~\ref{lem:buechiLangSubseteqSemiDetLang}:

\begin{proof}
Let $\arun = q_{0} \cdot q_{1} \cdot q_{2} \cdot \ldots$ be an accepting run of $\baut$.
We first show that $\width(\arun, \word) = \width_{h}(\arun, \word)$ holds for each $h \in \natIntK$.
For doing this, we select an $n \in \omega$ such that $\width(\arun, n, \word) = \width(\arun, \word)$ holds.
Note that such an $n$ exists due to the monotonicity of $\width(\arun, n, \word)$ in $n$ and the fact that $\width(\arun, n, \word)$ has $0$ as lower bound.
Moreover, $\width(\arun, m, \word) = \width(\arun, \word)$ holds for each $m \geq n$.

For a given $h$, we can now choose an arbitrary $m_{h} > n$ such that $\tran{\rho}(m_{h}) \in \afinal_{h}$. (As $\arun$ is accepting, arbitrarily large such $m_{h}$ exists.)
We have $\langle \arun, m_{h}-1, \word \rangle^h = \emptyset \cdot C^{h}_{m_{h}} \cdot C^{h}_{m_{h}+1} \cdot C^{h}_{m_{h}+2} \cdot \ldots$ and $\langle \arun, m_{h}, \word \rangle = R_{m_{h}} \cdot R_{m_{h}+1} \cdot R_{m_{h}+2} \cdot \ldots$.
Now, $\tran{\rho}(m_{h}) \in \afinal_{h}$ implies $q_{m_{h}} \in C^{h}_{m_{h}}$, which together with $\setnocond{q_{m_{h}}} = R_{m_{h}}$ provides $C^{h}_{m_{h}} \supseteq R_{m_{h}}$.
A simple inductive argument thus implies $C^{h}_{j} \supseteq R_{j}$ for all $j \geq m_{h}$, and thus
$\width_{h}(\arun, m_{h}-1, \word) \geq \width(\arun, m_{h}, \word)=\width(\arun, \word)$.
Together with $\width(\arun, \word) = \width(\arun, m_{h}-1, \word) \geq \width_{h}(\arun, m_{h}-1, \word)$, this provides $\width_{h}(\arun, m_{h}-1, \word) = \width(\arun, \word)$.
As $m_{h}$ can be chosen arbitrarily large, this implies $\width_{h}(\arun, \word) =\width(\arun, \word)$.

With this observation, we can construct an accepting run of $\bToSd(\baut)$ as follows.
We start with an initial sequence $R_{0} \cdot R_{1} \cdot \ldots \cdot R_{n-1}$ in $\sdstatesini$ where by definition $R_{0} = \ainit \in \sdinit$.
Note that this sequence is well defined and deterministic; moreover, for each $0 \leq i < n$, $q_{i} \in R_{i}$ holds.
Since $q_{n} \in \amat(q_{n-1}, \word(n-1))$, we have that $q_{n} \in \sdmatini(R_{n-1}, \word(n-1))$, thus $(R_{n-1}, \word(n-1),(\setnocond{q_{n}}, 1, \emptyset)) \in \sdmattra$ and we use such transition to extend the sequence $R_{0} \cdot R_{1} \cdot \ldots \cdot R_{n-1}$ to $R_{0} \cdot R_{1} \cdot \ldots \cdot R_{n-1} \cdot (R_{n}, j_{n}, C_{n})$ where $(R_{n}, j_{n}, C_{n}) = (\setnocond{q_{n}}, 1, \emptyset)$.
Again, $q_{n} \in R_{n}$.
Note that the choice of using the accepting set $\afinal_{1}$ is arbitrary.
The remainder of the run, $(R_{n+1}, j_{n+1}, C_{n+1}) \cdot (R_{n+2}, j_{n+2}, C_{n+2}) \ldots$ is well defined, as this second part is again deterministic, and still $q_{i} \in R_{i}$ for each $i > n$.
Since $q_{i} \in R_{i}$ for each $i \geq n$, the deterministic automaton in the second part does not block.
Moreover, $\langle \arun, n, \word \rangle = R_{n} \cdot R_{n+1} \cdot R_{n+2} \cdot \ldots$ holds by a simple inductive argument.

We assume for contradiction that this constructed run $\arun' = R_{0} \cdot R_{1} \cdot \ldots \cdot R_{n-1} \cdot (R_{n}, j_{n}, C_{n}) \cdot (R_{n+1}, j_{n+1}, C_{n+1}) \cdot (R_{n+2}, j_{n+2}, C_{n+2}) \cdot \ldots$ is not accepting.
Then there exists an $m > n$ such that, for all $i \geq m$, $\tran{\rho'}(i)$ is not in $\sdfinal$.

Let $\langle \arun, m, \word \rangle^{j_{m}} = C^{j_{m}}_{m} \cdot C^{j_{m}}_{m+1} \cdot C^{j_{m}}_{m+2} \cdot \ldots$.
A simple inductive argument provides that $C^{j_{m}}_{i} \subseteq C_{i} \subsetneq R_{i}$ holds for all $i \geq m$.
Thus, there is a $i \geq m$ with $|C^{j_{m}}_{i}| = \width^{j_{m}}(\arun, m, \word)$.
But since $\width^{j_{m}}(\arun, m, \word) \geq \width^{j_{m}}(\arun, \word) = \width(\arun, \word) = \width(\arun, n, \word)$, we have that $|R_{i}| \leq |C^{j_{m}}_{i}|$, which contradicts $C^{j_{m}}_{i} \subsetneq R_{i}$.
\end{proof}

\subsection{Proof of Proposition~\ref{pro:semiDetLanguageIgnoreBandi}}
\begin{proof}
We show that given $(R, j, C), (R, j', C') \in \sdstatesfin$ of $\bToSd(\baut)$ it is $\lang(\bToSd(\baut)_{(R, j, C)}) \subseteq \lang(\bToSd(\baut)_{(R, j', C')})$;
symmetry then provides us with equivalence.

Let us assume for contradiction that $\arun = (R_{0}, j_{0}, C_{0}) \cdot (R_{1}, j_{1}, C_{1}) \cdot \ldots$ is the accepting run of $\bToSd(\baut)_{(R, j, C)}$ while $\arun' = (R'_{0}, j_{0}, C_{0})\cdot (R'_{1}, j'_{1}, C'_{1}) \cdot \ldots$ is the rejecting run of $\bToSd(\baut)_{(R, j', C')}$ on an input word $\alpha$.
Note that $R'_{0} = R_{0}$ since $R_{0} = R = R'_{0}$;
similarly, $C_{0} = C$ and $j_{0} = j$ as well as $C'_{0} = C'$ and $j'_{0} = j'$.

We can first establish with a simple inductive argument that $R_{l} = R'_{l}$ for all $l \in \omega$ (and that $\bToSd(\baut)_{(R, j', C')}$ has a run on $\word$).
In fact, for $l = 0$, we have already noted that $R_{l} = R_{0} = R = R'_{0} = R'_{l}$;
suppose that $R_{l} = R'_{l}$;
by construction of $\bToSd(\baut)$, it follows that $R_{l+1} = \sdmatini(R_{l}, \word(l)) = \sdmatini(R'_{l}, \word(l)) = R'_{l+1}$.
Since $\arun$ is accepting, $R_{l+1} \neq \emptyset$ for each $l \in \omega$ and this implies that $\big((R_{l}, j'_{l}, C'_{l}), \word(l), (R_{l+1}, j'_{l+1}, C'_{l+1})\big) \in \sdmatfin \cup \sdfinal$, i.e., $\bToSd(\baut)_{(R, j', C')}$ has a run on $\word$.

As $\arun'$ is rejecting, there exists an $n \in \omega$ such that, for each $l \geq n$,
we have that $\big((R_{l}, j'_{l}, C'_{l}), \word(l), (R_{l+1}, j'_{l+1}, C'_{l+1})\big) \notin \sdfinal$, otherwise $\arun'$ would be accepting;
note that since we have that $\big((R_{l}, j'_{l}, C'_{l}), \word(l), (R_{l+1}, j'_{l+1}, C'_{l+1})\big) \in \sdmatfin \cup \sdfinal$, this implies that $\big((R_{l}, j'_{l}, C'_{l}), \word(l), (R_{l+1}, j'_{l+1}, C'_{l+1})\big) \in \sdmatfin$.
By definition of $\sdmatfin$, it follows that $C'_{l+1} \subsetneq r'_{l+1}$ and $j'_{l+1} = j'_{l}$, thus for each $l \geq n$ we have $j'_{l} = j'_{n}$.

As $\arun$ is accepting, transitions from $\sdfinal$ are taken infinitely often, thus the indices $j \in \natIntK$ are visited cyclically.
In particular, since $j'_{n} \in \natIntK$, this implies that there exists an $l > n$ such that $j_{l} = j'_{n}$, as effect of the transition $\big((R_{l-1}, C_{l-1}, j_{l-1}), \word(l-1), (R_{l}, \emptyset, j'_{n})\big) \in \sdfinal$.
However, since $\arun'$ is rejecting by assumption, we know that $C'_{m} \subsetneq R_{m}$.
Moreover
\begin{enumerate}[1.)]
\item
	by definition of $\sdaut$, the breakpoints sets are reset to $\emptyset$ after an accepting transition.
\item
	For non-accepting transitions, the breakpoint construction is monotonic in the sense that for each $A$, $B_{1}$, $B_2$, $i$, and $\sigma$ such that $B_{1} \subseteq B_2 \subsetneq A$ and $(A',i,B_2') = \sdmatfin((A,i,B_2), \sigma)$, it follows that $\sdmatfin((A,i,B_{1}), \sigma)$ is defined and, for $(A',i,B_{1}') = \sdmatfin((A,i,B_{1}), \sigma)$, we have $B_{1}' \subseteq B_2'$.
\item
	As $\arun'$ is rejecting, its breakpoint sets are never reset after position $n$ (and remain to be $j_l$).
\end{enumerate}

From 1.)-3.), it follows by induction that $C_{m} \subseteq C'_{m}$, thus $C_{m} \subseteq C'_{m} \subsetneq R_{m}$ for each $m > l$.
This, however, implies due to the definition of $\sdfinal$ that in $\arun$ no further accepting transitions follow, thus $j_{m} = j_l = j'_{n}$.
This implies that $\arun$ is rejecting, contradicting the initial assumption.
\end{proof}

\subsection{Proof of Proposition~\ref{pro:semiDetLangEqParityLang}}
As for Proposition~\ref{pro:buechiLangEqualSemiDetLang}, we split the proof of Proposition~\ref{pro:semiDetLangEqParityLang} into two lemmas, stating that for each $q \in \q \in Q'$, $\lang(\daut_{\q}) \subseteq \lang(\sdaut_q)$ and $\lang(\sdaut_q) \subseteq \lang(\daut_{\q})$ hold, respectively.
As notation, for states $q \in \sdstatesini$, $q' \in \sdstatesfin$, and $\q = (r,f) \in \astates'$, we write $q \in \q$ if $q = r$, and $q' \in \q$ if there exists $j \in \omega$ with $f(j) = q'$.

By a trivial inductive proof, we get the following observation.

\begin{lemma}
\label{lem:prerun}
	Given a semi-deterministic \buchi automaton $\sdaut$ and $\daut = \sdToDet(\sdaut)$, for each $q \in \sdstates$, each $\q \in Q'$ with $q \in \q$, and each input word $\word$, there is a pre-run $q \cdot q_{1} \cdot \ldots \cdot q_{n}$ of $\sdaut_{q}$ if, and only if, there is a pre-run $\q \cdot \q_{1} \cdot \ldots \cdot \q_{n}$ of $\daut_{\q}$ with $q_{n} \in \q_{n}$.
\end{lemma}

\begin{lemma}
\label{lem:parityLangSubseteqSemiDetLang}
	Given a semi-deterministic \buchi automaton $\sdaut$ and $\daut = \sdToDet(\sdaut)$, for each $q \in \q \in Q'$, $\lang(\daut_{\q}) \subseteq \lang(\sdaut_q)$ holds.
\end{lemma}
\begin{proof}
We first show $\lang(\daut) \subseteq \lang(\sdaut)$.

Let $\arun = \q \cdot \q_{1} \cdot \q_{2} \cdot \ldots$ be an accepting run of $\daut_{\q}$ on a word $\word$ with dominating priority $2a$;
let $n \geq 1$ be a natural number such that $\pri(\q_{l}, \word(l), \q_{l+1}) \geq 2a$ holds for all $l \geq n$;
and let $\q_{l} = (r_{l}, f_{l})$ for all $l \geq n$.

By Lemma~\ref{lem:prerun}, there is a pre-run $q \cdot q_{1} \cdot q_{2} \cdot \ldots \cdot q_{n}$ with $q_{n} = f(a)$ of $\sdaut_{q}$.
The observation that no priority less than $2a$ occurs from the $n$-th transition onwards in $\arun$ provides with the construction of $\daut$ that this pre-run can be continued to a unique run
$\arun' = q \cdot q_{1} \cdot q_{2} \cdot \ldots \cdot q_{n} \cdot q_{n+1} \ldots$ with $q_{l} = f_{l}(a)$ for all $l \geq n$.
Further, for all $l \geq n$ with $\pri(\q_{l}, \word(l), \q_{l+1}) = 2a$, we have $(q_{l}, \word(l), q_{l+1}) \in \sdfinal$.
As there are infinitely many such $l$, $\arun'$ is accepting.
\end{proof}

\begin{lemma}
\label{lem:semiDetLangSubseteqParityLang}
	Given a semi-deterministic \buchi automaton $\sdaut$ and $\daut = \sdToDet(\sdaut)$, for each $q \in \q$, $\lang(\sdaut_{q}) \subseteq \lang(\daut_{\q})$ holds.
\end{lemma}
\begin{proof}
Let $q_{0} \cdot q_{1} \cdot q_{2} \cdot \ldots$ with $q_{0} = q$ be an accepting run of  $\sdaut_{q}$ on an input word $\alpha$.
Then there is a minimal $n \in \omega$ such that $q_{n} \in \sdstatesfin$---and thus $q_{l} \in \sdstatesfin$ for all $l \geq n$ and $q_{l} \in \sdstatesini$ for all $l < n$.

By a simple inductive argument we can show that $\daut_{\q}$ has a run $\arun = (q_{0}, f_{0}) \cdot (q_{1}, f_{1}) \cdot \ldots \cdot (q_{n-1}, f_{n-1}) \cdot (r_{n}, f_{n}) \cdot (r_{n+1}, f_{n+1}) \ldots$ on $\word$, such that $q_{l} \in (r_{l}, f_{l})$ for all $l \geq n$.
Moreover, there is a descending chain $i_{n} \geq i_{n+1} \geq i_{n+2} \ldots$ of indices such that $q_{l} = f_{l}(i_{l})$.
This chain stabilises at some point to $a = i_{n'} = \lim_{l \to \infty} i_{l}$.
Consequently, we have that $\pri \big((r_{l}, f_{l}), \word(l), (r_{l+1}, f_{l+1})\big)$ is even \emph{or} no smaller than $2a$ for all $l \geq n'$.
(Assuming that the priority is an odd number less than $2a$ would imply that there is a $\blank$ sign in $g'$ at a position less than or equal to $a$, which would contradict that the index has stabilised.)
For all positions $l \geq n'$ with $\big(q'_{l}, \word(l), q'_{l+1}\big) \in \afinal$, $\pri \big((r_{l}, f_{l}), \word(l), (r_{l+1}, f_{l+1})\big)$ is an even number less than or equal to $2a$.
The smallest priority occurring infinitely often in the transitions of $\arun$ is therefore an even number less than or equal to $2a$.
\end{proof}

\subsection{Proof of Lemma~\ref{lem:parityOfSccAndAcceptingSccOfSemidet}}

\begin{proof}
To prove the lemma, we first note that no transition of any run of the product (which is an SCC) can see 
a priority smaller than $2a$.
Thus, for all such runs $(m_{0}, (r_{0}, f_{0})) \cdot (m_{1}, (r_{1}, f_{1})) \cdot (m_{2}, (r_{2}, f_{2})) \ldots$, the sequence $(m_{0}, f_{0}(a)) \cdot (m_{1}, f_{1}(a)) \cdot (m_{2}, f_{2}(a)) \ldots$ is a run, and a transition like $\big((m_{j}, f_{j}(a)), \word(j), (m_{j+1}, f_{j+1}(a))\big)$ is accepting if, and only if, the priority is minimum and even;
more precisely, we have that $\pri\big((m_{j}, (r_{j}, f_{j})), \word(j), (m_{j+1}, (r_{j+1}, f_{j+1}))\big) = 2a$.

It is then easy to see that the measure of the accepting paths of $\mc \prodAut \sdaut_{f_{0}(a)}$ equals the measure of the paths of $\mc \prodAut \daut_{(r_{0}, f_{0})}$ with dominating priority $2a$, which is $1$.
\end{proof}

\subsection{Proof of Theorem~\ref{thm:acceptanceOfSubsetAutomaton}}
\begin{proof}
\begin{description}
\item[$1) \implies 2)$]
	Let $\scc$ be an accepting bottom SCC of $\mc \prodMC \ssaut$, $\arun$ be a run of $\mc \prodMC \ssaut$ trapped into $\scc$, and $\word$ be the associated word.
	By Definition~\ref{def:acceptingSccOfSubset}, it follows that there exists an SCC $\scc'$ of $[\mc \prodAut \daut]$ that is isomorphic to $\scc$ containing only accepting states, where $\daut = \det(\baut)$.
	By a simple inductive argument, we can show that the run $\arun'$ of $\mc \prodAut \daut$ on $\word$ satisfies $\arun(i) = [\arun'(i)]$ for each $i \in \omega$ and that $\arun'$ is trapped into an $\scc''$ of $\mc \prodAut \daut$ with $[\scc''] \subseteq \scc'$ as well.
	Since $\scc'$ is accepting, by Corollary~\ref{cor:quotients} $\scc''$ contains only accepting states as well; this means that there exists $i \in \omega$ such that $\arun'(i) = (m,(r,f)) \in \scc''$ and a transition $((m,(r,f)),\word(i),(m',(r',f')))$ such that $\pri((r,f),\word(i),(r',f'))$ is even and minimum, say $2a$.
	By Lemma~\ref{lem:parityOfSccAndAcceptingSccOfSemidet} it follows that $\mc \prodAut \sdaut_{f(a)}$ forms an SCC that is accepting.
	By construction of $\daut$, $f(a) = (R,l,C)$  and $f'(a) = (R', l \oplus_{k} 1, \emptyset)$ where $R \subseteq r$ and $R' \subseteq r'$, thus we have $(m',r') \in \scc$ and $R' \subseteq r'$ such that $(m',(R', l \oplus_{k} 1, \emptyset))$ belongs to an accepting SCC of $\mc \prodAut \sdaut_{(m',(R', l \oplus_{k} 1, \emptyset))}$, as required.

\item[$2) \implies 1)$]
	Let $\scc$ be a bottom SCC of $\mc \prodMC \ssaut$ such that there exist $(m,R) \in \scc$ and $R' \subseteq R$ such that $(m,(R',j,\emptyset))$ belongs to an accepting SCC $\scc'$ of $\mc \prodAut \sdaut_{(m,(R',j,\emptyset))}$ for some $j \in \natIntK$.
	Since $\scc'$ is accepting, by construction of $\daut$, Proposition~\ref{pro:semiDetLangEqParityLang}, and Lemma~\ref{lem:parityOfSccAndAcceptingSccOfSemidet}, it follows that there exists $f$ and integer $a$ such that $f(a) = (R,l,C)$ and $\mc \prodAut \daut_{(m,(R,f))}$ is accepting.
	If $(m,(R,f))$ is already in an accepting SCC $\scc''$ of $\mc \prodAut \daut_{(m,(R,f))}$, then by definition of accepting SCC for $\mc \prodMC \ssaut$, by Theorem~\ref{thm:bottom_component} and Corollary~\ref{cor:quotients} we have that $\scc = [\scc'']$ is accepting as well.
	Since $\mc \prodAut \daut_{(m,(R,f))}$ is accepting, let $\word$ one of the accepted words and $\arun$ the resulting run starting from $(m,(R,f))$;
	eventually $\arun$ is trapped into an accepting SCC $\scc''$ with $[\scc''] = \scc$.
	Let $i \in \omega$ such $\arun(i) = (m',(r',f')) \in \scc''$ and $\pri(\tran{\arun}(i))$ is even and minimum, say $2a$.
	This implies by construction of $\daut$ that $f''(a) = (r'',l,\emptyset)$ where $\arun(i+1) = (m'',(r'',f'')) \in \scc''$ for some $m'' \in \pstates$.
	Since $(m'',(r'',f'')) \in \scc''$ and $\scc''$ is accepting, by Corollary~\ref{cor:quotients} it follows that all states in $\scc''$ are accepting, thus by Definition~\ref{def:acceptingSccOfSubsetApp}, $\scc$ is accepting as well.

\item[$1) \Longleftrightarrow 3)$]
	This equivalence follows directly from a combination of the proofs of Propositions~\ref{pro:buechiLangEqualSemiDetLang} and Proposition~\ref{pro:semiDetLangEqParityLang};
	in particular, the proof of Lemma~\ref{lem:buechiLangSubseteqSemiDetLang} (stating that $\lang(\baut) \subseteq \lang(\sdaut)$) provides the singleton $\setnocond{q}$ needed for the implication $1) \implies 3)$.
	\qedhere
\end{description}
\end{proof}

\end{document}